\documentclass[acmtocl]{acmtrans2m}
\usepackage{amssymb}
\usepackage{amsmath}
\usepackage{graphicx}
\usepackage{wrapfig}
\usepackage{algorithm2e}

\newtheorem{theorem}{Theorem}[section]

\newtheorem{corollary}[theorem]{Corollary}

\newtheorem{lemma}[theorem]{Lemma}
\newdef{definition}[theorem]{Definition}
\newdef{remark}[theorem]{Remark}
\newdef{example}[theorem]{Example}

\usepackage[all]{xypic}

\newcommand{\lrule}[3]{\label{#1} #2 & \to #3} 
\newcommand{\rrule}[3]{#2 & \to #3 \tag{\ref{#1}}} 
\newcommand{\urule}[2]{#1 & \to #2} 
\newcommand{\rLem}[1]{Lemma~\ref{#1}}
\newcommand{\rDef}[1]{Definition~\ref{#1}}
\newcommand{\rSec}[1]{Section~\ref{#1}}
\newcommand{\rThm}[1]{Theorem~\ref{#1}}
\newcommand{\rEx}[1]{Example~\ref{#1}}
\newcommand{\rAlg}[1]{Algorithm~\ref{#1}}
\newcommand{\Unlab}{\mathit{Unlab}}
\renewcommand{\P}{\mathcal{P}}
\newcommand{\R}{\mathcal{R}}

\newcommand{\D}{\mathcal{D}}
\newcommand{\T}{\mathcal{T}}
\newcommand{\V}{\mathcal{V}}
\newcommand{\A}{\mathcal{A}}

\newcommand{\Ff}{\mathsf{f}}
\newcommand{\Fz}{\mathsf{0}}
\newcommand{\Fg}{\mathsf{g}}
\newcommand{\Fs}{\mathsf{s}}
\newcommand{\Fp}{\mathsf{p}}
\newcommand{\Fu}{\mathsf{u}}
\newcommand{\FU}{\mathsf{U}}
\newcommand{\Frotate}{\mathsf{rotate}}
\newcommand{\Fstol}{\mathsf{s2\ell}}
\newcommand{\mIn}{\mbox{\sl \textbf{in}}}
\newcommand{\mOut}{\mbox{\sl \textbf{out}}}

\newcommand{\Fq}{\mathsf{q}}

\newcommand{\Fappend}{\mathsf{append}}
\newcommand{\Fcons}{\mathsf{cons}}
\newcommand{\sorsuccsim}{{\,{}_{_(}\!\!\succsim_{_)}}}
\newcommand{\Proc}{\mathit{Proc}}
\newcommand{\CAP}{\mathit{CAP}}
\newcommand{\rt}{\mathrm{root}}
\newcommand{\pin}{p_{in}}

\renewcommand{\to}{\rightarrow}
\newcommand{\from}{\mbox{\rm $\,$:--$\;$}}
\newcommand{\N}{\mathbb{N}}
\newcommand{\el}{\ell}

\newcommand{\email}[1]{{\tt #1}}

\newcommand{\point}{\raisebox{-0.25ex}[0cm]{\ensuremath{\bullet}}}

\title{Automated Termination Proofs for\\Logic Programs by Term
Rewriting} 

\author{
Peter Schneider-Kamp, RWTH Aachen, Germany\\
J\"{u}rgen Giesl, RWTH Aachen, Germany\\
Alexander Serebrenik, TU Eindhoven, The Netherlands\\
Ren\'{e} Thiemann, University of Innsbruck, Austria
}

\begin{abstract}
There
are two kinds of approaches for termination analysis of logic programs:
``transformational'' and ``direct''
ones. Direct approaches prove termination directly on the basis of the logic
program. Transformational approaches
transform a logic program into a
term rewrite system (TRS) and then analyze termination of
the resulting TRS instead. Thus, transformational approaches make all methods
previously developed for TRSs available for logic programs as well. However,
the applicability of most existing transformations is quite restricted, as
they can only be used for certain subclasses of logic programs. (Most of them
are restricted to \emph{well-moded} programs.)
In this paper we improve these transformations such that they become applicable
for \emph{any} definite logic program. To si\-mu\-late the behavior of logic
programs by TRSs, we slightly modify the notion of rewriting by permitting
infinite terms. We show that our transformation
results in TRSs which are indeed
suitable for \emph{automated} termination analysis.
In contrast to most other methods for termination of
logic programs, our technique is also sound for logic programming
\emph{without occur check}, which is typically used in
practice. We implemented
our approach
in the termination prover {\sf AProVE} and successfully evaluated it on a
large collection of examples.
\end{abstract}

\category{F.3.1}{Logic and Meanings of Programs}{Specifying and Verifying and
Reasoning about Programs}[Mechanical verification]
\category{D.1.6}{Programming Techniques}{Logic Programming}
\category{I.2.2}{Artificial Intelligence}{Automatic
Programming}[Automatic analysis of algorithms]
\category{I.2.3}{Artificial Intelligence}{Deduction and Theorem Proving}[Logic
programming]

\terms{Languages, Theory, Verification}

\keywords{logic programming, termination analysis, term rewriting, dependency pairs}

\begin{document}
\begin{bottomstuff}Supported by the Deutsche
Forschungsgsmeinschaft DFG under grant GI 274/5-2.\\

Authors' addresses:\\
Peter Schneider-Kamp, J\"{u}rgen Giesl, 
LuFG Informatik 2, RWTH Aachen, Ahornstr.\ 55, 52074 Aachen, Germany,
\email{\{psk,giesl\}@informatik.rwth-aachen.de}\\
Alexander Serebrenik,
Dept.\ of Mathematics and Computer Science, TU Eindhoven, P.O.\ Box 513, 5600
MB Eindhoven, The Netherlands,
\email{a.serebrenik@tue.nl}\\
Ren\'{e} Thiemann, Institute of Computer Science, University of Innsbruck,
Technikerstr.\ 21a,
6020 Innsbruck, Austria,
\email{rene.thiemann@uibk.ac.at}. Most of the work was carried out while R.\
Thiemann was at the LuFG Informatik 2, RWTH Aachen.
\end{bottomstuff}

\maketitle

\section{Introduction}\label{sec:Introduction}

Termination of logic programs is widely studied. Most automated techniques try to prove
\emph{universal termination} of definite logic programs, 
i.e., one tries to show that all derivations
of a logic program are finite w.r.t.\ the left-to-right selection
rule.

Both ``direct'' and ``transformational'' approaches have been proposed in
the literature (see,
e.g.,~\cite{DeSchreye:Decorte:NeverEndingStory} for an
overview
and~\cite{CodishTOPLAS,Codish:Lagoon:Stuckey,Codish:Lagoon:Stuckey:ESOP06,%
DeSchreye:Serebrenik:Kowalski,LMS03,Mesnard:Ruggieri,%
Mesnard:Serebrenik,Nguyen:DeSchreye,Nguyen:DeSchreye06,LOPSTR07,%
SerebrenikS05,Smaus}
for more recent work on ``direct'' approaches). ``Transformational''
approaches have been developed in
\cite{Aguzzi:Modigliani,Arts:Zantema:95,Chtourou:Rusinowitch,Ganzinger:Waldmann,%
Rao,Marchiori:ALP94,Marchiori:AMAST96,Raamsdonk}
and a comparison of these approaches is given in \cite{Ohlebusch01}.
Moreover, similar transformational approaches also exist for other programming
languages (e.g., see \cite{RTA06} for an approach to prove termination of {\sf
Haskell}-programs via a transformation to term rewriting). Moreover, there is
also work in progress to develop such approaches for imperative programs.

In order to be successful for termination analysis of logic programs, transformational methods

\begin{tabbing}
(II) \=\kill
(I)\>should be \emph{applicable} for a class of logic programs as large
as possible and\\
(II)\>should produce TRSs whose termination is \emph{easy to analyze
automatically}.
\end{tabbing}

\noindent
Concerning (I), the above existing transformations
can only be used for certain subclasses of logic
programs.
More precisely,
all approaches except \cite{Marchiori:ALP94,Marchiori:AMAST96} are restricted to \emph{well-moded}
programs. The transformations of \cite{Marchiori:ALP94,Marchiori:AMAST96} also
consider the classes of
\emph{simply
well-typed} and \emph{safely typed} programs.
However in contrast to all previous
transformations,
we present
a new transformation which is applicable for \emph{any}  (definite) logic program.
Like most
approaches for termination of logic programs, we restrict
ourselves to programs without cut and negation.
While there are transformational approaches which go beyond definite
programs \cite{Marchiori:AMAST96}, it is not clear how to transform non-definite logic programs
into TRSs that are suitable for \emph{automated} termination analysis, cf.\ (II).

Concerning (II),  one needs an implementation and an empirical
evaluation to find out  whether termination of the transformed TRSs can indeed
be verified automatically for a large class of examples. Unfortunately, to our
knowledge there is only a single other termination tool available which
implements a transformational approach. This tool  \textsf{TALP} \cite{TALP}
is based on the
transformations of \cite{Arts:Zantema:95,Chtourou:Rusinowitch,Ganzinger:Waldmann}
which are shown to be equally powerful in
\cite{Ohlebusch01}. So these transformations are indeed
suitable for automated termination analysis, but
consequently,  \textsf{TALP} only accepts well-moded
logic programs. This is in contrast to our approach which we
implemented  in our termination prover
{\sf AProVE} \cite{IJCAR06}.  Our
experiments on large collections of examples in
\rSec{sec:experiments} show that our transformation indeed produces TRSs
that are suitable for automated termination analysis and that
{\sf AProVE} is currently among the most powerful termination provers for logic
programs.

To illustrate the starting point for our research, we briefly review
related work on connecting termination analysis of logic programs and term
rewrite systems: in
\rSec{The Classical Transformation}
we recapitulate the
classical
transformation of  \cite{Arts:Zantema:95,Chtourou:Rusinowitch,Ganzinger:Waldmann,Ohlebusch01}
and  in \rSec{Term
Rewriting Techniques for Termination of LPs} we discuss the
approach of adapting TRS-techniques to the logic programming setting (which can be seen as an
alternative to our approach of transforming logic programs to TRSs). Then in
\rSec{Structure of the Paper} we give an overview on the structure of the
remainder of the paper.

\subsection{The Classical Transformation}\label{The Classical Transformation}

Our transformation is inspired by the
transformation of
\cite{Arts:Zantema:95,Chtourou:Rusinowitch,Ganzinger:Waldmann,Ohlebusch01}. In
this classical transformation,  each argument position of
each predicate is either determined to be an {\em input} or an {\em
output} position by a \emph{moding function} $m$. So for every predicate symbol $p$
of arity $n$ and every $1 \leq i \leq n$, we have $m(p,i) \in \{\mIn,
\mOut\}$. Thus, $m(p,i)$ states whether the $i$-th 
argument of $p$ is an input (\mIn) or an output (\mOut)
argument.

As mentioned, the moding must be such that the logic program is  {\em
well moded}~\cite{Apt:Etalle}.
Well-modedness guarantees
that each atom selected by the left-to-right selection rule is ``sufficiently''
instantiated during any derivation with a query that is ground on all input 
positions. More precisely, a program is well moded iff
for any of its clauses $H \from B_1,\ldots,B_k$ with $k \geq 0$, we have
\begin{itemize}
\item[(a)] $\V_{out}(H) \subseteq \V_{in}(H) \cup \V_{out}(B_1) \cup \ldots \cup \V_{out}(B_k)$
and
\item[(b)] $\V_{in}(B_i) \subseteq \V_{in}(H) \cup \V_{out}(B_1) \cup \ldots \cup \V_{out}(B_{i-1})$
for all $1 \leq i \leq k$
\end{itemize}
$\V_{in}(B)$
and $\V_{out}(B)$ are the variables in terms
on $B$'s input and output positions.

\begin{example}
\label{ex_not_sm}
{\sl Consider the following variant of a small example
from
\cite{Ohlebusch01}.
\[\begin{array}{lll}
{\sf p}(X, X). & &\\
{\sf p}({\sf f}(X), {\sf g}(Y)) &\from& {\sf p}({\sf f}(X), {\sf f}(Z)),
{\sf p}(Z,{\sf g}(Y)).
\end{array}\]
\noindent
Let $m$ be a moding with $m({\sf p},1) = \mIn$ and $m({\sf p},2) = \mOut$.
Then the program is well moded:  This is obvious for the first clause.
 For the second clause, (a) holds since the
output variable $Y$~of the head is also an output variable of the second body atom.
Similarly, (b) holds since
the input variable $X$~of the first body atom is also
an input variable of the head, and the input variable
$Z$~of the second body atom is also an output variable of the first body atom.}
\end{example}

\noindent
\hspace*{.2cm} In the classical transformation from logic programs to TRSs,
two new function symbols
$p_{in}$ and $p_{out}$ are introduced for each predicate $p$.
We write ``$p(\vec{s},\vec{t})$'' to denote that $\vec{s}$
and $\vec{t}$ are
the sequences of terms on $p$'s in- and output positions.
\begin{itemize}
\item[$\bullet$] For each fact
$p(\vec{s},\vec{t})$, the TRS contains the rule
$p_{in}(\vec{s}) \to p_{out}(\vec{t})$.
\item[$\bullet$] For each clause $c$ of the form
$p(\vec{s},\vec{t}) \from p_1(\vec{s}_1,\vec{t}_1), \ldots, p_k(\vec{s}_k,\vec{t}_k)$,
the resulting TRS contains the following
rules:
\begin{align*}
\urule{p_{in}(\vec{s})}{u_{c,1}(p_{1_{in}}(\vec{s}_1), \V(\vec{s}))}\\
\urule{u_{c,1}(p_{1_{out}}(\vec{t}_1), \V(\vec{s}))}{
u_{c,2}(p_{2_{in}}(\vec{s}_2), \V(\vec{s}) \cup \V(\vec{t}_1))}\\
& \ldots \\
\urule{u_{c,k}(p_{k_{out}}(\vec{t}_k), \V(\vec{s}) \cup \V(\vec{t}_1) \cup  \ldots \cup \V(\vec{t}_{k-1}))}{  p_{out}(\vec{t})}
\end{align*}
Here,  $\V(\vec{s})$ are the variables occurring in $\vec{s}$.
Moreover, if $\V(\vec{s}) =
\{x_1, \ldots, x_n\}$, then ``$u_{c,1}(p_{1_{in}}(\vec{s}_1), \V(\vec{s}))$''
abbreviates the term $u_{c,1}(p_{1_{in}}(\vec{s}_1), x_1, \ldots, x_n)$,
etc.
\end{itemize}
If  the resulting TRS is terminating, then the original logic program
terminates for any query with ground  terms on all input positions of the
predicates, cf.\ \cite{Ohlebusch01}.
However, the converse does not hold.

\begin{example}
\label{old transformation}
{\sl
For the program of \rEx{ex_not_sm},  the transformation results in the
following TRS $\R$.
\begin{align*}
\urule{{\sf p}_{in}(X)}{{\sf p}_{out}(X)} \\
\urule{{\sf p}_{in}({\sf f}(X))}{{\sf u}_1({\sf p}_{in}({\sf f}(X)),X)} \\
\urule{{\sf u}_1({\sf p}_{out}({\sf f}(Z)),X)}{{\sf u}_2({\sf p}_{in}(Z),X,Z)}\\
\urule{{\sf u}_2({\sf p}_{out}({\sf g}(Y)),X,Z)}{{\sf p}_{out}({\sf g}(Y))}
\end{align*}
The original logic program is terminating for any
query ${\sf p}(t_1,t_2)$ where $t_1$ is a ground term. However, the above TRS is not terminating:
\[{\sf p}_{in}({\sf f}(X)) \to_\R {\sf u}_1({\sf p}_{in}({\sf f}(X)),X) \to_\R {\sf u}_1({\sf
  u}_1({\sf p}_{in}({\sf f}(X)),X),X) \to_\R \ldots\]
In the logic program, after resolving with the second clause,
one obtains a query starting with ${\sf p}({\sf f}(\ldots), {\sf f}(\ldots))$.
Since  ${\sf p}$'s output argument  ${\sf f}(\ldots)$ is
already partly instantiated, the second clause  cannot be
applied again for this atom. However, this information is neglected in the translated TRS.
Here,
one only regards the input
argument of ${\sf p}$ in order to determine whether a rule can be applied.
Note that many
current tools for termination proofs of logic programs like
{\sf cTI}~\cite{Mesnard:Bagnara},
{\sf TALP}~\cite{TALP},
{\sf TermiLog}~\cite{Lindenstrauss:Sagiv:Serebrenik}, and
\mbox{{\sf
TerminWeb}}~\cite{Codish:Taboch} fail on
\rEx{ex_not_sm}.
}
\end{example}

So this example already illustrates a drawback of the classical
transformation: there are several terminating well-moded logic programs which
are transformed into
non-terminating TRSs. In such cases, one fails in
proving the termination of the logic program.
Even worse, most of the existing transformations are not applicable for logic
programs that are not well moded.\footnote{\rEx{append} is neither well
moded nor simply well typed nor safely typed (using the types ``\emph{Any}''
and ``\emph{Ground}'') as required by the transformations of
\cite{Aguzzi:Modigliani,Arts:Zantema:95,Chtourou:Rusinowitch,Ganzinger:Waldmann,Rao,Marchiori:ALP94,Marchiori:AMAST96,Raamsdonk}.}

\noindent\begin{example}\label{append}{\sl
We modify \rEx{ex_not_sm} by replacing ${\sf g}(Y)$ with ${\sf g}(W)$ in the
body of the second clause:
\[\begin{array}{lll}
{\sf p}(X, X). & &\\
{\sf p}({\sf f}(X), {\sf g}(Y)) &\from& {\sf p}({\sf f}(X), {\sf f}(Z)),
{\sf p}(Z,{\sf g}(W)).
\end{array}\]
Still, all queries ${\sf p}(t_1,t_2)$ terminate if $t_1$ is ground.
But this program is not well moded,
as the second clause violates Condition (a): $\V_{out}({\sf p}({\sf f}(X), {\sf
g}(Y))) = \{ Y \} \not\subseteq \V_{in}({\sf p}({\sf f}(X), {\sf g}(Y))) \cup
\V_{out}({\sf p}({\sf f}(X), {\sf f}(Z))) \cup \V_{out}({\sf p}(Z,{\sf g}(W)))
= \{ X, Z, W \}$.
Transforming the program as before yields a TRS with the rule
${\sf u}_2({\sf p}_{out}({\sf g}(W)),X, Z) \to {\sf p}_{out}({\sf g}(Y))$.
So non-well-moded programs result in rules with
variables like $Y$  in the
right- but not in the left-hand side.   Such rules are usually forbidden in
term rewriting and they do not terminate, since $Y$~may be
instantiated arbitrarily.
}
\end{example}

\begin{example}\label{real_append}{\sl
A natural non-well-moded example is the
{\sf append}-program with the clauses
\[\begin{array}{lll}
{\sf append}([\,], \mathit{M}, \mathit{M}).\\
{\sf append}([X|\mathit{L}], \mathit{M}, [X|\mathit{N}]) & \from & {\sf append}(\mathit{L},\mathit{M},\mathit{N}).
\end{array}\]
and the moding $m({\sf append},1) = \mIn$ and $m({\sf append},2) = m({\sf append},3) = \mOut$,
i.e., one only considers {\sf append}'s first argument as input.
Due to the first clause ${\sf append}([\,], \mathit{M}, \mathit{M})$, this
program is not well moded although all queries of the form ${\sf append}(t_1,t_2,t_3)$
are terminating if $t_1$ is ground.
}
\end{example}

\subsection{Term Rewriting Techniques for Termination of Logic Programs}\label{Term
Rewriting Techniques for Termination of LPs}

Recently, several authors tackled
the problem of
applying termination techniques from term rewriting for
(possibly non-well-moded) logic
programs. A framework for integrating
orders from term rewriting
into \emph{direct} termination
approaches for  logic programs
is discussed in \cite{DeSchreye:Serebrenik:Kowalski}.\footnote{But
 in   contrast
to   \cite{DeSchreye:Serebrenik:Kowalski}, 
transformational approaches like the one presented in this paper
can also apply
more recent
termination   techniques
from term rewriting
for termination of logic
programs
(e.g., refined variants of the
\emph{dependency pair} method like
\cite{LPAR04,JAR06,HM05}, \emph{semantic labelling} \cite{Z95},
\emph{matchbounds} \cite{GHW}, etc.).}
However, the automation of this framework
is non-trivial in general.
As an instance of this framework, the automatic
application of polynomial interpretations (well-known in term rewriting) to
termination analysis of logic programs is investigated in
\cite{Nguyen:DeSchreye,Nguyen:DeSchreye06}. Moreover, \cite{LOPSTR07}
extend this work further by also adapting a basic version of the
\emph{dependency pair approach} \cite{AG00} from TRSs to the logic programming
setting. This provides additional evidence that techniques developed for term rewriting
can successfully be applied to termination analysis of logic programs.

Instead of integrating each termination technique from term
rewriting separately, in the current paper we want to make all these techniques available at once.
Therefore,
unlike~\cite{DeSchreye:Serebrenik:Kowalski,Nguyen:DeSchreye,Nguyen:DeSchreye06,LOPSTR07},
we choose a transformational approach.
Our goal is a method which
\begin{itemize}
\item[(A)] handles programs like \rEx{ex_not_sm} where classical transformations
like the one of \rSec{The Classical Transformation} fail,
\item[(B)] handles non-well-moded programs like \rEx{append}
where most current transformational techniques
are not even applicable,
\item[(C)] allows the successful \emph{automated} application of powerful
techniques
from
rewriting for logic programs like \rEx{ex_not_sm} and \ref{append}
where current tools based on
direct
approa\-ches fail. For larger and more realistic examples we refer to
the experiments in \rSec{sec:experiments}.
\end{itemize}

\subsection{Structure of the Paper}\label{Structure of the Paper}

After presenting required preliminaries in
\rSec{sec_preliminaries},
in \rSec{sec_translation}
we modify the transformation from logic programs to TRSs
to achieve (A) and (B). So restrictions like
well-modedness, simple well-typedness, or safe typedness are no
longer required.
Our new transformation results in TRSs where the notion of
``rewriting'' has to be slightly modified: we regard a restricted form of
infinitary rewriting,
called {\em infinitary
constructor rewriting}. The reason is that
logic
programs use
\emph{unification}, whereas TRSs use
\emph{matching}.

To illustrate this difference, consider the logic program  $\Fp(\Fs(X)) \from
\Fp(X)$ which
does not terminate for the query $\Fp(X)$:
Unifying the query $\Fp(X)$ with the 
head of the variable-renamed rule $\Fp(\Fs(X_1)) \from \Fp(X_1)$ 
yields the new query $\Fp(X_1)$. Afterwards, unifying the new query
$\Fp(X_1)$ with the head of the variable-renamed rule $\Fp(\Fs(X_2)) \from \Fp(X_2)$ 
yields the new query $\Fp(X_2)$, etc.

In contrast, the related TRS
$\Fp(\Fs(X)) \to \Fp(X)$
terminates for all finite terms.
When applying the rule to some subterm $t$, 
one has to
\emph{match} the left-hand side $\ell$ of the rule against $t$.
For example,
when applying the rule to the term  $\Fp(\Fs(\Fs(Y)))$, one 
would use the matcher that instantiates $X$ with $\Fs(Y)$. Thus,
$\Fp(\Fs(\Fs(Y)))$ would be
rewritten to  the instantiated
right-hand side $\Fp(\Fs(Y))$. Hence, one occurrence of the symbol $\Fs$ is
eliminated in every rewrite step. This implies that rewriting will always
terminate.
So in contrast to unification (where one searches for a substitution $\theta$
with $t \theta = \ell \theta$), here we only use matching (i.e., we search for
a substitution $\theta$ with $t = \ell \theta$, but we do not instantiate the
term $t$ that is being rewritten).

However, the infinite derivation of the logic
program above corresponds to an infinite reduction of the TRS above
 with the \emph{infinite} term
$\Fp(\Fs(\Fs(\ldots)))$ containing infinitely many nested {\sf s}-symbols.
So to simulate unification by matching, we have to regard TRSs
where the variables in rewrite rules may be instantiated by infinite
constructor terms.
It turns out that this form of rewriting also allows us to analyze
the termination behavior of logic programming with infinite terms, i.e., of logic
programming without occur check.

\rSec{sec:termination} shows
that the existing termination techniques for TRSs can easily be
adapted
in order to prove termination of infinitary
constructor rewriting. For a full automation of the approach,
one has to transform the set of queries that has to be analyzed for the logic
program to a corresponding set of terms that has to be analyzed for the
transformed TRS. This set of terms is characterized by a so-called
\emph{argument filter} and we present heuristics to find a suitable
argument filter in
\rSec{sec:refining}.
\rSec{sec:previous} gives a formal proof that our new transformation and our
approach to automated termination analysis are strictly
more powerful than the classical ones of \rSec{The Classical Transformation}.
We present and discuss an extensive experimental evaluation of our results in
\rSec{sec:experiments}
which shows that Goal (C) is achieved as well.  In other words,
the implementation of our approach can indeed compete with modern tools
for direct termination analysis of logic programs and  it succeeds for many
programs where these tools fail. Finally, we conclude in
\rSec{sec:conclusion}.

Preliminary versions of parts of this paper appeared in
\cite{LOPSTR06}.
However, the present article extends \cite{LOPSTR06}
considerably (in particular, by the results of the Sections
\ref{sec:refining} and \ref{sec:previous}).
Section \ref{sec:previous} contains a new formal comparison with the existing
classical transformational approach to termination of logic programs and
\emph{proves} formally that our approach is more powerful.
The new contributions of
\rSec{sec:refining}
improve the power of our
method substantially as can be seen in
our new experiments in
 \rSec{sec:experiments}.
Moreover, in contrast to \cite{LOPSTR06}, this
article contains the full proofs of all results and a discussion on the
limitations of our approach in \rSec{sec:limitations}.

\section{Preliminaries on Logic Programming and Rewriting}\label{sec_preliminaries}

We start with introducing the basics on (possibly infinite) terms and
atoms. Then we present the required notions on logic programming and on
term rewriting in Sections \ref{Logic Programming} and \ref{Term Rewriting}, respectively.

A \emph{signature} is a pair $(\Sigma, \Delta)$ where
$\Sigma$ and $\Delta$
are finite sets of
function  and predicate symbols.
Each $f \in \Sigma \cup \Delta$ has an {\em arity} $n \geq 0$ and we often write
$f/n$ instead of $f$. We always assume
that $\Sigma$ contains at least one
constant $f/0$.
This is not a restriction, since
enriching the signature by a fresh constant would not
change the termination behavior.

\begin{definition}[(Infinite Terms and Atoms)]
\label{Finite, Rational, and Infinite Terms, Atom}
A
\emph{term} over $\Sigma$
is a tree where every leaf node
is labelled with a variable $X \in\V$ or with $f/0 \in \Sigma$
and every inner node with $n$ children ($n > 0$) is labelled with some $f/n \in \Sigma$.
We write $f(t_1,\ldots,t_n)$ for the term with root $f$ and direct subtrees
$t_1,\ldots,t_n$.
A term $t$ is called \emph{finite} if all paths in the tree $t$
are finite, otherwise it is  \emph{infinite}.
A term is \emph{rational} if it only contains finitely many different subterms. The sets
of all finite terms, all rational terms, and all (possibly infinite) terms
over $\Sigma$ are denoted by
$\T(\Sigma,\V)$,  $\T^{rat}(\Sigma, \V)$, and $\T^\infty(\Sigma,\V)$,
respectively.
If $\vec{t}$ is the sequence   $t_1, \ldots, t_n$, then
$\vec{t} \in \vec{\T}^\infty(\Sigma,\V)$ means that
$t_i \in \T^\infty(\Sigma,\V)$ for all $i$.
$\vec{\T}(\Sigma,\V)$ is defined analogously. For a term $t$, let $\V(t)$ be
the set of all variables occurring in $t$.
A \emph{position} $\mathit{pos} \in \N^*$ in a (possibly infinite) term $t$
addresses a subterm
$t|_{\mathit{pos}}$ of $t$.
We denote the empty word (and thereby the top 
position) by $\varepsilon$.
 The term
$t[s]_{\mathit{pos}}$ results from replacing the subterm $t|_{\mathit{pos}}$ at position ${\mathit{pos}}$ in $t$ by
the term $s$.
So for ${\mathit{pos}} = \varepsilon$ we have $t|_\varepsilon = t$ and $t[s]_\varepsilon =
s$. 
Otherwise ${\mathit{pos}} = i\,\mathit{pos}'$ for some $i \in \N$ and
$t = f(t_1,\ldots,t_n)$. Then we have
$t|_{\mathit{pos}} = t|_{i\,\mathit{pos}'} = t_i|_{\mathit{pos}'}$ and
$t[s]_{\mathit{pos}} = t[s]_{i\,\mathit{pos}'} = f(t_1,\ldots,t_i[s]_{\mathit{pos}'}\ldots,t_n)$.

An \emph{atom} over $(\Sigma,\Delta)$
is a tree  $p(t_1,\ldots,t_n)$, where $p/n \in \Delta$  and
$t_1,\ldots,t_n \in \T^\infty(\Sigma,\V)$.  $\A^\infty(\Sigma,\Delta,\V)$ is the set of  atoms
and $\A^{rat}(\Sigma,\Delta,\V)$ (and $\A(\Sigma,\Delta,\V)$, resp.)
 are the atoms
$p(t_1,\ldots,t_n)$
where $t_i \in \T^{rat}(\Sigma,\V)$
(and $t_i \in \T(\Sigma,\V)$, resp.) for all $i$.
We write $\A(\Sigma,\Delta)$ and $\T(\Sigma)$ instead of $\A(\Sigma,\Delta,\varnothing)$ and $\T(\Sigma,\varnothing)$.
\end{definition}

\subsection{Logic Programming}\label{Logic Programming}

A {\em   clause}
$\, c $
is a formula
$H \from B_1,\ldots, B_k$ with $k\geq0$ and
$H, B_i \in \A(\Sigma,\Delta,\V)$. $H$ is $c$'s  {\em head} and
$B_1,\ldots,B_k$ is $c$'s {\em body}.
A finite set of clauses $\P$ is a {\em definite logic program}.
A clause with empty body is a {\em fact} and a clause with empty head is a {\em query}.
We usually omit ``\from'' in
queries and just write ``$B_1,\ldots,B_k$''.
The empty query is denoted $\Box$. In queries, we also admit rational instead of finite atoms
$B_1,\ldots,B_k$.

 Since we are also interested in
logic programming without occur check
we consider
infinite \emph{substitutions}
$\theta: \V \rightarrow \T^\infty(\Sigma,\V)$.
 Here, we allow $\theta(X) \neq X$ for infinitely
many $X\in \V$.
Instead of $\theta(X)$
we often write $X\theta$. If $\theta$ is
a variable renaming (i.e., a one-to-one correspondence on
$\V$), then $t\theta$ is a {\em variant} of $t$, where
$t$ can be any expression (e.g., a term, atom, clause, etc.).
We write $\theta\sigma$ to denote that
the application of $\theta$ is followed by the application of
$\sigma$.

A substitution $\theta$ is a \emph{unifier} of two terms $s$ and $t$ if and only if
$s\theta = t\theta$. We call $\theta$ the \emph{most general unifier (mgu)} of $s$ and $t$
if and only if $\theta$ is a unifier of $s$ and $t$ and for all unifiers $\sigma$ of
$s$ and $t$ there is a substitution $\mu$ such that $\sigma = \theta\mu$.

We  briefly present the procedural semantics of logic programs based on
SLD-resolution
using the
left-to-right selection rule implemented by
most
\textsf{Prolog}  systems.
 More details on logic
programming
can be found in~\cite{Apt:Book}, for example.

\begin{definition}[(Derivation, Termination)]\label{termination}
Let $Q$ be a query $A_1,\ldots,A_m$, let $c$ be
a clause $H\from B_1,\ldots,B_k$.
Then $Q'$ is a {\em resolvent} of $Q$ and $c$ using $\theta$ (denoted
$Q \vdash_{c,\theta} Q'$)
if
$\theta$ is the mgu\footnote{
Note that for finite sets of \emph{rational} atoms or terms,
unification is decidable,
the mgu is unique modulo renaming, and it is a
substitution with \emph{rational} terms~\cite{Huet76}.} of $A_1$ and $H$,
and $Q' = (B_1,\ldots,B_k,A_{2},\ldots,A_m)\theta$.

A {\em derivation\/} of a program $\P$ and $Q$
is a possibly infinite sequence $Q_0, Q_1, \ldots$ of queries with $Q_0
= Q$ where for all $i$, we have
$Q_i \vdash_{c_{i+1},\theta_{i+1}} Q_{i+1}$ for some
substitution $\theta_{i+1}$ and some fresh variant $c_{i+1}$ of a
clause
of $\P$. For a derivation $Q_0, \ldots, Q_n$ as above,
we also write $Q_0 \vdash^n_{\P,\theta_1\ldots\theta_n} Q_n$
or $Q_0 \vdash^n_{\P} Q_n$, and we also write
$Q_i \vdash_\P Q_{i+1}$ for $Q_i \vdash_{c_{i+1},\theta_{i+1}} Q_{i+1}$.
The query
$Q$ {\em terminates} for  $\P$ if all
derivations of $\P$ and $Q$ are  finite.
\end{definition}

Our  notion of derivation  coincides with logic programming without
an occur check~\cite{Colmerauer82}
as  implemented in recent \textsf{Prolog} systems such as \textsf{SICStus} or
\textsf{SWI}.
Since we consider only
definite logic programs, any program which is terminating without occur
check is also terminating with occur check, but not vice versa. So
if our approach detects ``termination'', then the program is indeed
terminating, no matter whether one uses logic programming with or without
occur check. In other words,
our
approach is sound for both kinds of logic programming, whereas most other approaches
only consider logic programming with occur check.

\begin{example}{\sl
Regard a program $\P$ with the clauses
${\sf p}(X) \from {\sf equal}(X,{\sf s}(X)), {\sf p}(X)$ and ${\sf equal}(X,X)$.
We obtain ${\sf p}(X) \vdash^2_{\P} {\sf p}({\sf s}({\sf s}(\ldots)))
\vdash^2_{\P} {\sf p}({\sf s}({\sf s}(\ldots)))\vdash^2_{\P} \ldots$, where ${\sf s}({\sf
s}(\ldots))$ is the term containing infinitely many nested {\sf s}-symbols.
So the finite
query ${\sf p}(X)$ leads to a derivation with infinite (rational) queries.
While  ${\sf p}(X)$ is not terminating according to
\rDef{termination}, it would be terminating if one uses logic programming with occur check.
Indeed, tools like {\sf cTI}
\cite{Mesnard:Bagnara} and {\sf TerminWeb}~\cite{Codish:Taboch}
report that such queries are
``terminating''. So in contrast to our technique, such tools are in general not sound for
logic programming without occur check, although this form of logic
programming is typically used in practice.
}
\end{example}

\subsection{Term Rewriting}\label{Term Rewriting}

Now we define TRSs
and introduce the notion of \emph{infinitary constructor rewriting}. For
further details on term rewriting we refer to \cite{BN98}.

\begin{definition}[(Infinitary Constructor Rewriting)]
A \emph{TRS} $\R$ is a finite set
of rules $\el\to r$ with $\el,r \in \T(\Sigma,\V)$ and $\el \notin
\V$.\footnote{In standard term rewriting, one usually requires $\V(r)
\subseteq \V(\el)$ for all rules $\el \to r$. The reason is that otherwise the
standard rewrite relation is
never well founded. However, the \emph{infinitary constructor rewrite relation} defined
here can be well founded even if $\V(r)
\not\subseteq \V(\el)$.}
We divide the signature into  \emph{defined}
symbols $\Sigma_D = \{ f \mid \el \to r \in \R, \rt(\el) = f \}$
and  \emph{constructors} $\Sigma_C = \Sigma \setminus \Sigma_D$.
$\R$'s \emph{infinitary constructor rewrite relation}
is denoted  $\to_\R$: for $s,t \in \T^\infty(\Sigma,\V)$
we have $s \to_\R t$ if there is a rule $\el \to r$, a position ${\mathit{pos}}$ and a
substitution
 $\sigma : \V \to \T^\infty(\Sigma_C,\V)$
with
$s|_{\mathit{pos}} = \el\sigma$ and $t = s[r\sigma]_{\mathit{pos}}$.
Let $\to_\R^n$, $\to_\R^{\geq n}$, $\to^*_\R$ denote rewrite sequences
of $n$ steps, of at least $n$ steps, and of arbitrary many steps, respectively (where $n
\geq 0$).
A term $t$ is \emph{terminating} for
$\R$ if there is no infinite sequence of the form $t \to_\R t_1 \to_\R t_2 \to_\R \ldots$
A TRS $\R$ is \emph{terminating} if all terms are terminating for
$\R$.
\end{definition}

The above definition of $\to_\R$ differs from the usual rewrite relation in
two aspects:
\begin{itemize}
\item[(i)\phantom{i}] We only permit
instantiations of rule variables by constructor terms.
\item[(ii)]
We use
substitutions with possibly non-rational infinite terms.
\end{itemize}
In
\rEx{ex_not_sm_we} and \ref{Modification} in the next section, we will
motivate these modifications
and show that
there are TRSs which terminate w.r.t.\ the
usual rewrite relation, but are  non-terminating w.r.t.\ infinitary
constructor rewriting and vice versa.

\section{Transforming Logic Programs into Term Rewrite Systems}
\label{sec_translation}

Now we modify the transformation of logic programs into TRSs from
\rSec{sec:Introduction}  to make it applicable for
\emph{arbitrary} (possibly non-well-moded) programs as well.
We present the new transformation in \rSec{The Improved Transformation}
and prove its soundness in \rSec{sec Soundness of the Transformation}.
Later in \rSec{sec:previous} we will formally
prove that the classical transformation is  strictly subsumed by our new one.

\subsection{The Improved Transformation}
\label{The Improved Transformation}

Instead of separating between input and output positions of a predicate
$p/n$, now we keep
\emph{all} arguments both for $p_{in}$ and $p_{out}$ (i.e.,  $p_{in}$
and $p_{out}$ have arity $n$).

\begin{definition}[(Transformation)]
\label{unmodedTranslation}
A logic program $\P$ over $(\Sigma,\Delta)$
is transformed into the following TRS $\R_\P$ over $\Sigma_\P = \Sigma \cup
\{p_{in}/n, \,
p_{out}/n \mid p/n\in \Delta\} \cup
\{u_{c,i} \mid c \in \P, 1\leq i\leq k,$ where $k$ is the number of atoms
in the body of $c\,\}$.
\begin{itemize}
\item[$\bullet$] For each fact
$p(\vec{s})$ in $\P$, the TRS $\R_\P$ contains the rule
$p_{in}(\vec{s}) \to p_{out}(\vec{s})$.
\item[$\bullet$] For each clause $c$ of the form
$p(\vec{s}) \from p_1(\vec{s}_1), \ldots, p_k(\vec{s}_k)$ in $\P$,
 $\R_\P$ contains:
\begin{align*}
\urule{p_{in}(\vec{s}) }{ u_{c,1}(p_{1_{in}}(\vec{s}_1), \V(\vec{s}))}\\
\urule{u_{c,1}(p_{1_{out}}(\vec{s}_1), \V(\vec{s})) }{
u_{c,2}(p_{2_{in}}(\vec{s}_2), \V(\vec{s}) \cup \V(\vec{s}_1))}\\
& \ldots \\
\urule{u_{c,k}(p_{k_{out}}(\vec{s}_k), \V(\vec{s}) \cup \V(\vec{s}_1) \cup \ldots
\cup \V(\vec{s}_{k-1})) }{ p_{out}(\vec{s})}
\end{align*}
\end{itemize}
\end{definition}

The following two examples
motivate the need for infinitary constructor rewriting
in \rDef{unmodedTranslation},
i.e., they motivate Modifications (i) and (ii) in \rSec{Term Rewriting}.

\begin{example}\label{ex_not_sm_we}
{\sl For the logic program of \rEx{ex_not_sm}, the transformation of
\rDef{unmodedTranslation}
yields the following
TRS.
\begin{align}
\lrule{32-1}{
{\sf p}_{in}(X,X) }{ {\sf p}_{out}(X,X)}\\
\lrule{32-2}{
{\sf p}_{in}({\sf f}(X), {\sf g}(Y)) }{ {\sf u}_1({\sf p}_{in}({\sf f}(X),
{\sf f}(Z)),X,Y)} \\
\lrule{32-3}{
{\sf u}_1({\sf p}_{out}({\sf f}(X), {\sf f}(Z)),X,Y) }{ {\sf u}_2({\sf p}_{in}(Z,{\sf g}(Y)),X,Y,Z)}\\
\lrule{32-4}{
{\sf u}_2({\sf p}_{out}(Z,{\sf g}(Y)),X,Y,Z) }{ {\sf p}_{out}({\sf f}(X),{\sf g}(Y))}
\end{align}
This example shows why rules of TRSs may only be instantiated with constructor
terms (Modification (i)). The reason is that local variables
like $Z$ (i.e., variables occurring in the body but not in the head of a clause)
give rise to rules $\el \to r$ where $\V(r) \not\subseteq
\V(\el)$ (cf.\ Rule (\ref{32-2})).
Such rules are never terminating in standard term rewriting.
However, in our setting one may only instantiate $Z$ with
constructor terms. So in contrast to the old transformation in \rEx{old transformation},
now all terms ${\sf p}_{in}(t_1,t_2)$ terminate for the TRS
if $t_1$ is finite, since now the second argument of ${\sf p}_{in}$ prevents an
infinite application of Rule (\ref{32-2}). Indeed, constructor rewriting correctly
simulates the behavior of logic programs,
since the variables in a logic program are only instantiated by ``constructor
terms''.

For the non-well-moded program of \rEx{append}, one obtains a similar TRS
where ${\sf g}(Y)$ is replaced
by ${\sf g}(W)$
in the right-hand side of
Rule (\ref{32-3}) and the left-hand side of Rule (\ref{32-4}). Again, all
terms
${\sf p}_{in}(t_1,t_2)$ are terminating for this TRS provided that $t_1$ is finite.
 Thus, we
can now handle programs where the classical transformation  of
\cite{Arts:Zantema:95,Chtourou:Rusinowitch,Ganzinger:Waldmann,Ohlebusch01} failed,
cf.\ Goals (A) and (B) in \rSec{Term
Rewriting Techniques for Termination of LPs}.
}
\end{example}

Derivations in logic programming use \emph{unification}, while
rewriting is defined by \emph{matching}. \rEx{Modification} shows that
to simulate unification by matching, we have to consider
substitutions with infinite and
even non-rational terms (Modification
(ii)).

\begin{example}\label{Modification}
{\sl Let
$\P$ be
\mbox{\small ${\sf ordered}({\sf cons}(X,{\sf cons}({\sf s}(X),\mathit{XS}))) \from {\sf ordered}({\sf
cons}({\sf s}(X),\mathit{XS}))$}.
If one only considers
rewriting with finite or  rational
terms, then the transformed TRS $\R_\P$ is terminating.  However,
the query
${\sf ordered}(\mathit{YS})$ is not terminating for
$\P$.  Thus, to obtain a sound approach, $\R_\P$
must
also be non-terminating.  Indeed, the term $t = 
{\sf ordered}_{in}(\Fcons(X,\;\Fcons(\Fs(X),\; \Fcons(\Fs^2(X),\;\ldots))))$
is non-terminating  with $\R_\P$'s rule
${\sf ordered}_{in}({\sf cons}(X,{\sf cons}({\sf s}(X),\mathit{XS})))\to {\sf u}({\sf ordered}_{in}({\sf
cons}({\sf s}(X),\mathit{XS})),X,\linebreak \mathit{XS})$. The
non-rational term $t$ corresponds to the
infinite derivation with the query ${\sf ordered}(\mathit{YS})$.
}
\end{example}

\subsection{Soundness of the Transformation}\label{sec Soundness of the Transformation}

We first show an auxiliary lemma that
is needed
to prove the soundness of the transformation. It relates derivations
with the logic program $\P$ to rewrite sequences with
the TRS $\R_\P$.

\begin{lemma}[(Connecting $\P$ and $\R_\P$)]
\label{correctTransformation}
Let
$\P$ be a program,
let $\vec{t}$ be terms from $\T^{rat}(\Sigma,\V)$, let
$p(\vec{t})  \vdash_{\P,\sigma}^n Q$.
If $Q = \Box$, then $p_{in}(\vec{t})\sigma \to^{\geq n}_{\R_\P}
p_{out}(\vec{t})\sigma$.
Otherwise,  if $Q$ is ``$q(\vec{v}),\ldots$'', then $p_{in}(\vec{t})\sigma
\to^{\geq n}_{\R_\P} r$ for a term $r$ containing the subterm
$q_{in}(\vec{v})$.
\end{lemma}
\begin{proof}
Let $p(\vec{t}) = Q_0 \vdash_{c_{1},\theta_{1}}
\ldots \vdash_{c_{n},\theta_{n}} Q_n = Q$ with $\sigma =
\theta_1\ldots\theta_n$. We use induction
on
$n$.
The  base case $n = 0$ is trivial, since $Q = p(\vec{t})$
and $p_{in}(\vec{t}) \to^0_{\R_\P} p_{in}(\vec{t})$.

Now let $n\geq 1$. We first regard the case $Q_1 = \Box$ and $n = 1$.
Then $c_1$ is a fact $p(\vec{s})$ and
$\theta_1$ is the mgu of
$p(\vec{t})$ and $p(\vec{s})$.
Note that such mgu's instantiate all variables with constructor terms (as
symbols of $\Sigma$ are constructors of $\R_\P$).
We obtain $p_{in}(\vec{t})\theta_1 =
p_{in}(\vec{s})\theta_1 \to_{\R_\P} p_{out}(\vec{s})\theta_1 =
p_{out}(\vec{t})\theta_1$ where $\sigma = \theta_1$.

Finally, let $Q_1 \neq \Box$. Thus,
 $c_1$ is $p(\vec{s}) \from p_1(\vec{s_1}), \ldots,
p_k(\vec{s_k})$ and
$Q_1$ is $p_1(\vec{s_1})\theta_1, \ldots,$ $p_k(\vec{s_k})\theta_1$
where $\theta_1$ is the mgu
of $p(\vec{t})$ and $p(\vec{s})$.
There is an $i$ with $1 \leq i \leq k$ such that for all $j$ with $1 \leq j \leq i-1$
we have
$p_j(\vec{s}_j)\sigma_0 \ldots \sigma_{j-1} \vdash^{n_j}_{\P,\sigma_j} \Box$.
Moreover, if $Q=\Box$ then we can choose $i = k$ and
$p_i(\vec{s}_i)\sigma_0 \ldots \sigma_{i-1} \vdash^{n_i}_{\P,\sigma_i} \Box$.
Otherwise, 
if $Q$ is ``$q(\vec{v}),\ldots$'', then we can choose $i$ such that
$p_i(\vec{s}_i)\sigma_0 \ldots
\sigma_{i-1} \vdash^{n_i}_{\P,\sigma_i} q(\vec{v}),\ldots\;$ Here, $n =  n_1
+ \ldots + n_i + 1$, $\sigma_0 = \theta_1$, $\sigma_1 = \theta_2 \ldots
\theta_{n_1 + 1}$, \ldots, and $\sigma_i =
\theta_{n_1 + \ldots + n_{i-1} + 2} \ldots \theta_{n_1 + \ldots + n_i+1}$. So $\sigma =
\sigma_0 \ldots \sigma_i$.

By the induction hypothesis
we have $p_{j_{in}}(\vec{s}_j)\sigma_0 \ldots \sigma_{j}
\to_{\R_\P}^{\geq n_j} p_{j_{out}}(\vec{s}_j)\sigma_0 \ldots \sigma_{j}$
and thus also $p_{j_{in}}(\vec{s}_j)\sigma
\to_{\R_\P}^{\geq n_j} p_{j_{out}}(\vec{s}_j)\sigma$.
Moreover, if $Q=\Box$ then we also have $p_{i_{in}}(\vec{s}_i)\sigma
\to_{\R_\P}^{\geq n_i} p_{i_{out}}(\vec{s}_i)\sigma$ where $i =k$.
Otherwise, if $Q$ is ``$q(\vec{v}),\ldots$'', then
the induction hypothesis implies $p_{i_{in}}(\vec{s}_i)\sigma
\to_{\R_\P}^{\geq n_i} r'$, where $r'$ contains $q_{in}(\vec{v})$.
Thus
\[\begin{array}{lll}
p_{in}(\vec{t})\sigma \; = \; p_{in}(\vec{s})\sigma
& \to_{\R_\P} &
u_{c_1,1}(p_{1_{in}}(\vec{s}_1), \V(\vec{s}))\sigma\\
& \to_{\R_\P}^{\geq n_1} &
u_{c_1,1}(p_{1_{out}}(\vec{s}_1), \V(\vec{s}))\sigma\\
& \to_{\R_\P} &
u_{c_1,2}(p_{2_{in}}(\vec{s}_2), \V(\vec{s}) \cup  \V(\vec{s}_1))\sigma\\
& \to_{\R_\P}^{\geq n_2} &
u_{c_1,2}(p_{2_{out}}(\vec{s}_2), \V(\vec{s}) \cup  \V(\vec{s}_1))\sigma\\
& \to_{\R_\P}^{\geq n_3 + \ldots + n_{i-1}} &
u_{c_1,i}(p_{i_{in}}(\vec{s}_i), \V(\vec{s}) \cup \V(\vec{s}_1) \cup \ldots \cup
\V(\vec{s}_{i-1})) \sigma
\end{array}\]
Moreover, if $Q=\Box$, then $i= k$ and the rewrite sequence yields
$p_{out}(\vec{t})\sigma$, since
\[\begin{array}{lll}
 u_{c_1,i}(p_{i_{in}}(\vec{s}_i), \V(\vec{s}) \cup  \ldots \cup
\V(\vec{s}_{i-1})) \sigma
& \to_{\R_\P}^{\geq n_i} &
u_{c_1,i}(p_{i_{out}}(\vec{s}_i), \V(\vec{s}) \cup \ldots \cup
\V(\vec{s}_{i-1})) \sigma \\
&\to_{\R_\P}& p_{out}(\vec{s})\sigma \quad = \quad p_{out}(\vec{t})\sigma.
\end{array}\]
Otherwise, if $Q$ is ``$q(\vec{v}),\ldots$'', then rewriting yields a term containing
$q_{in}(\vec{v})$:
\[\begin{array}{lll}
 u_{c_1,i}(p_{i_{in}}(\vec{s}_i), \V(\vec{s})\cup  \ldots \cup
\V(\vec{s}_{i-1})) \sigma
& \to_{\R_\P}^{\geq n_i} &
u_{c_1,i}(r', \V(\vec{s}) \sigma \cup \ldots \cup
\V(\vec{s}_{i-1}) \sigma). \quad \qed
\end{array}\]
\end{proof}

For the soundness proof, we need another lemma which
states that we can restrict ourselves to non-terminating queries
which only consist of a single atom.

\begin{lemma}[(Non-Terminating Queries)]
\label{infLemma}
Let $\P$  be a logic program. Then for   every infinite derivation
$Q_0 \vdash_\P Q_1 \vdash_P \ldots$,
there is a $Q_i$ of the form ``$q(\vec{v}), \ldots$''
with $i>0$
 such that
the query $q(\vec{v})$ is also
non-terminating.
\end{lemma}
\begin{proof}
Assume that for all $i>0$, the first atom in $Q_i$ does not have an infinite
derivation. Then for each $Q_i$ there are two cases: either the first atom fails or it
can successfully be proved. In the former case, there is no infinite
reduction from $Q_i$ which contradicts the infiniteness of the derivation from
$Q_0$. Thus for all $i > 0$, the first atom of $Q_i$ is successfully proved in $n_i$ steps
during the derivation $Q_0 \vdash_\P Q_1 \vdash_\P \dots\;$
Let $m$ be the number of atoms
in $Q_1$. But then $Q_{1+n_1+\ldots+n_m}$ is the empty query $\Box$ which again
contradicts the infiniteness of the derivation.
\qed
\end{proof}

We use \emph{argument filters}
to characterize the classes of queries
whose termination we want to analyze. 
Related definitions can be found  in, e.g.,
\cite{AG00,Leuschel:Sorensen}.

\begin{definition}[(Argument Filter)]
\label{defn:af}
A function $\pi:\Sigma\cup \Delta\rightarrow 2^\N$ is an \emph{argument filter $\pi$ over a signature $(\Sigma,\Delta)$} if and only if
$\pi(f/n)\subseteq \{1, \ldots, n\}$ for every $f/n\in \Sigma\cup\Delta$.
We extend $\pi$ to terms and atoms by defining
$\pi(x) = x$ if $x$ is a variable and
$\pi(f(t_1,\ldots,t_n)) = f(\pi(t_{i_1}),\ldots,\pi(t_{i_k}))$ if $\pi(f/n) =
\{i_1,\ldots,i_k\}$ with $i_1 < \ldots < i_k$.
Here, the new terms and atoms are from the filtered signature
$(\Sigma_\pi,\Delta_\pi)$ where $f/n \in \Sigma$ implies $f/k \in \Sigma_\pi$
and likewise for $\Delta_\pi$.
For a logic program $\P$ we write $(\Sigma_{\P_\pi}, \Delta_{\P_\pi})$ instead
of $((\Sigma_\P)_\pi, (\Delta_\P)_\pi)$.  For any TRS $\R$, we define $\pi(\R)
= \{ \pi(\el) \to \pi(r) \mid \el \to r \in \R \}$. 
The set of all argument filters over a signature $(\Sigma,\Delta)$
is denoted by $AF(\Sigma,\Delta)$. We write $AF(\Sigma)$ instead of
$AF(\Sigma,\varnothing)$ and speak of an argument filter ``over $\Sigma$''.
We also write $\pi(f)$ instead of $\pi(f/n)$ if the arity of $f$ is clear from the context.

An argument filter $\pi'$ is a \emph{refinement} of a filter $\pi$
if and only if $\pi'(f) \subseteq \pi(f)$ for all $f \in \Sigma \cup \Delta$.
\end{definition}

Argument filters  specify those
positions
which have to be
instantiated with finite ground terms. Then, we analyze
termination of all
queries  $Q$  where $\pi(Q)$ is
a (finite) ground atom. In
\rEx{ex_not_sm}, we  wanted to prove termination for all
queries ${\sf p}(t_1,t_2)$ where $t_1$ is finite and ground. These queries
are described by the filter
$\pi(h) =
\{ 1 \}$ for all $h \in \{ {\sf p},
{\sf f}, {\sf g} \}$. Thus, we have $\pi({\sf p}(t_1,t_2)) = {\sf p}(\pi(t_1))$.

Note that argument
filters also operate on \emph{function} instead of just
\emph{predicate} symbols.
Therefore, they can
describe more sophisticated classes of queries than the classical approach of
\cite{Arts:Zantema:95,Chtourou:Rusinowitch,Ganzinger:Waldmann,Ohlebusch01}
which
only distinguishes between input and output
positions of predicates. For example, if one wants to analyze all queries ${\sf
append}(t_1,t_2,t_3)$ where $t_1$ is a finite list, one would use the
filter $\pi({\sf append}) = \{1\}$ and $\pi(\point) = \{2\}$, where ``\point'' is the list constructor (i.e., $\point(X,L) = [X|L]$). Of course, our method can easily prove
that all these queries are terminating for the program of \rEx{real_append}. 

Now we show the soundness theorem:
to prove termination of all queries
$Q$ where $\pi(Q)$ is a finite ground atom,
it suffices to show termination of all those terms $p_{in}(\vec{t})$ for the TRS $\R_\P$
where $\pi(p_{in}(\vec{t}))$ is a finite ground term and where $\vec{t}$
only contains function symbols from the logic program $\P$. Here, $\pi$ has to
be extended to the new function symbols $p_{in}$ by defining $\pi(p_{in}) = \pi(p)$.

\begin{theorem}[(Soundness of the Transformation)]
\label{Soundness of the Transformation}
Let $\P$ be a logic program and let $\pi$ be an argument filter
over $(\Sigma,\Delta)$. We extend $\pi$ such that
$\pi(p_{in}) = \pi(p)$ for all $p \in \Delta$.
Let $S = \{ p_{in}(\vec{t})\mid
p\in\Delta,
\; \vec{t} \in \vec{\T}^\infty(\Sigma,\V), \; \pi(p_{in}(\vec{t}))\in\T(\Sigma_{\P_\pi})
\, \}$. If all terms $s\in S$ are
terminating for $\R_\P$, then
all queries $Q\in\A^{rat}(\Sigma,\Delta,\V)$ with
$\pi(Q)\in\A(\Sigma_\pi,\Delta_\pi)$
are terminating
for $\P$.\footnote{It is currently open whether the converse holds as well. For a short discussion
see \rSec{sec:limitations}.}
\end{theorem}
\begin{proof}
Assume that there
is a non-terminating query $p(\vec{t})$ as above with
$p(\vec{t}) \vdash_\P Q_1 \vdash_\P Q_2 \vdash_\P    \ldots\;$
By \rLem{infLemma} there is an $i_1 > 0$ with
$Q_{i_1} = q_1(\vec{v}_1), \ldots$ and an infinite derivation
$q_1(\vec{v}_1)  \vdash_\P Q_1'  \vdash_\P Q_2'  \vdash_P\ldots\;$
From $p(\vec{t}) \vdash_{\P,\sigma_0}^{i_1} q_1(\vec{v}_1), \ldots$ and \rLem{correctTransformation}
we get $p_{in}(\vec{t})\sigma_0 \to_{\R_\P}^{\geq i_1} r_1$, where $r_1$
contains the subterm $q_{1_{in}}(\vec{v}_1)$.

By \rLem{infLemma} again, there is an  $i_2 > 0$ with
$Q_{i_2}' = q_2(\vec{v}_2), \ldots$ and  an infinite derivation
$q_2(\vec{v}_2) \vdash_\P Q_1'' \vdash_\P  \ldots\;$
From $q_1(\vec{v}_1) \vdash_{\P,\sigma_1}^{i_2} q_2(\vec{v}_2), \ldots$
and \rLem{correctTransformation}
we get $p_{in}(\vec{t})\sigma_0\sigma_1 \to_{\R_\P}^{\geq i_1} r_1\sigma_1
\to_{\R_P}^{\geq i_2} r_2,$
where $r_2$
contains the subterm $q_{2_{in}}(\vec{v}_2)$.

Continuing this reasoning we obtain an infinite sequence $\sigma_0, \sigma_1, \ldots$
of substitutions.  For each $j\geq 0$, let $\mu_j = \sigma_j \; \sigma_{j+1}  \ldots$ result from
the infinite composition of these substitutions.\footnote{The composition of
\emph{infinitely} many substitutions
$\sigma_0,\sigma_1,\ldots$ is defined as follows. The definition ensures that
$t\sigma_0\sigma_1\ldots$ is an
instance of $t\sigma_0\ldots\sigma_n$ for all terms (or atoms) $t$ and all $n\geq 0$.
It suffices to define
 the
symbols at the positions of
$t\sigma_0\sigma_1 \ldots$
for any term $t$.
Obviously, ${\mathit{pos}}$ is a position of $t\sigma_0\sigma_1 \ldots$ iff
${\mathit{pos}}$ is a position of
$t\sigma_0 \ldots \sigma_n$ for some
$n \geq 0$.  We define
that the symbol of $t\sigma_0\sigma_1 \ldots$ at such a position ${\mathit{pos}}$ is
$f \in \Sigma$  iff $f$ is at position ${\mathit{pos}}$ in
$t\sigma_0 \ldots \sigma_m$
 for some $m\geq0$.
Otherwise, $(t\sigma_0 \ldots \sigma_n)|_{\mathit{pos}} = X_0 \in \V$.
Let $n = i_0 < i_1 < \ldots$ be the maximal (finite or infinite) sequence with
 $\sigma_{i_j+1}(X_j) = \ldots = \sigma_{i_{j+1}-1}(X_j) = X_j$
and $\sigma_{i_{j+1}}(X_j)$ $= X_{j+1}$ for all $j$. We require $X_j \neq
X_{j+1}$, but permit $X_j = X_{j'}$ otherwise. If this sequence is
finite (i.e., it has the form $n = i_0 < \ldots < i_m$), then we define
$(t\sigma_0\sigma_1\ldots)|_{\mathit{pos}} = X_m$. Otherwise,  the substitutions perform
infinitely many variable renamings. In this case, we use one special
variable $Z_\infty$ and define $(t\sigma_0\sigma_1\ldots)|_{\mathit{pos}} = Z_\infty$.
So if $\sigma_0(X) = Y$, $\sigma_1(Y) = X$,
 $\sigma_2(X) = Y$, $\sigma_3(Y) = X$, etc., we define
$X\sigma_0\sigma_1 \ldots = Y \sigma_0\sigma_1 \ldots = Z_\infty$.}
Since $r_j\mu_j$ is an instance of
$r_j\sigma_j \ldots \sigma_n$ for all $n\geq j$,
we obtain that $\pin(\vec{t})\mu_0$ is non-terminating for $\R_\P$:
\[p_{in}(\vec{t})\mu_0 \to_{\R_\P}^{\geq i_1} r_1\mu_1
\to_{\R_\P}^{\geq i_2} r_2\mu_2 \to_{\R_\P}^{\geq i_3} \ldots\]
As $\pi(p(\vec{t})) \in \A(\Sigma_\pi,\Delta_\pi)$ and thus
$\pi(p_{in}(\vec{t})) =
\pi(p_{in}(\vec{t})\mu_0) \in \T(\Sigma_{\P_\pi})$, this is a contradiction. 
\qed\end{proof}

\section{Termination of Infinitary Constructor Rewriting}
\label{sec:termination}

One of the most powerful methods for automated termination analysis of
rewriting is the \emph{dependency pair} (DP) method \cite{AG00} which is
implemented in most current termination tools for  TRSs.  However, since
the DP method only proves
termination of term rewriting with \emph{finite} terms, its use is not sound
in our setting. Nevertheless, we now show that
only very slight modifications are required to adapt
dependency pairs from ordinary
rewriting to infinitary constructor  rewrit\-ing. So any rewriting tool
implementing dependency pairs
can easily be modified in order to prove termination of infinitary constructor rewriting
as well. Then,
it can also analyze termination of logic programs
using the transformation of \rDef{unmodedTranslation}.

Moreover,
dependency pairs are a general framework that permits the
integration of \emph{any} termination technique for TRSs
\cite[Thm.\ 36]{LPAR04}. Therefore, instead of adapting each technique separately,
it is sufficient  only to adapt the DP framework
to infinitary constructor
rewriting. Then, \emph{any} termination technique can be directly used for infinitary constructor
rewriting. In \rSec{Dependency Pairs for
Infinitary Rewriting}, we adapt the notions and the main
termination criterion of the dependency pair method to infinitary constructor
rewriting
and in \rSec{Automation
by Adapting the DP Framework} we show how
to automate this criterion by adapting the ``DP processors'' of the DP framework.

\subsection{Dependency Pairs for Infinitary Rewriting}\label{Dependency Pairs for Infinitary Rewriting}

Let $\R$ be a TRS. For each defined symbol
$f/n \in \Sigma_D$, we extend the set of constructors
$\Sigma_C$ by a fresh
\emph{tuple symbol} $f^\sharp/n$.
We often write $F$
instead of $f^\sharp$.
For $t = g(\vec{t})$ with $g \in \Sigma_D$, let
$t^\sharp$ denote $g^\sharp(\vec{t})$.

\begin{definition}[(Dependency Pair \cite{AG00})]\label{Dependency Pair}
The set of \emph{dependency pairs} for a TRS $\R$ is $DP(\R) =
\{ \el^\sharp \to t^\sharp \mid \el \to r \in \R,$  $t$ is a subterm of
$r$,  $\rt(t) \in \Sigma_D \}$.
\end{definition}

\begin{example}\label{DP example}
{\sl Consider again the logic program of \rEx{ex_not_sm}
which was transformed into the following TRS $\R$ in  \rEx{ex_not_sm_we}.
\begin{align}
\rrule{32-1}{
{\sf p}_{in}(X,X) }{ {\sf p}_{out}(X,X)}\\
\rrule{32-2}{
{\sf p}_{in}({\sf f}(X), {\sf g}(Y)) }{ {\sf u}_1({\sf p}_{in}({\sf f}(X),
{\sf f}(Z)),X,Y)} \\
\rrule{32-3}{
{\sf u}_1({\sf p}_{out}({\sf f}(X), {\sf f}(Z)),X,Y) }{ {\sf u}_2({\sf p}_{in}(Z,{\sf g}(Y)),X,Y,Z)}\\
\rrule{32-4}{
{\sf u}_2({\sf p}_{out}(Z,{\sf g}(Y)),X,Y,Z) }{ {\sf p}_{out}({\sf f}(X),{\sf g}(Y))}
\end{align}
For this TRS $\R$, we have $\Sigma_D = \{ {\sf
p}_{in}, {\sf u}_1, {\sf u}_2 \}$ and $DP(\R)$ is
\begin{align}
\lrule{42-1}{
{\sf P}_{in}({\sf f}(X), {\sf g}(Y)) }{ {\sf P}_{in}({\sf f}(X), {\sf
f}(Z))}\\
\lrule{42-2}{
{\sf P}_{in}({\sf f}(X), {\sf g}(Y)) }{
{\sf U}_1({\sf p}_{in}({\sf f}(X), {\sf f}(Z)),X,Y)}\\
\lrule{42-3}{
{\sf U}_1({\sf p}_{out}({\sf f}(X), {\sf f}(Z)),X,Y) }{ {\sf P}_{in}(Z,{\sf
g}(Y))}\\
\lrule{42-4}{{\sf U}_1({\sf p}_{out}({\sf f}(X), {\sf f}(Z)),X,Y) }{
{\sf U}_2({\sf p}_{in}(Z,{\sf g}(Y)),X,Y,Z)}
\end{align}}
\end{example}

While \rDef{Dependency Pair} is from \cite{AG00}, all following definitions and
theorems are new. They extend existing concepts from
ordinary to infinitary constructor rewriting.

For termination, one tries to prove that there are no infinite \emph{chains}
of dependency
pairs.  Intuitively, a dependency pair corresponds to a function call and a chain
represents a possible sequence of calls that can occur  during
rewriting. \rDef{Chain} extends the notion of chains to
infinitary constructor rewriting. To this end, we use  an argument filter
$\pi$ that describes which arguments of function
symbols have to be \emph{finite} terms. So if $\pi$ does not delete
arguments (i.e., if
$\pi(f) = \{1, \ldots, n \}$ for all $f/n$), then this corresponds to ordinary
(finitary) constructor rewriting and if $\pi$ deletes  all arguments (i.e., if $\pi(f) = \varnothing$ for
all $f$), then this corresponds to full infinitary constructor rewriting.
In \rDef{Chain}, the TRS
$\D$
usually stands for  a set of dependency pairs.
(Note that if $\R$ is a TRS, then  $DP(\R)$ is also a TRS.)

\begin{definition}[(Chain)]
\label{Chain}
Let $\D,\R$ be TRSs and $\pi$  be
an argument filter.
 A (possibly infinite) sequence of pairs
$s_1\!\to\!t_1, s_2\!\to\!t_2, \ldots$ from $\D$  \mbox{is a
$(\D,\R,\pi)$-\emph{chain}
 iff}
\begin{itemize}
\item[$\bullet$] for all $i \geq 1$,  there are substitutions $\sigma_i : \V \to \T^\infty(\Sigma_C,\V)$ such that
$t_i\sigma_i  \to^*_\R s_{i+1}\sigma_{i+1}$, and
\item[$\bullet$]  for all $i \geq 1$, we have $\pi(s_i\sigma_i),
\pi(t_i\sigma_i)
\in\T(\Sigma_\pi)$. Moreover, if the
rewrite sequence from
$t_i\sigma_i$ to $s_{i+1}\sigma_{i+1}$ has the form
$t_i\sigma_i
 = q_0 \to_\R \ldots \to_\R q_{m} = s_{i+1}\sigma_{i+1}$,
then for all  terms in this
rewrite sequence  we have
$\pi(q_0),\ldots,\pi(q_m)\in\T(\Sigma_\pi)$ as well.
                 So all terms in the sequence have finite ground terms on
                 those positions which are not filtered away by $\pi$.
\end{itemize}
\end{definition}

In \rEx{DP example},
``(\ref{42-2}), (\ref{42-3})'' is a chain for any argument filter $\pi$: if one
instantiates $X$ and $Z$ with the same finite ground term, then
(\ref{42-2})'s instantiated right-hand side rewrites to an instance of (\ref{42-3})'s
left-hand side. Note that if one uses an argument filter
$\pi$ which permits  an instantiation of $X$ and
$Z$ with the infinite
term ${\sf f}({\sf f}(\ldots))$, then there is also an infinite chain ``(\ref{42-2}),
(\ref{42-3}), (\ref{42-2}), (\ref{42-3}),
\ldots''

In order to prove termination of a program $\P$, by \rThm{Soundness of the Transformation}
we have to show that all terms $p_{in}(\vec{t})$ are terminating for
$\R_\P$ whenever $\pi(p_{in}(\vec{t}))$ is a finite ground term and $\vec{t}$
only contains function symbols from the logic program   (i.e., $\vec{t}$
contains no defined
symbols of the TRS $\R_\P$).
\rThm{Proving Infinitary Termination}  states that one can prove
absence of infinite $(DP(\R_\P),\R_\P,\pi')$-chains instead.
Here, $\pi'$ is a filter which filters away ``at least
as much'' as  $\pi$.
However, $\pi'$ has to be chosen in such a way
that  the filtered TRSs  $\pi'(DP(\R_\P))$ and $\pi'(\R_\P)$  satisfy the ``\emph{variable
condition}'', i.e., $\V(\pi'(r)) \subseteq \V(\pi'(\el))$ for all $\el\to r \in
DP(\R_\P) \cup \R_\P$.
Then the filter $\pi'$ detects
all potentially infinite
subterms in rewrite sequences (i.e., all subterms which correspond to
``non-unification-free parts'' of $\P$, i.e., to non-ground subterms when
``executing'' the program $\P$).

\begin{theorem}[(Proving Infinitary Termination)]
\label{Proving Infinitary Termination}
Let $\R$ be a TRS over $\Sigma$ and let $\pi$ be an argument filter over
$\Sigma$. We extend $\pi$ to tuple symbols such that $\pi(F) = \pi(f)$ for all $f \in \Sigma_D$.
Let $\pi'$ be a refinement of $\pi$ such
that $\pi'(DP(\R))$ and $\pi'(\R)$ satisfy the variable
condition.\footnote{To see why the variable
condition  is needed in \rThm{Proving Infinitary Termination}, let $\R =
\{ \Fg(X) \to \Ff(X), \Ff(\Fs(X)) \to \Ff(X) \}$ and  $\pi =
\pi'$ where $\pi'(\Fg) =
\varnothing$,
$\pi'(\Ff) = \pi'({\sf F}) = \pi'(\Fs) = \{1\}$. $\R$'s first rule violates the
variable condition:
$\V(\pi'(\Ff(X))) = \{ X \} \not\subseteq
\V(\pi'(\Fg(X))) = \varnothing$. There is no infinite chain, since $\pi'$ does
not allow us to instantiate the variable $X$ in
the  dependency pair ${\sf F}(\Fs(X)) \to
{\sf F}(X)$ by an infinite term.
Nevertheless, there is a non-terminating term
$\Fg(\Fs(\Fs(\ldots)))$
which
is filtered to a finite ground term $\pi'(\Fg(\Fs(\Fs(\ldots)))) = \Fg$.}
If there is no infinite $(DP(\R),\R,\pi')$-chain, then all terms $f(\vec{t})$ with
$\vec{t} \in \vec{\T}^\infty(\Sigma_C,\V)$ and
$\pi(f(\vec{t})) \in  \T({\Sigma}_\pi)$
are terminating for
$\R$.
\end{theorem}
\begin{proof}
Assume there is a non-terminating term $f(\vec{t})$ as above.
Since $\vec{t}$ does not contain defined symbols,
the first rewrite step in the infinite sequence is on the root position
with a rule $\el = f(\vec{\el}) \to
r$ where $\el\sigma_1 = f(\vec{t})$.  Since $\sigma_1$ does not introduce
defined symbols, all defined symbols of $r\sigma_1$ occur on positions of $r$. So there is
a subterm $r'$ of $r$ with defined root such that $r'\sigma_1$ is also
non-terminating. Let $r'$
denote the smallest such subterm (i.e., for all proper subterms $r''$ of $r'$, the term
$r''\sigma_1$ is terminating).
Then $\el^\sharp \to r'^\sharp$ is the first dependency pair of
the infinite chain that we construct. Note that
$\pi(\el\sigma_1)$  and thus,
$\pi(\el^\sharp\sigma_1)$
and hence, also
$\pi'(\el^\sharp\sigma_1)
 =
\pi'(F(\vec{t}))$ is  a finite ground term by assumption. Moreover, as
$\el^\sharp \to r'^\sharp \in DP(\R)$ and as $\pi'(DP(\R))$ satisfies the variable condition, $\pi'(r'^\sharp\sigma_1)$ is
finite and ground as well.

The infinite sequence continues by rewriting $r'\sigma_1$'s  proper
subterms repeatedly. During this rewriting, the left-hand sides of rules are
instantiated by constructor substitutions (i.e., substitutions
with range $\T^\infty(\Sigma_C,\V)$).
As $\pi'(\R)$ satisfies the variable condition, the terms remain
finite and ground when applying the
filter $\pi'$.  Finally, a root rewrite step is performed again.
Repeating this construction infinitely many times results in
an infinite chain.
\qed
\end{proof}

The following corollary combines \rThm{Soundness of the Transformation} and
\rThm{Proving Infinitary Termination}. It
describes how we use the DP method for proving termination of logic
programs.

\begin{corollary}[(Termination of Logic Prog. by Dependency
Pairs)]\label{Termination of Logic Programs by Dependency Pairs}
\ \\Let $\P$ be a logic program and let $\pi$ be an argument filter
over $(\Sigma,\Delta)$.
We extend $\pi$ to $\Sigma_\P$ and to tuple symbols such that
$\pi(p_{in}) = \pi(P_{in}) = \pi(p)$ for all $p \in \Delta$. For all other
symbols $f/n$ that are not from $\Sigma$ or $\Delta$, we define $\pi(f/n) =
\{1,\ldots,n\}$.
Let $\pi'$ be a refinement of $\pi$ such
that $\pi'(DP(\R_\P))$ and $\pi'(\R_\P)$ satisfy the variable
condition.
If there is no infinite $(DP(\R_\P),\R_\P,\pi')$-chain, then
all queries $Q\in\A^{rat}(\Sigma,\Delta,\V)$ with
$\pi(Q)\in\A(\Sigma_\pi,\Delta_\pi)$ are terminating
for $\P$.
\end{corollary}

\begin{example}\label{choosing AF}{\sl
We want to prove termination of \rEx{ex_not_sm} for all
queries $Q$ where $\pi(Q)$
is finite and ground  for the filter
$\pi(h) = \{ 1 \}$ for all $h \in \{ {\sf p},
{\sf f}, {\sf g} \}$.
By Corollary \ref{Termination of Logic Programs by Dependency Pairs}, it
suffices to show
absence of infinite $(DP(\R),\R,\pi')$-chains. Here, $\R$ is the TRS
$\{ (\ref{32-1}), \ldots, (\ref{32-4}) \}$ 
from \rEx{ex_not_sm_we} and $DP(\R)$ are Rules (\ref{42-1}) -- (\ref{42-4}) from
\rEx{DP example}. The filter
$\pi'$ has to satisfy
$\pi'(h) \subseteq \pi(h) = \{ 1 \}$ for $h \in \{ \Ff, \Fg \}$
and moreover, $\pi'({\sf p}_{in})$ and $\pi'({\sf P}_{in})$ must be subsets of
$\pi({\sf p}_{in}) = \pi({\sf P}_{in}) =
\pi({\sf p}) = \{1\}$.
Moreover, we have to choose $\pi'$ such that the variable condition is
fulfilled.
So while $\pi$ is always given, $\pi'$ has to be determined automatically.
Of course, there are only finitely many
possibilities for $\pi'$. In particular, defining $\pi'(h) =
\varnothing$ for all  symbols $h$ is always possible. But to
obtain a successful  termination proof afterwards, one should try to
generate
filters where
the sets $\pi'(h)$ are as large as possible, since such filters
provide more information
about the finiteness of arguments.  We will present suitable heuristics for
finding such filters $\pi'$ in \rSec{sec:refining}.
In our
example, we use  $\pi'({\sf p}_{in}) = \pi'({\sf P}_{in}) = \pi'({\sf f}) =
\pi'({\sf g}) =\{1\}$,
$\pi'({\sf p}_{out}) =\pi'({\sf u}_1)  =\pi'({\sf
U}_1) =
 \{1,2\}$, and $\pi'({\sf
u}_2)  = \pi'({\sf U}_2) =\{1,2,4\}$. For the non-well-moded \rEx{append} we choose $\pi'({\sf g}) = \varnothing$
instead to satisfy the variable condition.}
\end{example}

\noindent
So to automate the criterion of Corollary \ref{Termination of Logic Programs by Dependency Pairs},
we have to tackle two problems:
\begin{itemize}
\item[(I)] We start with a given filter $\pi$ which describes the set
of queries whose termination
should be proved. Then we
have to find a suitable argument filter $\pi'$ that refines $\pi$
in such a way that the
variable condition of \rThm{Proving Infinitary Termination} is fulfilled and
that the termination proof is
``likely to succeed''. This problem will be discussed in \rSec{sec:refining}.
\item[(II)] For the chosen argument filter $\pi'$,
we have to prove that there is no infinite $(DP(\R_\P),\R_\P,\pi')$-chain.
We show how to do this in the following subsection.
\end{itemize}

\subsection{Automation by Adapting the DP Framework}\label{Automation by Adapting the DP Framework}

Now
we show how to prove absence of infinite $(DP(\R),\R,\pi)$-chains automatically. To
this end, we adapt the \emph{DP framework} of \cite{LPAR04} to infinitary
rewriting. In this framework, we now consider arbitrary
\emph{DP problems} $(\D,\R,\pi)$ where $\D$ and
$\R$
are TRSs and $\pi$ is an argument filter.
Our
goal is to show that there is no infinite $(\D,\R,\pi)$-chain.
In this case, we call the problem \emph{finite}.
Termination techniques should  now be formulated as
\emph{DP processors} which
operate on DP problems instead of TRSs. A DP processor $\Proc$ takes a DP problem as
input and returns a new set
of DP problems which then have to be solved instead.
$\Proc$ is \emph{sound} if for all DP problems $d$, $d$ is
finite whenever all DP problems in $\Proc(d)$ are
finite.
So termination proofs  start with the initial DP problem
$(DP(\R),\R,\pi)$. Then this problem is transformed repeatedly by sound DP
processors. If the final processors return empty sets of DP problems, then termination is
proved.

In \rThm{DP Processor based on the Dependency Graph}, \ref{DP Processors
Based on Reduction Pairs}, and \ref{Argument Filter Processor}
we will recapitulate three of the most important existing DP
processors \cite{LPAR04} and describe how they must be modified for
infinitary constructor rewriting.
To this end, they now also have to take the argument
filter $\pi$ into account.
The first processor uses an  \emph{estimated dependency graph}
 to
estimate which dependency pairs can follow each other in chains.

\begin{definition}[(Estimated Dependency Graph)]
\label{Estimated Dependency Graph}
Let  $(\D,\R,\pi)$ be a DP problem.   The nodes of the
\emph{estimated $(\D,\R,\pi)$-dependency graph}
are the pairs of $\D$ and there is an arc from $s\to t$
to $u \to v$ iff $\CAP(t)$ and a variant $u'$ of $u$ unify with an mgu $\mu$ where
$\pi(\CAP(t)\mu) = \pi(u'\mu)$ is a finite term.
Here, $\CAP(t)$
replaces all subterms of $t$ with defined root symbol by
different fresh variables.
\end{definition}

\begin{example}\label{DP Graph Example}{\sl
For the DP problem $(DP(\R),\R,\pi')$ from
\rEx{choosing AF} we obtain:

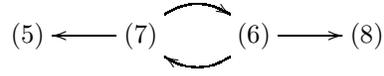
\[ \xymatrix{ (\ref{42-1})\ar@{<-}[r] & (\ref{42-3}) \ar@/^1pc/[r]  &
(\ref{42-2})\ar@{->}[r]\ar@/^1pc/[l]
  &(\ref{42-4})}\]

\vspace*{.2cm}

For example, there is an arc  $(\ref{42-2}) \to (\ref{42-3})$, as $\CAP({\sf U}_1({\sf
p}_{in}({\sf f}(X), {\sf f}(Z)),X,Y)) = {\sf U}_1(V, X,Y)$ unifies with ${\sf U}_1({\sf p}_{out}({\sf f}(X'),
{\sf f}(Z')),X',Y')$  by instantiating the arguments of ${\sf U}_1$
with finite terms. But there
are no arcs $(\ref{42-1}) \to (\ref{42-1})$ or $(\ref{42-1}) \to (\ref{42-2})$, since ${\sf
P}_{in}({\sf f}(X),{\sf f}(Z))$ and ${\sf P}_{in}({\sf f}(X'), {\sf g}(Y'))$ do
not unify, even if one instantiates
$Z$ and $Y'$ by infinite terms (as permitted by
the
filter $\pi'({\sf P}_{in}) =\{1\}$).}
\end{example}

Note that filters are used to \emph{detect} potentially infinite
arguments, but these arguments are not \emph{removed}, since they can still   be useful in
the termination proof. In \rEx{DP Graph Example}, 
they are needed to
determine that 
 (\ref{42-1}) has no outgoing arcs.

If $s \to t,  u \to v$ is a $(\D,\R,\pi)$-chain
then there is an arc from  $s \to t$
to $u \to v$ in the estimated dependency graph.
Thus,  absence of infinite chains can be proved
separately for each maximal strongly connected component (SCC) of the
graph. This observation is used by the following processor to modularize termination proofs
by decomposing a DP
problem into  sub-problems. If there are $n$ SCCs in the graph and if $\D_i$
are the dependency pairs of the $i$-th SCC (for $1 \leq i \leq n$), then one
can decompose the set of dependency pairs $\D$ into the subsets
$\D_1,\ldots,\D_n$. 

\begin{theorem}[(Dependency Graph Processor)]
\label{DP Processor based on the Dependency Graph} For a DP problem
$(\D, \R,$ $\pi)$, let $\Proc$ return
$\{(\D_1,\R,\pi), \ldots, (\D_n,\R,\pi)\}$ where $\D_1, \ldots, \D_n$ are the
sets of nodes of
the SCCs in the estimated dependency graph. Then $\Proc$ is sound.
\end{theorem}
\begin{proof}
We prove that if $s \to t, \,  u \to v$ is a chain,
then there is an arc from $s \to t$ to $u \to v$ in the estimated dependency
graph. This suffices for \rThm{DP Processor based on the Dependency
Graph}, since then
 every infinite $(\D,\R,\pi)$-chain corresponds to an infinite path
in the graph. This path ends in an SCC with nodes $\D_i$ and thus,
there is also an infinite $(\D_i, \R, \pi)$-chain. Hence, if all
$(\D_i,\R,\pi)$ are finite DP problems, then so is $(\D,\R,\pi)$.

Let $s \to t, \,  u \to v$ be a $(\D,\R,\pi)$-chain, i.e.,
$t\sigma_1 \to_\R^* u\sigma_2$ for some constructor substitutions
$\sigma_1,\sigma_2$ where $\pi(t\sigma_1)$ and
$\pi(u\sigma_2)$ are finite.
Let $\mathit{pos}_1,\ldots,\mathit{pos}_n$ be the top positions where $t$ has defined
symbols.
 Then  $\CAP(t) =
t[Y_1]_{\mathit{pos}_1}\ldots[Y_n]_{\mathit{pos}_n}$ for fresh variables $Y_j$.
Moreover, let the variant $u'$ result from $u$ by replacing every $X \in
\V(u)$ by a fresh variable $X'$.
Thus, the  substitution $\sigma$ with
$\sigma(X') = \sigma_2(X)$ for all $X\in \V(u)$, $\sigma(X) = \sigma_1(X)$
for all $X \in \V(t)$, and $\sigma(Y_j) = u\sigma_2|_{\mathit{pos}_j}$ unifies $\CAP(t)$ and $u'$.
So there is also an mgu $\mu$ where $\sigma = \mu\tau$ for some
substitution $\tau$.
Moreover, since $\pi(u\sigma_2) = \pi(u'\sigma)$ is finite,
the term $\pi(u'\mu)$ is finite, too.
Hence, by \rDef{Estimated Dependency Graph}
there is indeed an
arc from $s \to t$ to $u \to v$.
\qed
\end{proof}

\begin{example}\label{DP Graph Example 2}{\sl
In \rEx{DP Graph Example}, the only SCC consists of (\ref{42-2}) and (\ref{42-3}). Thus,
the dependency graph processor transforms the initial DP problem
$(DP(\R),\R,\pi')$
into $(\{(\ref{42-2}),(\ref{42-3})\},\R,\pi')$, i.e., it deletes the dependency pairs
(\ref{42-1}) and (\ref{42-4}).}
\end{example}

The next processor
is based on {\em reduction pairs\/}
$(\succsim,\succ)$ where $\succsim$ and
$\succ$ are relations on finite terms. Here, $\succsim$
is
reflexive, transitive, monotonic (i.e., $s \succsim t$ implies $f(\ldots s
\ldots) \succsim f(\ldots t \ldots)$ for all function symbols $f$), and stable
(i.e., $s \succsim t$ implies $s\sigma
 \succsim t\sigma$ for all substitutions $\sigma$)
and $\succ$ is a stable  well-founded order
compatible with $\succsim$ (i.e., ${\succsim \circ \succ} \subseteq {\succ}$ or ${\succ
\circ \succsim} \subseteq {\succ}$).
There are many techniques  to
search for such relations automatically (recursive path orders,
polynomial interpretations, etc.~\cite{D87}).

For a DP problem $(\D,\R,\pi)$, we now try to find a reduction pair
$(\succsim,\succ)$ such that all filtered
 $\R$-rules are weakly decreasing (w.r.t.\
$\succsim$)
and all filtered $\D$-dependency pairs  are weakly or strictly
decreasing  (w.r.t.\
$\succsim$ or $\succ$).\footnote{We only consider
\emph{filtered} rules and dependency pairs. Thus,
$\succsim$ and $\succ$
are only used to compare those parts of terms
which remain \emph{finite} for all
instantiations in chains.}
Requiring  $\pi(\el) \succsim \pi(r)$ for all $\el \to r \in \R$
ensures that in chains $s_1\to t_1, s_2\to t_2,\ldots$ with $t_i
\sigma_i\to^*_\R s_{i+1}\sigma_{i+1}$ as in \rDef{Chain}, we have $\pi(t_i \sigma_i) \succsim
\pi(s_{i+1} \sigma_{i+1})$.
Hence, if a reduction pair
satisfies the above conditions, then the strictly decreasing dependency pairs
(i.e., those $s \to t \in \D$ where
$\pi(s) \succ \pi(t)$)
cannot occur
infinitely often in chains.
So the following processor
deletes these pairs from
$\D$.
For any TRS $\D$ and any relation $\succ$,
let  $\D_{\succ_\pi} = \{ s \to t \in \D \mid \pi(s)
\succ \pi(t) \}$.

\begin{theorem}[(Reduction Pair Processor)]
\label{DP Processors Based on Reduction Pairs}
\hspace*{.1cm} Let $(\succsim, \succ)$ be a reduction\linebreak  pair.
Then the following DP processor $\Proc$ is
sound.  For
$(\D,\R,\pi)$, $\Proc$ returns
\begin{itemize}
\item[$\bullet$] $\{ (\D \setminus \D_{\succ_\pi}, \R, \pi) \}$, if $\D_{\succ_\pi} \cup
\D_{\succsim_\pi} = \D$ and $\R_{\succsim_\pi} = \R$
\item[$\bullet$] $\{ (\D, \R, \pi) \}$, otherwise
\end{itemize}
\end{theorem}
\begin{proof}
We prove this theorem by contradiction, i.e., we assume that
$(\D,\R,\pi)$ is infinite and then proceed to show that
$(\D \setminus \D_{\succ_\pi}, \R, \pi)$ has to be infinite, too.

From the assumption that $(\D,\R,\pi)$ is infinite, we know that
there is an infinite $(\D,\R,\pi)$-chain
$s_1\to t_1, s_2 \to t_2, \ldots$ with
$t_i\sigma_i \to_\R^* s_{i+1}\sigma_{i+1}$.
For any term $t$ we have $\pi(t\sigma) = \pi(t)\pi(\sigma)$ where
$\pi(\sigma)(x) = \pi(\sigma(x))$ for all $x \in \V$. So by
stability of $\succ $ and $\succsim$, $\D_{\succ_\pi} \cup\D_{\succsim_\pi} = \D$ implies
\begin{equation}
\label{DP decrease}
\pi(s_i\sigma_i) = \pi(s_i)\pi(\sigma_i) \, \sorsuccsim \,
\pi(t_i)\pi(\sigma_i) =\pi(t_i\sigma_i).
\end{equation}
Note that $\pi(s_i \sigma_i)$ and $\pi(t_{i} \sigma_{i})$ are
finite. Thus, comparing them with $\succsim$
is  possible.

\noindent
Similarly,
by the
observation $\pi(t\sigma) = \pi(t)\pi(\sigma)$  we also get that
$t_i\sigma_i \to_\R^* s_{i+1}\sigma_{i+1}$ implies
$\pi(t_i\sigma_i) \to_{\pi(\R)}^* \pi(s_{i+1}\sigma_{i+1})$.
As $\R_{\succsim_\pi} = \R$
means that  $\pi(\R)$'s rules
are decreasing w.r.t.\
$\succsim$, by
monotonicity and stability of $\succsim $ we get
$\pi(t_i\sigma_i) \succsim \pi(s_{i+1}\sigma_{i+1})$.
With (\ref{DP decrease}), this implies
$\pi(s_1\sigma_1) \sorsuccsim \pi(t_1\sigma_1) \succsim \pi(s_2\sigma_2) \sorsuccsim
\pi(t_2\sigma_2) \succsim \ldots\;$
As $\succ$ is compatible with $\succsim$ and well founded,
$\pi(s_i\sigma_i) \succ \pi(t_i\sigma_i)$
only holds for finitely many $i$. So  $s_j \to t_j, s_{j+1} \to t_{j+1},\ldots $ is an infinite
$(\D \setminus \D_{\succ_\pi}, \R, \pi)$ chain for some $j$ and thus, the
DP problem $(\D \setminus \D_{\succ_\pi}, \R, \pi)$ is infinite.
\qed
\end{proof}

\begin{example} {\sl
For the DP problem $(\{(\ref{42-2}),(\ref{42-3})\},\R,\pi')$ in
\rEx{DP Graph Example 2}, one can easily find a reduction pair\footnote{For
example, one can use the polynomial interpretation 
$|{\sf P}_{in}(t_1)| = |{\sf p}_{in}(t_1)| = |{\sf U}_1(t_1,t_2)|
= |{\sf u}_1(t_1,t_2)| = |{\sf u}_2(t_1,t_2,t_3)| = |t_1|$, $|{\sf
p}_{out}(t_1,t_2)| = |t_2|$, $|{\sf f}(t_1)| = |t_1| + 1$, and $|{\sf g}(t_1)| = 0$.} where the dependency pair (\ref{42-3}) is strictly decreasing and where (\ref{42-2}) and
all rules are weakly decreasing after applying the filter $\pi'$:
\[ \begin{array}{rcl@{\hspace*{-1.2cm}}rcl}
{\sf P}_{in}({\sf f}(X)) & \succsim &{\sf U}_1({\sf p}_{in}({\sf f}(X)),X)&{\sf p}_{in}(X) & \succsim & {\sf p}_{out}(X,X)\\
{\sf U}_1({\sf p}_{out}({\sf f}(X), {\sf f}(Z)),X) &\succ& {\sf
P}_{in}(Z)&{\sf p}_{in}({\sf f}(X)) & \succsim &{\sf u}_1({\sf p}_{in}({\sf
f}(X)),X)\\
&&&{\sf u}_1({\sf p}_{out}({\sf f}(X), {\sf f}(Z)),X) &\succsim& {\sf u}_2({\sf
p}_{in}(Z),X,Z)\\
&&&{\sf u}_2({\sf p}_{out}(Z,{\sf g}(Y)),X,Z) &\succsim& {\sf p}_{out}({\sf
f}(X),{\sf g}(Y))
\end{array}\]
Thus, the reduction pair processor can remove (\ref{42-3}) from the DP problem which
results in $(\{(\ref{42-2})\},\R,\pi')$. By applying the dependency
graph processor again, one obtains  the empty set of DP problems, since now the
estimated dependency graph only has the node (\ref{42-2}) and no arcs. This proves that the initial DP problem
$(DP(\R),\R,\pi')$ from
\rEx{choosing AF} is finite and thus,  the logic program from
\rEx{ex_not_sm} terminates for all queries $Q$ where $\pi(Q)$ is finite
and ground. Note that termination of the non-well-moded  program from
\rEx{append} can be shown analogously since
finiteness of the initial DP problem can be proved in the same way. The only
difference is that we obtain {\sf g} instead of ${\sf g}(Y)$ in the
last inequality above.}
\end{example}

As in \rThm{DP Processor based on the Dependency Graph} and \ref{DP Processors Based on Reduction Pairs},
many other existing DP processors \cite{LPAR04} can easily be
adapted   to infinitary
constructor rewriting as well. Finally, one can also use
the
following processor to transform  a
DP problem $(\D,\R,\pi)$ for infinitary constructor rewriting into a DP
problem $(\pi(\D),\pi(R),id)$ for ordinary rewriting. Afterwards,
\emph{any} existing DP
processor for \emph{ordinary}  rewriting becomes applicable.\footnote{If $(\D,\R,\pi)$ results
from the transformation of a logic program,
then for $(\pi(\D),\pi(R),id)$ it is even sound to apply the existing
DP processors for \emph{innermost} rewriting \cite{LPAR04,JAR06}.
These processors are usually more powerful than those for ordinary rewriting.
The framework presented in \cite{LPAR04} even supports constructor rewriting.}
Since any termination technique for TRSs
can immediately be formulated as a DP processor \cite[Thm.\ 36]{LPAR04}, now
any termination technique for ordinary rewriting
can be directly used for infinitary constructor
rewriting as well.

\begin{theorem}[(Argument Filter Processor)]
\label{Argument Filter Processor} \hspace*{.3cm}
Let $\Proc(\,(\D,\R,\pi)\,) =\linebreak
\{(\pi(\D),\pi(\R),id)\}$ where  $id(f) = \{1,\ldots,n\}$
for all $f/n$. Then $\Proc$ is
sound.
\end{theorem}
\begin{proof}
If $s_1 \to t_1, s_2 \to t_2, \ldots$ is an infinite
$(\D,\R, \pi)$-chain  with the substitutions
$\sigma_i$
as in \rDef{Chain},
then $\pi(s_1) \to \pi(t_1), \pi(s_2) \to \pi(t_2), \ldots$ is
an infinite $(\pi(\D), \pi(\R), id)$-chain with the substitutions
$\pi(\sigma_i)$. The reason is that $t_i\sigma_i \to_\R^* s_{i+1}\sigma_{i+1}$ implies
$\pi(t_i)\pi(\sigma_i) = \pi(t_i\sigma_i) \to_{\pi(\R)}^* \pi(s_{i+1}\sigma_{i+1}) =
\pi(s_{i+1})\pi(\sigma_{i+1})$. Moreover,
by  \rDef{Chain},
all terms in the rewrite sequence $\pi(t_i\sigma_i) \to_{\pi(\R)}^* \pi(s_{i+1}\sigma_{i+1})$ are finite.
\qed
\end{proof}

\section{Refining the Argument Filter}
\label{sec:refining}

In
\rSec{sec_translation} we introduced a new transformation
from logic programs $\P$ to TRSs $\R_\P$ and showed that to prove
the termination of a class of queries for $\P$, it is sufficient to
analyze the termination behavior of $\R_\P$.
Our criterion to prove termination of logic programs was summarized in
Corollary \ref{Termination of Logic Programs by Dependency Pairs}.

The transformation itself is trivial to automate
and as shown in \rSec{sec:termination}, existing
systems implementing the DP method
can easily be adapted to
prove termination of infinitary constructor rewriting.
The missing part in the
automation is the generation of a suitable argument filter from
the user input, cf.\ Task (I) at the end of \rSec{Dependency Pairs for Infinitary Rewriting}.
After presenting the general algorithm to refine argument filters in
\rSec{Refinement
Algorithm for Argument Filters}, we introduce suitable heuristics in
Sections \ref{Simple Refinement Heuristics} and \ref{Type-Based Refinement
Heuristic}.
Finally, we extend the general algorithm
for the refinement of argument filters by integrating a mode analysis based
on argument filters
in \rSec{sec:afsmoding}. This allows us to handle logic programs
where a predicate is used with several different modes (i.e., 
where different occurrences of the
same predicate have different
input and output positions). The usefulness of the different heuristics from
Sections \ref{Simple Refinement Heuristics} and \ref{Type-Based Refinement
Heuristic} and the power of our extension in \rSec{sec:afsmoding} will be
evaluated empirically in \rSec{sec:experiments}.

\subsection{Refinement Algorithm for Argument Filters}\label{Refinement
Algorithm for Argument Filters}

In our approach of Corollary \ref{Termination of Logic Programs by Dependency
Pairs}, the user supplies an
initial argument filter $\pi$ to describe the set of queries whose
termination should be proved.
There are two issues with this approach.
First, while argument filters provide the user with a more expressive tool to
characterize classes of queries, termination problems are often rather posed in the
form of a moding function for compatibility reasons.
Fortunately, it is straightforward to extract an appropriate initial argument filter from such
a moding function $m$: we define $\pi(p) = \{i \mid m(p,i) = \mIn\}$ for all $p \in \Delta$
and $\pi(f/n) = \{1,\ldots,n\}$ for all function symbols $f/n \in \Sigma$.

Second, and less trivially, the
variable condition $\V(\pi(r)) \subseteq \V(\pi(\el))$ for all rules $\el\to r \in
DP(\R_\P) \cup \R_\P$ does
not necessarily hold for the argument filter $\pi$.
Thus, a refinement $\pi'$ of $\pi$ must be found such that the
variable condition holds for $\pi'$. Then, our method from Corollary
\ref{Termination of Logic Programs by Dependency Pairs} 
can be applied.

Unfortunately, there are often
many refinements $\pi'$ of a given filter $\pi$ such that the variable condition holds.
The right choice of $\pi'$ is crucial for the success of the termination analysis.
As already mentioned in \rEx{choosing AF},
the argument filter that simply filters away all arguments of all function symbols
in the TRS, i.e., that has $\pi'(f) = \varnothing$ for all $f \in \Sigma_\P$,
is a refinement of every argument filter $\pi$ and it
obviously satisfies the variable condition. But
of course, only termination of trivial logic programs can be shown when
using this refinement $\pi'$.

\begin{example}
\label{ex_filter_choices}
{\sl We consider
the logic program of \rEx{ex_not_sm}. As shown in \rEx{ex_not_sm_we},
the following rule results (among others) from the
translation of the logic program.
\begin{align}
\rrule{32-2}{{\sf p}_{in}({\sf f}(X), {\sf g}(Y)) }{ {\sf u}_1({\sf p}_{in}({\sf
f}(X),{\sf f}(Z)),X,Y)}
\end{align}
Suppose that we want to prove termination of all queries $\Fp(t_1,t_2)$ where
both $t_1$ and $t_2$ are (finite) ground terms. This corresponds to the moding
$m(\Fp,1) = m(\Fp,2) = \mIn$, i.e., to the initial argument filter $\pi$
with $\pi(\Fp) = \{1,2\}$.

In Corollary \ref{Termination of Logic Programs by Dependency
Pairs}, we extend $\pi$ to $\Fp_{in}$ and ${\sf P}_{in}$ by defining it to be
$\{1,2\}$ as well.
In order to prove termination, we now have to find a refinement $\pi'$ of
$\pi$ such that $\pi'(DP(\R_\P))$ and $\pi'(\R_\P)$ satisfy the variable
condition and such that there is no  infinite $(DP(\R_\P),\R_\P,\pi')$-chain.

Let us first try to define $\pi' = \pi$. Then
$\pi'$ does not filter away any arguments. Thus,
$\pi'({\sf p}_{in}) = \{1, 2\}$, $\pi'({\sf u}_1) = \{1, 2, 3\}$, and
$\pi'({\sf f}) = \pi'({\sf g}) = \{1\}$.
But then clearly, the variable condition does not hold as $Z$ occurs in
$\pi'(r)$ but not in $\pi'(\el)$ if $\el \to r$ is Rule (\ref{32-2}) above.

So we have to choose a different refinement $\pi'$.
There remain three choices how we can refine $\pi$ to $\pi'$ in order to
filter away the variable $Z$ in the right-hand side of Rule (\ref{32-2}):
we can filter away the first argument of {\sf f} by defining $\pi'({\sf f}) = \varnothing$,
we can filter away ${\sf p}_{in}\!$'s second argument by defining $\pi({\sf p}_{in}) = \{1\}$,
or we can filter away the first argument of ${\sf u}_1$ by defining $\pi({\sf u}_1) = \{2,3\}$.
}
\end{example}

The decision which of the three choices above should be taken must be done
by a suitable \emph{heuristic}. The following definition gives a formalization
for such heuristics. Here
we assume that the choice only depends on the term $t$ containing a
variable that leads to a violation of the variable condition and on the
position $\mathit{pos}$
of the variable. Then a \emph{refinement heuristic} $\rho$ is a function such
that $\rho(t,\mathit{pos})$ returns a function symbol $f/n$ and an argument
position $i \in \{1,\ldots,n\}$ such that filtering away the $i$-th argument
of $f$ would erase the position $\mathit{pos}$ in the term $t$. For instance,
if $t$ is the right-hand side ${\sf u}_1({\sf p}_{in}({\sf
f}(X),{\sf f}(Z)),X,Y)$ of Rule (\ref{32-2}) and $\mathit{pos}$ is the
position of the variable $Z$ in this term (i.e., $\mathit{pos} = 121$), then
$\rho(t,\mathit{pos})$ can be either $(\Ff,1)$, $(\Fp_{in},2)$, or
$({\sf u}_1,1)$.

\begin{definition}[(Refinement Heuristic)]\label{Refinement Heuristic}
A
\emph{refinement heuristic} is a mapping
  $\rho : \T(\Sigma_\P,\V) \times
\N^* \rightarrow \Sigma_\P \times \N$
such that whenever $\rho(t,\mathit{pos}) = (f,i)$, then there is a position $\mathit{pos}'$ with
$\mathit{pos}' \, i$ being a prefix of $\mathit{pos}$ and
$\rt(t|_{\mathit{pos}'}) = f$.
\end{definition}

Given a TRS $\R_\P$ resulting from the transformation of a
logic program $\P$ and a refinement heuristic $\rho$, \rAlg{alg:refinement}
computes a refinement $\pi'$ of a given argument filter
$\pi$ such that the variable condition holds
for $DP(\R_\P)$ and $\R_\P$.
\begin{algorithm}
\label{alg:refinement}
\KwIn{argument filter $\pi$, refinement heuristic $\rho$, TRS $\R_\P$}
\KwOut{refined argument filter $\pi'$ such that $\pi'(DP(\R_\P))$ and
$\pi'(\R_\P)$ satisfy the variable condition}
\begin{itemize}
\item[\bf 1.] $\pi' := \pi$ \vspace*{.2cm}
\item[\bf 2.] \parbox[t]{10cm}{If there is a rule $\el \to r$ from $DP(\R_\P)  \cup
\R_\P$\\ and a
position $\mathit{pos}$ with $r|_{\mathit{pos}} \in \V(\pi'(r)) \setminus
\V(\pi'(\el))$,  then:} \vspace*{.2cm}
\begin{itemize}
\item[\bf 2.1.] Let $(f,i)$ be the result of $\rho(r,\mathit{pos})$, i.e., $(f,i) :=
\rho(r,\mathit{pos})$. \vspace*{.2cm}
\item[\bf 2.2.]  \parbox[t]{10cm}{Modify $\pi'$ by removing $i$ from $\pi'(f)$, i.e., $\pi'(f) :=
\pi'(f) \setminus \{i\}$.\\For all other symbols from $\Sigma_\P$, $\pi'$
remains unchanged.} \vspace*{.2cm}
\item[\bf 2.3.] Go back to {\bf Step 2}. \vspace*{.2cm}
\end{itemize}
\end{itemize}
\caption{General Refinement Algorithm}
\end{algorithm}

Termination of this algorithm is obvious as $\R_\P$ is finite and
each change of the argument filter in {\bf Step 2.2} reduces the number of unfiltered
arguments. Note also that $\rho(r,\mathit{pos})$ is always defined since $\mathit{pos}$ is never the
top position $\varepsilon$. The reason is that the TRS $\R_\P$ is
non-collapsing (i.e., it
has no right-hand side consisting just of a variable).
The algorithm is correct as it only terminates if the
variable condition holds for every dependency pair and every rule.

Note that if $\pi'(F) = \pi'(f)$ for every defined function symbol $f$
and if we do not filter away the first argument position of the
function symbols $u_{c,i}$,
i.e., $1 \in \pi'(u_{c,i})$, then the satisfaction of the variable
condition for $\R_\P$ implies that the variable condition for $DP(\R_\P)$
holds as well. Thus, for heuristics that guarantee the above properties,
we only have to consider $\R_\P$ in the above algorithm.

\subsection{Simple Refinement Heuristics}\label{Simple Refinement Heuristics}

The following definition introduces
two simple possible refinement
heuristics. If a term $t$ has a position $\mathit{pos}$ with
a
variable  that violates the variable condition, then
these heuristics filter away the respective argument position
of the \emph{innermost} resp.\ the \emph{outermost} function
symbol above the variable.

\begin{definition}[(Innermost/Outermost Refinement Heuristic)]\label{iorh}
Let $t$ be a term and let ``$\mathit{pos} \, i$'' resp.\
``$i \, \mathit{pos}$''
 be a position in $t$.
The \emph{innermost refinement heuristic} $\rho_{im}$ is defined as follows:
\[
\rho_{im}(t, \mathit{pos} \, i) = (\rt(t|_{\mathit{pos}}), i)
\]
The \emph{outermost refinement heuristic} $\rho_{om}$ is defined as follows:
\[
\rho_{om}(t, i \, \mathit{pos}) = (\rt(t), i)
\]
\end{definition}

So if $t$ is again the term ${\sf u}_1({\sf p}_{in}({\sf
f}(X),{\sf f}(Z)),X,Y)$, then the innermost
refinement heuristic would result in
$\rho_{im}(t, 121) = (\Ff,1)$ and the outermost refinement heuristic gives
$\rho_{om}(t, 121) = ({\sf u}_1,1)$.

Both heuristics defined above are simple but problematic, as shown in \rEx{ex:imom}.
Filtering the innermost function symbol often results in the removal of an
argument position that is relevant for termination of another rule.
Filtering
the outermost function symbol excludes the possibility of
filtering the arguments of function symbols from
the signature $\Sigma$ of the original logic program. Moreover, the outermost
heuristic also often removes the first argument of
some $u_{c,i}$-symbol. Afterwards, a successful termination proof is hardly
possible anymore.

\begin{example}
\label{ex:imom}
{\sl Consider again the logic program of \rEx{ex_not_sm}
which was transformed into the following TRS in  \rEx{ex_not_sm_we}.
\begin{align}
\rrule{32-1}{
{\sf p}_{in}(X,X) }{ {\sf p}_{out}(X,X)}\\
\rrule{32-2}{
{\sf p}_{in}({\sf f}(X), {\sf g}(Y)) }{ {\sf u}_1({\sf p}_{in}({\sf f}(X),
{\sf f}(Z)),X,Y)} \\
\rrule{32-3}{
{\sf u}_1({\sf p}_{out}({\sf f}(X), {\sf f}(Z)),X,Y) }{ {\sf u}_2({\sf p}_{in}(Z,{\sf g}(Y)),X,Y,Z)}\\
\rrule{32-4}{
{\sf u}_2({\sf p}_{out}(Z,{\sf g}(Y)),X,Y,Z) }{ {\sf p}_{out}({\sf f}(X),{\sf g}(Y))}
\end{align}
As shown in \rEx{DP example} we obtain the following
dependency pairs for the above rules.
\begin{align}
\rrule{42-1}{
{\sf P}_{in}({\sf f}(X), {\sf g}(Y)) }{ {\sf P}_{in}({\sf f}(X), {\sf
f}(Z))}\\
\rrule{42-2}{
{\sf P}_{in}({\sf f}(X), {\sf g}(Y)) }{
{\sf U}_1({\sf p}_{in}({\sf f}(X), {\sf f}(Z)),X,Y)}\\
\rrule{42-3}{
{\sf U}_1({\sf p}_{out}({\sf f}(X), {\sf f}(Z)),X,Y) }{ {\sf P}_{in}(Z,{\sf
g}(Y))}\\
\rrule{42-4}{{\sf U}_1({\sf p}_{out}({\sf f}(X), {\sf f}(Z)),X,Y) }{
{\sf U}_2({\sf p}_{in}(Z,{\sf g}(Y)),X,Y,Z)}
\end{align}
As in \rEx{ex_filter_choices} we want to prove termination of $\Fp(t_1,t_2)$
for all ground terms $t_1$ and $t_2$. Hence,
we start with the argument filter $\pi$ that does not filter away
any arguments, i.e., $\pi(f/n) = \{1, \ldots, n\}$ for all $f \in \Sigma_\P$.
We will now illustrate \rAlg{alg:refinement}
using our two heuristics.

Using the innermost refinement heuristic $\rho_{im}$ in the algorithm, for the
second DP (\ref{42-2})
we get $\rho_{im}({\sf U}_1({\sf p}_{in}({\sf f}(X), {\sf f}(Z)),X,Y),121) =
({\sf f},1)$. This
requires us to filter away the only argument of ${\sf f}$, i.e., $\pi'({\sf f}) = \varnothing$.
Now $Z$ is contained in the right-hand side of the third DP (\ref{42-3}),
but not in the filtered left-hand side anymore. Thus, we now have to filter away the
first argument of ${\sf P}_{in}$,
i.e., $\pi'({\sf P}_{in}) = \{2\}$.
Due to the DP (\ref{42-2}), we now also
have
to remove the second argument $X$ of ${\sf U}_1$,
i.e., $\pi'({\sf U}_1) =
\{1,3\}$.
Consequently, we lose the information about
finiteness of ${\sf p}$'s first argument and therefore cannot show termination
of the program anymore.
More precisely, there is an infinite $(DP(\R_\P),\R_\P,\pi')$-chain consisting of
the dependency pairs (\ref{42-2}) and (\ref{42-3}) using a substitution that
instantiates the variables $X$ and $Z$  by the infinite
term $\Ff(\Ff(\ldots))$. This is indeed a chain since all infinite terms are
filtered away by the refined argument filter $\pi'$. Hence, the termination
proof fails.

Using the outermost refinement heuristic $\rho_{om}$ instead, for the second
DP (\ref{42-2})
we get $\rho_{om}({\sf U}_1({\sf p}_{in}({\sf f}(X), {\sf f}(Z)),X,Y),121) = ({\sf U}_1,1)$,
i.e., $\pi'({\sf U}_1) = \{2,3\}$.
Considering the third DP (\ref{42-3}) we have to filter away 
the first argument of ${\sf P}_{in}$,
i.e., $\pi'({\sf P}_{in}) = \{2\}$.
Due to the DP (\ref{42-2}), we now also
have
to remove the second argument of ${\sf U}_1$, i.e.,
$\pi'({\sf U}_1) = \{3\}$. So we obtain the same infinite chain
as above since we lose the  information about
finiteness of ${\sf p}$'s first argument. Hence, we again cannot show termination.
}
\end{example}

A slightly improved version of the outermost refinement heuristic
can be achieved by disallowing the filtering of the first arguments of the
symbols
$u_{c,i}$ and $U_{c,i}$.

\begin{definition}[(Improved Outermost Refinement Heuristic)]\label{def:omprime}
Let $t$ be a term and $\mathit{pos}$ be a position in $t$.
The \emph{improved outermost refinement heuristic} $\rho_{om'}$ is defined as:
\[
\rho_{om'}(t, i \, \mathit{pos}) = \begin{cases}
\rho_{om'}(t|_i, \mathit{pos})   & \text{if } i = 1 \text{ and either }
\rt(t) =
u_{c,i} \text{ or } \rt(t) =
U_{c,i}\\
(\rt(t), i)    & \text{otherwise }
\end{cases}
\]
\end{definition}

\begin{example}
\label{ex:omprime}
{\sl Reconsider \rEx{ex:imom}. Using the improved outermost refinement
heuristic, for the second rule (\ref{32-2})
we get $\rho_{om'}({\sf u}_1({\sf p}_{in}({\sf f}(X), {\sf f}(Z)),X,Y),121) =
\rho_{om'}({\sf p}_{in}({\sf f}(X), {\sf f}(Z)),21) =
({\sf p}_{in},2)$
requiring us to filter away the second argument of ${\sf p}_{in}$, i.e.,
$\pi'({\sf p}_{in}) = \{1\}$. 
Consequently, the algorithm filters away the third arguments of both ${\sf u}_1$ and ${\sf u}_2$,
i.e., $\pi'({\sf u}_1) = \{1,2\}$ and $\pi'({\sf u}_2) = \{1,2,4\}$.
Now the variable condition holds for $\R_\P$.
Therefore, by defining
$\pi'({\sf P}_{in}) = \pi'({\sf p}_{in})$,
$\pi'({\sf u}_1) = \pi'({\sf U}_1)$, and $\pi'({\sf u}_2) =\pi'({\sf U}_2)$,
the variable condition also holds for $DP(\R_\P)$. (As mentioned at the end
of
\rSec{Refinement
Algorithm for Argument Filters},
by filtering tuple symbols $F$ in the same way as the original symbols $f$ and
by ensuring $1 \in \pi'(u_{c,i})$, it suffices to check the variable condition
only for the rules $\R_\P$ and not for the dependency pairs $DP(\R_\P)$.)
This argument filter corresponds to the one chosen in
\rEx{choosing AF} and as shown in \rSec{Automation by Adapting the DP
Framework} one can now easily prove termination.
}
\end{example}

\subsection{Type-Based Refinement Heuristic}\label{Type-Based Refinement Heuristic}

The improved outermost heuristic from \rSec{Simple Refinement Heuristics} only
filters symbols of the form $p_{in}$, $p_{out}$, $P_{in}$, and $P_{out}$.
Therefore, the generated argument filters are similar to modings.
However, there are cases  where one needs to filter function symbols from the original logic
program, too.
In this section we show how to obtain a more
powerful refinement heuristic using information from inferred types.

There are many approaches to (direct) termination analysis of logic programs that use type
information in order to guess suitable ``norms'' or ``ranking functions'', e.g.,
\cite{Bossi92,CodishTOPLAS,Decorteetal93,Martin96}. In contrast to most of
these approaches, we do not consider typed logic programs, but untyped ones
and we use types only as a basis for a heuristic to prove termination of the
transformed TRS. To our knowledge, this is the first time that types are
considered in the transformational approach to termination analysis of logic
programs.

\begin{example}\label{append15}
{\sl Now we regard the logic program from  \rEx{append}. The rules after the transformation of
\rDef{unmodedTranslation} are:
\begin{align}
\label{57-1}\rrule{32-1}{ 
{\sf p}_{in}(X,X) }{ {\sf p}_{out}(X,X)}\\
\label{57-2}\rrule{32-2}{ 
{\sf p}_{in}({\sf f}(X), {\sf g}(Y)) }{{\sf u}_1({\sf p}_{in}({\sf f}(X),
{\sf f}(Z)),X,Y)} \\
\lrule{57-3}{
{\sf u}_1({\sf p}_{out}({\sf f}(X), {\sf f}(Z)),X,Y) }{ {\sf u}_2({\sf
p}_{in}(Z,{\sf g}(W)),X,Y,Z)}\\
\lrule{57-4}{
{\sf u}_2({\sf p}_{out}(Z,{\sf g}(W)),X,Y,Z) }{ {\sf p}_{out}({\sf f}(X),{\sf g}(Y))}
\end{align}
Using the improved outermost refinement heuristic $\rho_{om'}$ we start off as
in \rEx{ex:omprime} and obtain
$\pi'({\sf p}_{in}) = \{1\}$, $\pi'({\sf u}_1) = \{1,2\}$, and $\pi'({\sf u}_2) = \{1,2,4\}$.
However, due to the last rule (\ref{57-4}) we now
get $\rho_{om'}({\sf p}_{out}({\sf f}(X),{\sf g}(Y)),21) = ({\sf p}_{out},2)$,
i.e., $\pi'({\sf p}_{out}) = \{1\}$.
Considering the third rule (\ref{57-3}),
we have to filter  ${\sf p}_{in}$ once more  and obtain
$\pi'({\sf p}_{in}) = \varnothing$.
So we again lose the information about
finiteness of ${\sf p}$'s first argument and cannot show
termination. Similar to \rEx{ex:imom}, the innermost refinement heuristic
which filters away the only argument of $\Ff$ also fails for this program.}
\end{example}

So in the example above, neither the innermost nor the (improved)
outermost refinement heuristic succeed. We therefore propose a better
heuristic which is like the innermost refinement heuristic, but
which avoids the filtering of certain arguments of original function symbols from the logic
program.
Close inspection of the cases where filtering such function symbols is
required reveals that it is not advisable to filter away ``reflexive'' arguments.
Here, we call an argument position $i$ \linebreak of a function symbol $f$
\emph{reflexive} (or ``\emph{recursive}''), 
if the arguments on position $i$ have the same ``type'' as
the whole term $f(\ldots)$ itself, cf.\ \cite{Walther94}.
A \emph{type assignment} associates a predicate $p/n$ with an $n$-tuple of
types for its arguments and, similarly, a function $f/n$ with an $(n+1)$-tuple where the last
element specifies the result type of $f$.

\begin{definition}[(Types)]
Let $\Theta$ be a set of types (i.e., a set of names).
A \emph{type assignment} $\tau$ over a signature
$(\Sigma,\Delta)$ and a set of types $\Theta$ is a mapping
$\tau : \Sigma \cup \Delta \rightarrow \Theta^*$ such that
$\tau(p/n) \in \Theta^n$ for all $p/n \in \Delta$ and $\tau(f/n) \in \Theta^{n+1}$ for all
$f/n \in \Sigma$.

Let $f/n \in \Sigma$ be  a function symbol and $\tau$ be a type assignment
with $\tau(f) = (\theta_1,\ldots,\theta_n,\theta_{n+1})$. Then
the \emph{set of reflexive
positions} of $f/n$ is $\mathit{Reflexive}_\tau(f/n) = \{ i \mid
1 \leq i \leq n \mbox{ and }  \theta_i = \theta_{n+1} \}$.
\end{definition}

To infer a suitable type assignment for a logic program,
we use the following simple algorithm. However, since we only use types as a
heuristic to find suitable argument filters, any other type assignment
would also yield a correct method for termination analysis. In other words, the choice of the type
assignment only influences the power of our method, not its soundness.
So unlike~\cite{CodishTOPLAS}, the
correctness of our approach does not depend on the
logic program or the query being well-typed.
More sophisticated type inference algorithms were presented in
\cite{DBLP:conf/sas/BruynoogheGH05,DBLP:conf/sas/CharatonikP98,GallagherPuebla02,Janssens:Bruynooghe,Lu,DBLP:conf/sas/VaucheretB02},
for example.

In our simple type inference algorithm,
we define $\simeq$ as the reflexive and
transitive closure of
the following ``similarity'' relation on the argument positions:
Two argument
positions of (possibly different) function or predicate symbols
are ``similar'' if there exists a program clause such that the argument positions
are occupied by identical variables.
 Moreover, if a term $f(\ldots)$ occurs in the $i$-th position of a function or
predicate symbol $p$, then the argument position of $f$'s result is similar to
the $i$-th argument position of $p$. (For a
function symbol $f/n$ we also consider the argument position $n+1$ which
stands for the result of the function.)
After having computed the relation $\simeq$, we then
use a type assignment which corresponds to the equivalence classes
imposed by $\simeq$. So our simple type inference algorithm
is related to
sharing
analysis~\cite{DBLP:conf/sas/BruynoogheDBDM96,DBLP:journals/jlp/CortesiF99,DBLP:conf/ppdp/LagoonS02},
i.e., the program analysis that aims at detecting program variables that in some
program execution might be bound to terms having a common variable.

\begin{example}\label{append2}
{\sl As an example, we compute a suitable type assignment 
for the logic program from  \rEx{append}:
\[\begin{array}{lll}
{\sf p}(X, X). & &\\
{\sf p}({\sf f}(X), {\sf g}(Y)) &\from& {\sf p}({\sf f}(X), {\sf f}(Z)),
{\sf p}(Z,{\sf g}(W)).
\end{array}\]
Let $\Fp_i$ denote the $i$-th argument position of $\Fp$, etc.
Then due to the first clause we obtain $\Fp_1 \simeq \Fp_2$,
since both argument positions are
occupied by the variable $X$.
Moreover, since $Z$ occurs both in the first
argument positions of $\Ff$ and $\Fp$ in the second clause, we also have
$\Fp_1 \simeq \Ff_1$. Finally, since an $\Ff$-term occurs in the first
and second argument of $\Fp$ and since a $\Fg$-term occurs in the second argument of
$\Fp$ we also have $\Ff_2 \simeq \Fp_1 \simeq \Fp_2$ and
$\Fg_2 \simeq \Fp_2$.  In other words, the relation $\simeq$
imposes the two equivalence classes $\{{\sf p}_1, {\sf p}_2, {\sf
f}_1, \Ff_2, \Fg_2\}$ and $\{{\sf g}_1\}$.
Hence, we compute a type assignment with two types $a$ and $b$ where
$a$ and $b$ correspond to
$\{{\sf p}_1, {\sf p}_2, {\sf
f}_1, \Ff_2, \Fg_2\}$ and $\{{\sf g}_1\}$, respectively.
Thus, the type assignment is defined as $\tau(\Fp) = \tau(\Ff) = (a,a)$ and $\tau(\Fg) = (b,a)$.

Note that the first argument of $\Ff$ has the same type $a$ as its result and
hence, this argument position is reflexive. On the other hand, the first
argument of $\Fg$ has a different type than its result and is therefore not reflexive.
Hence, $\mathit{Reflexive}_\tau(\Ff) = \{ 1 \}$ and
$\mathit{Reflexive}_\tau(\Fg) = \varnothing$.
}
\end{example}

Now we can define the following heuristic based
on type assignments. It is like the innermost refinement heuristic of \rDef{iorh},
but now reflexive arguments of function symbols from $\Sigma$ (i.e., from the
original logic program) are not filtered away.

\begin{definition}[(Type-based Refinement Heuristic)]
\label{def:tb1}
Let $t$ be a term,  let ``$\mathit{pos} \, i$''
 be a position in $t$, and let $\tau$ be a type assignment.
The \emph{type-based refinement heuristic} $\rho^\tau_{\mathit{tb}}$ is defined as follows:
\begin{eqnarray*}
\rho_{\mathit{tb}}^\tau(t,\mathit{pos} \, i) & = & \left\{
\begin{array}{ll}
(\rt(t|_\mathit{pos}),i) & \text{if }\rt(t|_\mathit{pos}) \notin \Sigma\text{ or }i \notin \mathit{Reflexive}_\tau(\rt(t|_\mathit{pos}))\\
\rho_{\mathit{tb}}^{\tau}(t,\mathit{pos}) & \text{otherwise}
\end{array}\right.
\end{eqnarray*}
\end{definition}

Note  that the heuristic $\rho_{\mathit{tb}}^\tau$ never filters away the first
argument of a symbol  $u_{c,i}$ or $U_{c,i}$ from the TRSs $DP(\R_\P)$ and $\R_\P$.
Therefore, as mentioned at the end
of
\rSec{Refinement
Algorithm for Argument Filters}, we only have to check the variable condition
for the rules of $\R_\P$, but not for the dependency pairs.

\begin{example}
{\sl
We continue with the logic program from  \rEx{append} and use the type
assignment computed in  \rEx{append2}
above. The rules after the transformation of
\rDef{unmodedTranslation} are the following, cf.\ \rEx{append15}.
\begin{align}
\rrule{57-1}{
{\sf p}_{in}(X,X) }{ {\sf p}_{out}(X,X)}\\
\rrule{57-2}{
{\sf p}_{in}({\sf f}(X), {\sf g}(Y)) }{{\sf u}_1({\sf p}_{in}({\sf f}(X),
{\sf f}(Z)),X,Y)} \\
\rrule{57-3}{
{\sf u}_1({\sf p}_{out}({\sf f}(X), {\sf f}(Z)),X,Y) }{ {\sf u}_2({\sf
p}_{in}(Z,{\sf g}(W)),X,Y,Z)}\\
\rrule{57-4}{
{\sf u}_2({\sf p}_{out}(Z,{\sf g}(W)),X,Y,Z) }{ {\sf p}_{out}({\sf f}(X),{\sf g}(Y))}
\end{align}
Due to the occurrence of $Z$ in the right-hand side of the second rule
(\ref{57-2}), we compute:
\[\begin{array}{cll}
  & \rho_{\mathit{tb}}^\tau({\sf u}_1({\sf p}_{in}({\sf f}(X),{\sf f}(Z)),X,Y), 121)\\
= & \rho_{\mathit{tb}}^\tau({\sf u}_1({\sf p}_{in}({\sf f}(X),{\sf f}(Z)),X,Y),
12)&\text{as } \Ff \in\Sigma \text{ and } 1 \in \mathit{Reflexive}_\tau(\Ff)\\
= & ({\sf p}_{in}, 2)&\text{as }{\sf p}_{in} \not\in \Sigma
\end{array}\]
Thus, we filter away the second argument of ${\sf p}_{in}$, i.e., $\pi'({\sf p}_{in}) = \{1\}$.
Consequently, we obtain $\pi'({\sf u}_1) = \{1,2\}$ and $\pi'({\sf u}_2) = \{1,2,4\}$.

Considering the fourth rule (\ref{57-4}) we compute:
\[\begin{array}{cll}
  & \rho_{\mathit{tb}}^\tau({\sf p}_{out}({\sf f}(X),{\sf g}(Y)), 21)\\
= & ({\sf g}, 1)&\text{as }1 \not\in \mathit{Reflexive}_\tau({\sf g})
\end{array}\]
Thus, we filter away the only argument of ${\sf g}$, i.e., $\pi'({\sf g}) = \varnothing$.
By filtering the tuple symbols in the same way as the corresponding
``lower-case'' symbols, now the
variable condition holds for $\R_\P$ and therefore also for
$DP(\R_\P)$. Indeed, this is the argument filter chosen  in
\rEx{choosing AF}. With this filter, one can easily prove termination of the
program, cf.\ \rSec{Automation by Adapting the DP
Framework}.
}
\end{example}

For the above example, it is sufficient only to avoid the filtering of reflexive
positions. However, in general one should also avoid the filtering of all ``unbounded''
argument positions. An argument position of type $\theta$ is ``unbounded'' if it may contain subterms
from a recursive data structure, i.e., if there
 exist
infinitely many terms of type $\theta$. The decrease of the terms on such argument positions
might be the reason for the termination of the program and therefore, they should not be
filtered away. 	To formalize the concept of unbounded argument positions, we define the set of
\emph{constructors} of a type $\theta$ to consist of all function symbols whose result has
type $\theta$. Then an argument position of a function symbol $f$ is \emph{unbounded} if it
is reflexive or if it has a type $\theta$ with a constructor that has an unbounded
argument position. For the sake of brevity, we also speak of just \emph{unbounded positions}
when referring to unbounded argument positions.

\begin{definition}[(Unbounded Positions)]
Let  $\theta \in \Theta$ be a type and $\tau$ be a type assignment.
A function symbol $f/n$ with  $\tau(f/n) =
(\theta_1,\ldots,\theta_n,\theta_{n+1})$ is a \emph{constructor} of $\theta$ iff
$\theta_{n+1} = \theta$. Let $\mathit{Constructors}_\tau(\theta)$ be the set of all
constructors of $\theta$.

For a function symbol $f/n$ as above, we define the \emph{set of unbounded
positions} as the smallest set such that $\mathit{Reflexive}_\tau(f/n) \subseteq
\mathit{Unbounded}_\tau(f/n)$ and such that
$i \in \mathit{Unbounded}_\tau(f/n)$ if there is a $g/m \in Constructors_\tau(\theta_i)$ and
a $1 \leq j \leq m$
with $j \in \mathit{Unbounded}_\tau(g/m)$.
\end{definition}

In the logic program from Examples \ref{append} and \ref{append2},
we had  $\tau(\Fp) = \tau(\Ff) = (a,a)$ and $\tau(\Fg) = (b,a)$. Thus,
$\mathit{Constructors}_\tau(a) = \{ \Ff,\Fg\}$ and $\mathit{Constructors}_\tau(b) =
\varnothing$. Since the first argument position of $\Ff$ is reflexive, it is also unbounded. The first
argument position of $\Fg$ is not unbounded, since it is not reflexive and there is no constructor of
type $b$ with an unbounded argument position. So in this example, there is no difference
between reflexive and unbounded
positions.

However, we will show in \rEx{inversion} that there are programs where these two notions
differ. For that reason, we now improve our type-based refinement heuristic and disallow the
filtering of unbounded (instead of just reflexive) positions.

\begin{definition}[(Improved Type-based Refinement Heuristic)]
\label{def:tb2}
Let $t$ be a term,  let ``$\mathit{pos} \, i$''
 be a position in $t$, and let $\tau$ be a type assignment.
The \emph{improved type-based refinement heuristic} $\rho^\tau_{\mathit{tb}'}$ is defined as follows:
\begin{eqnarray*}
\rho_{\mathit{tb}'}^\tau(t,\mathit{pos} \, i) & = & \left\{
\begin{array}{ll}
(\rt(t|_\mathit{pos}),i) & \text{if }\rt(t|_\mathit{pos}) \notin \Sigma\text{ or }i \notin \mathit{Unbounded}_\tau(\rt(t|_\mathit{pos}))\\
\rho_{\mathit{tb}'}^{\tau}(t,\mathit{pos}) & \text{otherwise}
\end{array}\right.
\end{eqnarray*}
\end{definition}

\begin{example}\label{inversion}
{\sl The following logic program inverts
an integer represented by a sign ({\sf neg} or {\sf pos}) and by
a natural number in Peano notation (using {\sf s} and {\sf 0}). So the integer number $1$ is
represented by the term ${\sf pos}({\sf s}({\sf 0}))$, the integer number
$-1$ is represented by ${\sf neg}({\sf s}({\sf 0}))$, and the integer number $0$ has the two
representations ${\sf pos}(\Fz)$ and ${\sf neg}(\Fz)$.
Here ${\sf nat}(t)$ holds iff $t$ represents a natural number (i.e., if $t$ is
a term containing just ${\sf s}$ and ${\sf 0}$) and ${\sf inv}$ simply
exchanges the function symbols ${\sf neg}$ and ${\sf pos}$. The main predicate
${\sf safeinv}$ performs the desired inversion where ${\sf safeinv}(t_1,t_2)$
only holds if $t_1$ really represents an integer number and $t_2$ is its inversion.

\[\begin{array}{lll}
{\sf nat}({\sf 0}). & &\\
{\sf nat}({\sf s}(X)) &\from& {\sf nat}(X).\\
{\sf inv}({\sf neg}(X), {\sf pos}(X)). & &\\
{\sf inv}({\sf pos}(X), {\sf neg}(X)). & &\\
{\sf safeinv}(X,{\sf neg}(Y)) &\from& {\sf inv}(X, {\sf neg}(Y)), {\sf nat}(Y).\\
{\sf safeinv}(X,{\sf pos}(Y)) &\from& {\sf inv}(X, {\sf pos}(Y)), {\sf nat}(Y).
\end{array}\]

\vspace*{.2cm}

\noindent
The rules after the transformation of \rDef{unmodedTranslation} are:
\begin{align}
\lrule{513-1}{
{\sf nat}_{in}({\sf 0}) }{ {\sf nat}_{out}({\sf 0})} \\
\lrule{513-2}{
{\sf nat}_{in}({\sf s}(X)) }{ {\sf u}_1({\sf nat}_{in}(X), X)} \\
\lrule{513-3}{
{\sf u}_1({\sf nat}_{out}(X), X) }{ {\sf nat}_{out}({\sf s}(X))} \\
\lrule{513-4}{
{\sf inv}_{in}({\sf neg}(X), {\sf pos}(X)) }{ {\sf inv}_{out}({\sf neg}(X), {\sf
pos}(X))}\\
\lrule{513-5}{
{\sf inv}_{in}({\sf pos}(X), {\sf neg}(X)) }{ {\sf inv}_{out}({\sf pos}(X), {\sf
neg}(X))}\\
\lrule{513-6}{
{\sf safeinv}_{in}(X, {\sf neg}(Y)) }{ {\sf u}_2({\sf inv}_{in}(X, {\sf neg}(Y)), X,
Y)} \\
\lrule{513-7}{
{\sf u}_2({\sf inv}_{out}(X, {\sf neg}(Y)), X, Y) }{ {\sf u}_3({\sf nat}_{in}(Y), X,
Y)} \\
\lrule{513-8}{
{\sf u}_3({\sf nat}_{out}(Y), X, Y) }{ {\sf safeinv}_{out}(X, {\sf neg}(Y))} \\
\lrule{513-9}{
{\sf safeinv}_{in}(X, {\sf pos}(Y)) }{ {\sf u}_4({\sf inv}_{in}(X, {\sf pos}(Y)), X,
Y)} \\
\lrule{513-10}{
{\sf u}_4({\sf inv}_{out}(X, {\sf pos}(Y)), X, Y) }{ {\sf u}_5({\sf nat}_{in}(Y), X,
Y)} \\
\lrule{513-11}{
{\sf u}_5({\sf nat}_{out}(Y), X, Y) }{ {\sf safeinv}_{out}(X, {\sf pos}(Y))}
\end{align}

Let us assume that the user wants to prove termination of all queries ${\sf
safeinv}(t_1,t_2)$ where $t_1$ is  ground. So we use the moding $m({\sf safeinv},1) = \mIn$ and
$m({\sf safeinv},2) = \mOut$.
Thus, as initial argument filter $\pi$ we have
$\pi({\sf safeinv}) = \{1\}$ and hence  $\pi({\sf safeinv}_{in}) =\pi({\sf SAFEINV}_{in}) =
\{1\}$, while
$\pi(f/n) = \{1,\ldots,n\}$ for all $f \notin \{{\sf safeinv},\linebreak 
{\sf safeinv}_{in}, {\sf SAFEINV}_{in}\}$.
In  Rule (\ref{513-6})
one has to filter away the second argument
of ${\sf inv}_{in}$ or the only argument of ${\sf neg}$
in order to remove the ``extra'' variable $Y$ on the right-hand side. From a type inference
for these rules we obtain the type assignment $\tau$
with $\tau({\sf s}) = (b,b)$, $\tau({\sf 0}) = (b)$, and  $\tau({\sf neg}) = \tau({\sf pos})
= (b,a)$. So ``$a$'' corresponds to the type of integers and ``$b$'' corresponds to
the type of naturals. The constructors of the naturals are 
$\mathit{Constructors}_\tau(b) = \{ {\sf s}, {\sf 0}\}$. This is a recursive
data structure since ${\sf s}$ has an unbounded argument: $1 \in
\mathit{Reflexive}_\tau({\sf s}) \subseteq  \mathit{Unbounded}_\tau({\sf
s})$. Thus, while ${\sf neg}$'s first argument position of type $b$
is not reflexive, it is still unbounded, i.e.,
$1 \in \mathit{Unbounded}_\tau({\sf neg})$. Hence, our improved type-based heuristic decides
to filter away the second argument of ${\sf inv}_{in}$
(as ${\sf inv}_{in}$  is not from the original signature $\Sigma$). Now
termination is easy to show.

If one had considered the original type-based heuristic instead, then the non-reflexive first
argument of ${\sf neg}$ would be filtered away. Due to Rule (\ref{513-6}),
then also the last
argument of ${\sf u}_2$ has to be removed by the filter. But then the variable
$Y$ would not occur anymore in the filtered 
left-hand side of
Rule (\ref{513-7}). So
to satisfy the variable condition for Rule (\ref{513-7}),
we would have to filter away the only argument of ${\sf nat}_{in}$. Similarly, the only
argument of the corresponding tuple symbol ${\sf NAT}_{in}$ would also be filtered away,
blocking any possibility for a successful termination proof.}
\end{example}

\subsection{Mode Analysis based on Argument Filters and an Improved Refinement
Algorithm}
\label{sec:afsmoding}

In
logic programming, it is not unusual that a predicate is used with
different modes (i.e., that different occurrences of the predicate have
different input and output positions).
Uniqueness
of moding can then be achieved by creating appropriate copies
of these predicate symbols and their clauses for 
every different moding.

\begin{example}\label{ex:nonunique}
{\sl Consider the following logic program for rotating a list taken from \cite{Codish:examples:URL}.
Let $\P$ be the {\sf append}-program consisting of the clauses from \rEx{real_append}
and the new clause
\begin{equation}
\label{nonunique-clause}
{\sf rotate}(\mathit{N}, \mathit{O}) \from {\sf append}(\mathit{L}, \mathit{M}, \mathit{N}),
{\sf append}(\mathit{M}, \mathit{L}, \mathit{O}). 
\end{equation}
with the moding $m({\sf rotate},1) = \mIn$ and $m({\sf rotate},2) = \mOut$.
For this moding, the program
is
terminating.

But while the first use of {\sf append} in Clause (\ref{nonunique-clause})
supplies it with a ground term only
on the last argument position, the second use  in 
(\ref{nonunique-clause}) is with ground terms only on
the first two argument positions. Although the {\sf append}-clauses are even
well moded for both kinds of uses, the whole program is not.

The logic program is transformed 
into the following TRS. As before, ``$[X|L]$'' is an abbreviation for
$\point(X,L)$, i.e., $\point$ is the constructor for list insertion.
\begin{align}
\lrule{spl1}{ {\sf append}_{in}([\,], M,M) }{ \Fappend_{out}([\,],
M,M)} \\
\lrule{spl2}{ {\sf append}_{in}(\point(X,L), M, \point(X,N)) }{
\Fu_1({\sf append}_{in}(L,M,N),X,L,M,N)} \\
\lrule{spl3}{ \Fu_1({\sf append}_{out}(L,M,N),X,L,M,N)}{
{\sf append}_{out}(\point(X,L), M, \point(X,N))} \\
\lrule{spl4}{ {\sf rotate}_{in}(N,O) }{ \Fu_2({\sf append}_{in}(L,M,N),N,O)} \\
\lrule{spl5}{ \Fu_2({\sf append}_{out}(L,M,N),N,O)}{
\Fu_3(\Fappend_{in}(M,L,O),L,M,N,O)} \\
\lrule{spl6}{ \Fu_3(\Fappend_{out}(M,L,O),L,M,N,O) }{ {\sf rotate}_{out}(N,O)}
\end{align}
Due to the ``extra'' variables $L$ and $M$ in the right-hand side of Rule
(\ref{spl4}) and the ``extra'' variable $O$ in the right-hand side of Rule
(\ref{spl5}),\footnote{In the left-hand side of Rule (\ref{spl4}), the variable
$O$ in the second argument of ${\sf rotate}_{in}$ is removed by the initial
filter that describes the desired set of queries given by the
user. Consequently, one also has to filter away the last 
argument of $\Fu_2$. Hence, then $O$ is indeed an ``extra'' variable 
in the right-hand side of Rule
(\ref{spl5}).}
the only refined argument
filter which would satisfy the variable condition of 
Corollary \ref{Termination of Logic Programs by Dependency Pairs} 
is the one where $\pi({\sf append}_{in}) =
\varnothing$.\footnote{Alternatively, one could also filter away the first
arguments of $\Fu_2$ and $\Fu_3$. But then one would also have to satisfy the
variable condition for the dependency pairs and one would obtain $\pi({\sf
APPEND}_{in}) = \varnothing$. Hence, the termination proof attempt would fail as well.}
As we can expect, for the queries described by this filter, the {\sf
append}-program 
is not terminating
and, thus, our new approach fails, too.

The common solution~\cite{Apt:Book}
is to produce two copies of the {\sf append}-clauses
and to rename them apart. This is often referred to as ``mode-splitting''.
First, we create labelled copies of the predicate symbol ${\sf append}$ and label the
predicate of each ${\sf append}$-atom by the input positions of the moding in which it is used.
Then, we extend our moding to $m({\sf append}^{\{3\}},3) = m({\sf
append}^{\{1,2\}},1)
 = m({\sf
append}^{\{1,2\}},2) = \mIn$
and $m({\sf append}^{\{3\}},1) = m({\sf append}^{\{3\}},2) =  m({\sf
append}^{\{1,2\}},3) =  \mOut$.
In our example, termination of the resulting logic program can easily be shown
using both the classical transformation from 
\rSec{The Classical Transformation} or our new transformation:
\[\begin{array}{lll}
{\sf rotate}(\mathit{N}, \mathit{O}) & \from & {\sf append}^{\{3\}}(\mathit{L}, \mathit{M}, \mathit{N}), {\sf append}^{\{1,2\}}(\mathit{M}, \mathit{L}, \mathit{O}).\\
{\sf append}^{\{3\}}([\,], \mathit{M}, \mathit{M}).\\
{\sf append}^{\{3\}}([X|\mathit{L}], \mathit{M}, [X|\mathit{N}]) & \from &  {\sf append}^{\{3\}}(\mathit{L},\mathit{M},\mathit{N}).\\
{\sf append}^{\{1,2\}}([\,], \mathit{M}, \mathit{M}).\\
{\sf append}^{\{1,2\}}([X|\mathit{L}], \mathit{L}, [X|\mathit{N}]) & \from & {\sf append}^{\{1,2\}}(\mathit{L},\mathit{M},\mathit{N}).
\end{array}\]}
\end{example}

In the example above, a pre-processing based on modings was sufficient
for a successful termination proof.
In general, though, this is insufficient to handle queries described
by an argument filter.
The following example demonstrates this.

\begin{example}\label{ex:nonuniqueafs}
{\sl Consider again the logic program $\P$ from \rEx{ex:nonunique} which is
translated to the TRS $\R_\P = \{ (\ref{spl1}), \ldots,  (\ref{spl6})\}$. This time
we want to show termination  for all queries of the form
${\sf rotate}(t_1, t_2)$ where $t_1$
is a finite list (possibly containing non-ground terms as elements). So $t_1$
is instantiated by terms of the form $\point(r_1, \point(r_2, \ldots
\point(r_n, [\,])\ldots))$ where the $r_i$ can be arbitrary terms possibly
containing variables.\footnote{Such a termination problem can also result from
an initial termination problem that was described by modings. To demonstrate this, we 
could extend the program by the following clauses.
\[
\begin{array}{lll}
\Fp(X,O) & \from & \Fstol(X,N), \Frotate(N,O).\\
\Fstol(0,[\,]).\\
\Fstol(s(X),[Y|N]) & \from & \Fstol(X,N).
\end{array}
\]
To prove termination of all queries described by the moding $m(\Fp,1)$ =
\mIn{} and $m(\Fp,2)$ = \mOut, one essentially has to show
termination
for  all queries of the form ${\sf rotate}(t_1, t_2)$ where $t_1$
is a finite list.}
 
To specify these queries, the user would provide
the initial argument filter $\pi$ with $\pi({\sf rotate}) = \{1\}$
and $\pi(\point) = \{2\}$. Now our aim is to prove termination of all queries that
are ground under the filter $\pi$. Thus,
the first argument of $\Frotate$ is not necessarily a ground term (it is only
guaranteed to be ground \emph{after filtering
away the second argument of $\point$}). 

Therefore, if one wanted to
pre-process the program using modings, then one could not assume
that the first
argument of $\Frotate$ were ground. Instead, one would have to use the moding
$m(\Frotate, 1) = m(\Frotate, 2)$ = \mOut. Therefore, 
in the calls to ${\sf append}$, all argument positions would be considered as
``\mOut''.
As a consequence, no renamed-apart copies of clauses
would be created and the termination proof would fail.}
\end{example}

In general, our refinement algorithm from \rSec{Refinement
Algorithm for Argument Filters}
(\rAlg{alg:refinement})
aims to compute an argument filter that filters away as few arguments as
possible while ensuring that the variable condition holds.
In this way we make sure that the maximal amount of information remains
for the following termination analysis.

But as Examples \ref{ex:nonunique} and \ref{ex:nonuniqueafs} above demonstrate,
there are cases where we need to
create renamed-apart copies of clauses for certain predicates in order to
obtain a viable refined argument filter. 
To this end, a first idea might be to combine an existing mode inference algorithm with
\rAlg{alg:refinement}. However,
it is not clear how to do such a combination. The problem is that 
we already need to know the refined argument
filter in order
to create
suitable copies of clauses. 
At the same time,  we already need the
renamed-apart copies of the clauses in order to compute the refined
argument filter. Thus, we have a classical
``chicken-and-egg'' problem. Moreover, such an approach would always fail for
programs like \rEx{ex:nonuniqueafs} where there exists no suitable
pre-processing based on modings.

Therefore, we replace \rAlg{alg:refinement} 
by the following new \rAlg{alg:ImprovedRefinement}
that simultaneously refines the
argument filter and creates renamed-apart copies on demand.

\begin{algorithm}
\label{alg:ImprovedRefinement}
\KwIn{argument filter $\pi$, refinement heuristic $\rho$, TRS $\R_\P$}
\KwOut{\parbox[t]{10cm}{refined argument filter $\pi'$ and modified TRS $\R_\P'$\\
such that 
$\pi'(\R_\P')$ satisfies the variable condition}}
\begin{itemize}
\item[\bf 1.] $\R_\P' := \R_\P \cup \{ \ell^{\pi(p)} \to
r^{\pi(p)} \mid \ell \to r \in \R_\P(p),  \; p/n \in \Delta, \; \pi(p) \subsetneq \{1,\ldots,n\}\}$
 \vspace*{.2cm}
\item[\bf 2.] $\pi'(f) := 
\left\{ \begin{array}{ll}
\pi(f),& \mbox{for all $f \in \Sigma$ (i.e., for functions of
 $\P$)}\\
I,& \mbox{for all $f = p_{in}^I$ with $p \in \Delta$}\\
\{1,\ldots,n\},& \mbox{for all other symbols $f/n$}
\end{array} \right.$
\vspace*{.2cm}
\item[\bf 3.] \parbox[t]{10cm}{If there is a rule $\el \to r$ from $\R_\P'$\\ and a
position $\mathit{pos}$ with $r|_{\mathit{pos}} \in \V(\pi'(r)) \setminus
\V(\pi'(\el))$,  then:} \vspace*{.2cm}
\begin{itemize}
\item[\bf 3.1.] Let $(f,i)$ be the result of $\rho(r,\mathit{pos})$, i.e., $(f,i) :=
\rho(r,\mathit{pos})$. \vspace*{.2cm}
\item[\bf 3.2.]  \parbox[t]{10cm}{We perform a case analysis depending on whether
$f$ has the form $p_{in}^I$ for some $p \in \Delta$. Here, unlabelled 
symbols of the form $p_{in}/n$ are treated as if they were labelled
with $I = \{1,\ldots,n\}$. \vspace*{.2cm}
\begin{itemize}
\item[$\bullet$]
If $f = p_{in}^I$, then we must have $r =
u(p_{in}^I(...),\ldots)$
for some symbol $u$. We
introduce a new function symbol $p_{in}^{I \setminus \{ i \}}$ with $\pi'(p_{in}^{I \setminus \{ i \}}) = I \setminus \{ i \}$ if it has
not yet been introduced. 
Then:\vspace*{.1cm}
\begin{itemize}
\item[$\circ$]  We replace $p_{in}^I$ by $p_{in}^{I \setminus \{ i \}}$ in the right-hand
side of 
$\el \to r$:
\[\R_\P' := \R_\P' \setminus \{
\ell \to r \} \cup \{ \ell \to \overline{r} \},\] where
$\overline{r} = u(p_{in}^{I
\setminus \{ i \}}(...),\ldots)$. \vspace*{.2cm}
\item[$\circ$] 
$\R_\P' := \R_\P' \cup \{ s^{I\setminus\{i\}} \to t^{I\setminus\{i\}}  \mid
s \to t \in \R_\P'(p) 
\}$.\\
If this introduces new labelled function symbols $f/n$ where $\pi'$ was not
yet defined on, we define $\pi'(f) = \{1,\ldots,n\}$. \vspace*{.2cm}
\item[$\circ$] 
Let $\ell' \to r'$ be the rule in $\R_\P'$ with 
$\ell' = u(p_{out}^I(...),\ldots)$. 
We now replace $p_{out}^I$ by $p_{out}^{I \setminus \{ i \}}$ in the left-hand
side of $\ell' \to r'$:
\[\R_\P' := \R_\P' \setminus \{
\ell' \to r' \} \cup \{ \overline{\ell'} \to r' \},\]
where $\overline{\ell'}
= u(p_{out}^{I\setminus\{i\}}(...),\ldots)$. \vspace*{.2cm}
\end{itemize}
\item[$\bullet$] Otherwise (i.e., if $f$ does not have the form $p_{in}$ or $p_{in}^I$), 
 then
modify $\pi'$ by removing $i$ from $\pi'(f)$, i.e., $\pi'(f) :=
\pi'(f) \setminus \{i\}$. 
\end{itemize}
} \vspace*{.2cm}
\item[\bf 3.3.] Go back to {\bf Step 3}. \vspace*{.2cm}
\end{itemize}
\end{itemize}
\caption{Improved Refinement Algorithm}
\end{algorithm}

The idea of the algorithm is the following. Whenever our refinement heuristic
suggests to filter away an argument of a symbol $p_{in}$, then instead of
changing the argument filter appropriately, we introduce a new copy of the
symbol $p_{in}$. To distinguish the different copies of the symbols
$p_{in}$, we label them by the argument positions that are not filtered
away. 

In general, a removal of argument positions of $p_{in}$ can already be performed by
the initial filter $\pi$ that the user provides in order to describe the
desired set of queries. Therefore, if $\pi(p)$ does not contain all
arguments $\{1,\ldots,n\}$ for some predicate symbol $p/n$, then we already
introduce a new symbol $p_{in}^{\pi(p)}$ and new copies of the
rewrite rules originating from $p$. In these rules, we use the new 
symbol $p_{in}^{\pi(p)}$ instead of $p_{in}$. 

Let us reconsider \rEx{ex:nonuniqueafs}. To prove termination of all queries
${\sf rotate}(t_1, t_2)$ with a finite list $t_1$, the 
user would select
the argument filter $\pi$ that eliminates the second argument of {\sf rotate}
and the first argument of the list constructor \point. So we have $\pi({\sf rotate}) = \{1\}$,
$\pi(\point) = \{2\}$,  and $\pi(\Fappend) = \{1,2,3\}$.
Then in addition to the rules 
(\ref{spl4}) - (\ref{spl6})
for the symbol $\Frotate_{in}$ we also
introduce the symbol $\Frotate_{in}^{\{1\}}$. Moreover, in order to ensure
that $\Frotate_{in}^{\{1\}}$ does the same computation as 
$\Frotate_{in}$, we add the following copies of the
rewrite rules (\ref{spl4}) - (\ref{spl6}) originating from the predicate $\Frotate$.
Here, all root symbols of left- and right-hand sides are labelled with $\{1\}$.
\begin{align}
\lrule{lab4}{ {\sf rotate}_{in}^{\{1\}}(N,O) }{ \Fu_2^{\{1\}}({\sf append}_{in}(L,M,N),N,O)}\\
\lrule{lab5}{ \Fu_2^{\{1\}}({\sf append}_{out}(L,M,N),N,O)}{ \Fu_3^{\{1\}}(\Fappend_{in}(M,L,O),L,M,N,O)} \\
\lrule{lab6}{ \hspace{-0.5em}\Fu_3^{\{1\}}(\Fappend_{out}(M,L,O),L,M,N,O) }{ {\sf rotate}_{out}^{\{1\}}(N,O)}
\end{align}

So in \textbf{Step 1} of the algorithm, we initialize $\R_\P'$ to contain all rules of $\R_\P$.
But in addition, $\R_\P'$ contains
labelled copies
of the rules resulting from those predicates $p/n$ where $\pi(p) \subsetneq  
\{1,\ldots,n\}$. In these rules, the root symbols of left- and right-hand
sides are labelled with $\pi(p)$.

Formally,  for every predicate symbol $p \in \Delta$, let $\R_\P(p)$
denote those rules of $\R_\P$ which result from $p$-clauses (i.e., from
clauses whose head is built with the predicate $p$). So $\R_\P({\Frotate})$
consists of the rule for $\Frotate_{in}$ and the rules for $\Fu_2$ and $\Fu_3$,
i.e., $\R_\P({\Frotate}) = \{ (\ref{spl4}), (\ref{spl5}),  (\ref{spl6})\}$.

Then for a term $t = f(t_1,\ldots,t_n)$ and a set
of argument positions $I \subseteq \mathbb{N}$, let $t^I$ denote $f^I(t_1,\ldots,t_n)$. So for
$t = \Frotate_{in}(N,O)$ and $I = \{1\}$, we have $t^I =
\Frotate_{in}^{\{1\}}(N,O)$. 
Hence if $\pi(\Frotate) = \{1\}$, then we extend
$\R_\P'$ by copies of
the rules in $\R_\P({\Frotate})$ where the root symbols are labelled by
$\{1\}$. In other words, we have to add the rules  $\{ \ell^{\pi(p)} \to
r^{\pi(p)} \mid \ell \to r \in \R_\P({\Frotate}) \} = \{ (\ref{lab4}),
(\ref{lab5}), (\ref{lab6})\}$.

In \textbf{Step 2}, we initialize our
desired argument filter $\pi'$. This filter does not yet eliminate any
arguments except for original function symbols from the logic program and for
symbols of the form $p_{in}^I$. Since in our example,
the initial
argument filter $\pi$ of the user is $\pi(\Frotate) = \{1\}$,  we
have $\pi'(\Frotate_{in}) = \{1,2\}$, but
$\pi'({\sf
rotate}_{in}^{\{1\}}) = \{1\}$. So for symbols $p_{in}^I$, the label $I$
describes those arguments that are not filtered away. However, this does not
hold for the other labelled symbols. So the labelling of the symbols
$\Fu_2^{\{1\}}$,
$\Fu_3^{\{1\}}$,  and
$\Fappend_{out}^{\{1\}}$ only represents that they ``belong'' to the symbol
${\sf
rotate}_{in}^{\{1\}}$. But the argument filter for these symbols can be
determined arbitrarily. Initially, 
$\pi'$ would not filter away any of their arguments, i.e.,
$\pi'(\Fu_2^{\{1\}}) =\{1,2,3\}$,
$\pi'(\Fu_3^{\{1\}}) =
\{1,2,3,4,5\}$, and $\pi'(\Frotate_{out}^{\{1\}}) = \{1,2\}$. The filter
for original function symbols of the logic program is taken 
from the user-defined argument filter $\pi$. 
So since the
user described the desired set of queries by setting $\pi(\point) = \{ 2\}$,
we also have $\pi'(\point) = \{2\}$.

In \textbf{Steps 3} and \textbf{3.1}, we look for rules violating the variable
condition as in Algorithm \ref{alg:refinement}. Again, we use a
refinement heuristic $\rho$ to suggest a suitable function symbol $f$ and an
argument position $i$ that should be filtered away. As before, 
we restrict ourselves to refinement heuristics $\rho$ which never select
the first argument of a symbol $u_{c,i}$. In this way, we only have to examine
the rules (and not also the dependency pairs) for possible violations of the
variable condition.

If $f$ is not a (possibly labelled) symbol of the form $p_{in}$ or $p_{in}^I$, then we
proceed in \textbf{Step 3.2} as before (i.e., as in Step 2.2 of Algorithm
\ref{alg:refinement}).
But if $f$ is a (possibly labelled) symbol of the form $p_{in}$ or $p_{in}^I$, then
we do not modify the filter for $f$. If $I$ are
the 
non-filtered argument positions of $f$,
then we introduce a new function symbol labelled with $I \setminus
\{i\}$ instead and replace $f$ by this new function symbol in the rule that violated the
variable condition.

In our example, we had
 $\R_\P' = \{ (\ref{spl1}), \ldots, (\ref{spl6}),
(\ref{lab4}), (\ref{lab5}), (\ref{lab6})\}$ and $\pi'$ was the filter that does
not eliminate any arguments except for $\pi'({\sf rotate}_{in}^{\{1\}}) = \{1\}$
and $\pi'(\point) = \{2\}$.

The rules (\ref{spl2}), (\ref{spl4}), and (\ref{lab4}) violate the variable
condition.
In the following, we mark the violating variables by boxes.
Let us regard Rule (\ref{spl2}) first:
\begin{align}
\rrule{spl2}{{\sf append}_{in}(\point(X,L), M, \point(X,N))}{
\Fu_1({\sf append}_{in}(L,M,N),\framebox{$X$},L,M,N)}
\end{align}
To remove the variable $X$ from the right-hand side,
in \textbf{Step 3.1} any refinement heuristic must suggest to
filter away the second argument of $\Fu_1$.
As $\Fu_1$ does not have the form $p_{in}^I$, we use the
second case of \textbf{Step 3.2}.
Thus, we change $\pi'$
such that $\pi'(\Fu_1) = \{1,2,3,4,5\} \setminus \{2\} = \{1,3,4,5\}$.
Indeed, now this rule does not violate the variable condition anymore.

We reach \textbf{Step 3.3} and, thus, go back to \textbf{Step 3}
where we again choose a rule that violates the variable condition.
Let us now regard Rule (\ref{lab4}):
\begin{align}
\rrule{lab4}{{\sf rotate}_{in}^{\{1\}}(N,O)}{\Fu_2^{\{1\}}({\sf append}_{in}(\framebox{$L$},\framebox{$M$},N),N,\framebox{$O$})}
\end{align}
To remove the first violating variable $L$, in \textbf{Step 3.1} our refinement heuristic
suggests to filter away the first argument of the symbol ${\sf append}_{in}$.
But instead of changing $\pi'({\sf append}_{in})$, we introduce a
new symbol ${\sf append}_{in}^{\{2,3\}}$ with $\pi'({\sf append}_{in}^{\{2,3\}}) = \{2,3\}$.
Moreover, we replace the symbol ${\sf append}_{in}$ in the right-hand side
of Rule (\ref{lab4}) by the new symbol ${\sf append}_{in}^{\{2,3\}}$. Thus,
Rule (\ref{lab4}) is modified to
\begin{align}
\lrule{lab4'}{{\sf rotate}_{in}^{\{1\}}(N,O)}{\Fu_2^{\{1\}}({\sf
append}_{in}^{\{2,3\}}(L,\framebox{$M$},N),N,\framebox{$O$}).}
\end{align}
To make sure that $\Fappend_{in}^{\{2,3\}}$ has rewrite rules corresponding to
the rules of $\Fappend_{in}$, we now have to add copies of all rules
that result from the $\Fappend$-predicate. However, here we label every root
symbol by $\{2,3\}$. In other words, we have to add the following rules to $\R_\P'$:
{\small
\begin{align}
\lrule{lab7}{ {\sf append}_{in}^{\{2,3\}}([\,], M,M) }{ \Fappend_{out}^{\{2,3\}}([\,],M,M)}\\
\lrule{lab8}{ {\sf append}_{in}^{\{2,3\}}(\point(X,L), M, \point(X,N)) }{\Fu_1^{\{2,3\}}({\sf append}_{in}(L,M,N),X,L,M,N)}\\
\lrule{lab9}{\Fu_1^{\{2,3\}}({\sf append}_{out}(L,M,N),X,L,M,N)}{{\sf append}_{out}^{\{2,3\}}(\point(X,L), M, \point(X,N))}
\end{align}}
Now the result of rewriting a term $\Fappend_{in}^{\{2,3\}}(\ldots)$
will always be a term of the form
$\Fappend_{out}^{\{2,3\}}(\ldots)$. Therefore, we have to replace
$\Fappend_{out}$ by $\Fappend_{out}^{\{2,3\}}$ in the left-hand side of 
Rule (\ref{lab5}) (since (\ref{lab5}) is the rule that always ``follows''
(\ref{lab4})). So the original rule (\ref{lab5})
\begin{align}
\rrule{lab5}{\Fu_2^{\{1\}}({\sf append}_{out}(L,M,N),N,O)}{
\Fu_3^{\{1\}}(\Fappend_{in}(M,L,O),L,M,N,O)}
\intertext{is replaced by the modified rule}
\lrule{lab5'}{\Fu_2^{\{1\}}({\sf append}_{out}^{\{2,3\}}(L,M,N),N,O) }{
\Fu_3^{\{1\}}(\Fappend_{in}(M,L,O),L,M,N,O).}
\end{align}
Thus, after the execution of \textbf{Step 3.2}, we have
$\R_\P' = \{ (\ref{spl1}) - (\ref{spl6}), (\ref{lab4'})
- (\ref{lab9}), (\ref{lab5'}),$ $(\ref{lab6}) \}$.
In this way, we have introduced three new labelled symbols
${\sf append}_{in}^{\{2,3\}}\!$, $\Fu_1^{\{2,3\}}\!$, and ${\sf append}_{out}^{\{2,3\}}\!$. 
On the unlabelled symbols, the argument filter $\pi'$ did not change,
but we now additionally have $\pi'({\sf append}_{in}^{\{2,3\}}) = \{2,3\}$,
$\pi'(\Fu_1^{\{2,3\}}) = \{1,2,3,4,5\}$, and $\pi'(\Fappend_{out}^{\{2,3\}}) = \{1,2,3\}$.

We reach \textbf{Step 3.3} and, thus, go back to \textbf{Step 3}
where we again choose a rule that violates the variable condition.
Let us again regard Rule (\ref{lab4}), albeit in its modified
form as Rule (\ref{lab4'}). The variable $M$ still violates the
variable condition. In \textbf{Step 3.1}, the refinement heuristic
suggests to filter away the second argument of the symbol
${\sf append}_{in}^{\{2,3\}}$. Instead of changing $\pi'$, we again
introduce a new symbol, namely ${\sf append}_{in}^{\{3\}}$ with
$\pi'({\sf append}_{in}^{\{3\}}) = \{3\}$, and replace the symbol
${\sf append}_{in}^{\{2,3\}}$ in the right-hand side of Rule
(\ref{lab4'}) by ${\sf append}_{in}^{\{3\}}$. Thus, we obtain a
further modification of Rule (\ref{lab4'}):
\begin{align}
\lrule{lab4''}{{\sf rotate}_{in}^{\{1\}}(N,O) }{ \Fu_2^{\{1\}}({\sf append}_{in}^{\{3\}}(L,M,N),N,\framebox{$O$})}
\end{align}
Again, we have to ensure that $\Fappend_{in}^{\{3\}}$ has rewrite rules
corresponding to the rules of $\Fappend_{in}$. Thus, we
add copies of all rules that result from the $\Fappend$-predicate
where every root symbol is labelled by $\{3\}$:
\begin{align}
\lrule{lab10}{ {\sf append}_{in}^{\{3\}}([\,], M,M) }{ \Fappend_{out}^{\{3\}}([\,],M,M)}\\
\lrule{lab11}{ {\sf append}_{in}^{\{3\}}(\point(X,L), M, \point(X,N)) }{\Fu_1^{\{3\}}({\sf append}_{in}(L,M,N),X,L,M,N)\hspace{-0.2em}}\\
\lrule{lab12}{\hspace{-0.7em}\Fu_1^{\{3\}}({\sf append}_{out}(L,M,N),X,L,M,N) }{ {\sf append}_{out}^{\{3\}}(\point(X,L), M, \point(X,N))}
\end{align}
We also have to replace $\Fappend_{out}^{\{2,3\}}$ by
$\Fappend_{out}^{\{3\}}$ in the left-hand side of Rule (\ref{lab5'})
(since (\ref{lab5'}) is the rule that always ``follows''
(\ref{lab4'})). So the rule (\ref{lab5'}) is replaced by the modified rule
\begin{align}
\lrule{lab5''}{\Fu_2^{\{1\}}({\sf append}_{out}^{\{3\}}(L,M,N),N,O) }{
\Fu_3^{\{1\}}(\Fappend_{in}(M,L,O),L,M,N,O)}
\end{align}
Thus, after the execution of \textbf{Step 3.2}, we have
$\R_\P' = \{ (\ref{spl1}) - (\ref{spl6}), (\ref{lab4''})
- (\ref{lab12}),
(\ref{lab7}) - (\ref{lab9}),
(\ref{lab5''}), (\ref{lab6}) \}$.
Again we have introduced three new labelled symbols
${\sf append}_{in}^{\{3\}}\!$, $\Fu_1^{\{3\}}\!$, and ${\sf append}_{out}^{\{3\}}$.
On the unlabelled symbols, the argument filter $\pi'$ did not change,
but we now additionally have $\pi'({\sf append}_{in}^{\{3\}}) = \{3\}$,
$\pi'(\Fu_1^{\{3\}}) = \{1,2,3,4,5\}$, and $\pi'(\Fappend_{out}^{\{3\}}) = \{1,2,3\}$.

We reach \textbf{Step 3.3} and, thus, go back to \textbf{Step 3}
where we again choose a rule that violates the variable condition.
We again regard Rule (\ref{lab4}), albeit in its modified
form as Rule (\ref{lab4''}). The variable $O$ still violates the
variable condition. In \textbf{Step 3.1}, any refinement heuristic
must suggest to filter away the third argument of the symbol
$\Fu_2^{\{1\}}$. As $\Fu_2^{\{1\}}$ does not have the form $p_{in}^I$,
we use the second case of \textbf{Step 3.2}.
Thus, we change $\pi'$ such that
$\pi'(\Fu_2^{\{1\}}) = \{1,2,3\} \setminus \{3\} = \{1,2\}$.
Indeed, now Rule (\ref{lab4''}) does not violate the variable
condition anymore.

We reach \textbf{Step 3.3} and, thus, go back to \textbf{Step 3}
where we again choose a rule that still violates the variable condition.
Let us now regard Rule (\ref{lab5''}):
\begin{align}
\rrule{lab5''}{\Fu_2^{\{1\}}({\sf append}_{out}^{\{3\}}(L,M,N),N,O) }{
\Fu_3^{\{1\}}(\Fappend_{in}(M,L,\framebox{$O$}),L,M,N,\framebox{$O$})}
\end{align}

Here our refinement heuristic suggests to filter
away the third argument of the symbol $\Fappend_{in}$ in order to remove the extra
variable $O$. Instead of changing $\pi'$, we again introduce a
new symbol, namely $\Fappend_{in}^{\{1,2\}}$
with $\pi'(\Fappend_{in}^{\{1,2\}}) = \{1,2\}$, and replace the symbol
$\Fappend_{in}$ in the right-hand side of Rule (\ref{lab5''})
by $\Fappend_{in}^{\{1,2\}}$. Thus, we obtain a further modification of
Rule (\ref{lab5''}):
\begin{align}
\lrule{lab5'''}{\Fu_2^{\{1\}}({\sf append}_{out}^{\{3\}}(L,M,N),N,O) }{
\Fu_3^{\{1\}}(\Fappend_{in}^{\{1,2\}}(M,L,O),L,M,N,\framebox{$O$})}
\end{align}
Again, we have to ensure that $\Fappend_{in}^{\{1,2\}}$ has rewrite rules
corresponding to the rules of $\Fappend_{in}$. Thus, we
add copies of all rules that result from the $\Fappend$-predicate
where every root symbol is labelled by $\{1,2\}$:

\vspace*{-.3cm}

{\small
\begin{align}
\lrule{lab1}{ {\sf append}_{in}^{\{1,2\}}([\,], M,M) }{ \Fappend_{out}^{\{1,2\}}([\,],
M,M)}\\
\lrule{lab2}{ {\sf append}_{in}^{\{1,2\}}(\point(X,L), M, \point(X,N)) }{
\Fu_1^{\{1,2\}}({\sf append}_{in}(L,M,N),X,L,M,N)}\\
\lrule{lab3}{\Fu_1^{\{1,2\}}({\sf append}_{out}(L,M,N),X,L,M,N) }{
{\sf append}_{out}^{\{1,2\}}(\point(X,L), M, \point(X,N))}
\end{align}}

\vspace*{-.3cm}

We also have to replace $\Fappend_{out}$ by
$\Fappend_{out}^{\{1,2\}}$ in the left-hand side of Rule (\ref{lab6})
(since (\ref{lab6}) is the rule that always ``follows''
(\ref{lab5''})). So the rule (\ref{lab6}) is replaced by the modified rule
\begin{align}
\lrule{lab6'}{\Fu_3^{\{1\}}({\sf append}_{out}^{\{1,2\}}(M,L,O),L,M,N,O) }{
{\sf rotate}_{out}^{\{1\}}(N,O)}
\end{align}
Thus, after the execution of \textbf{Step 3.2}, we now have
$\R_\P' = \{ (\ref{spl1}) - (\ref{spl6}), (\ref{lab4''})
- (\ref{lab12}), (\ref{lab7}) - (\ref{lab9}),
(\ref{lab5'''}) - (\ref{lab3}),
(\ref{lab6'})\}$.
Again we have introduced three new labelled symbols
${\sf append}_{in}^{\{1,2\}}\!$, $\Fu_1^{\{1,2\}}\!$, and ${\sf append}_{out}^{\{1,2\}}$.
On the unlabelled symbols, the argument filter $\pi'$ did not change,
but we now additionally have $\pi'({\sf append}_{in}^{\{1,2\}}) = \{1,2\}$,
$\pi'(\Fu_1^{\{1,2\}}) = \{1,2,3,4,5\}$, and $\pi'(\Fappend_{out}^{\{1,2\}}) =
\{1,2,3\}$.

Note that now we have indeed separated the two copies of the {\sf
append}-rules
where ${\sf append}_{in}^{\{3\}}$ corresponds to the version of $\Fappend$
that has the third argument as input and ${\sf append}_{in}^{\{1,2\}}$ is the
version where the first two arguments serve as input. This copying of
predicates works although the initial argument filter already filtered away
arguments of function symbols like ``$\point$'' (i.e., the initial argument
filter was already beyond the expressivity of modings).

\textbf{Step 3} is repeated until the variable condition is not violated
anymore.
Note that Algorithm \ref{alg:ImprovedRefinement} 
always terminates since there are only finitely many
possible labelled variants for every symbol.
In our example, we obtain the following set of rules $\R_\P'$:

{\footnotesize
\begin{align}
\rrule{spl1}{{\sf append}_{in}([\,], M,M)} { \Fappend_{out}([\,],M,M)}\\
\rrule{spl2}{{\sf append}_{in}(\point(X,L), M, \point(X,N)) }{  \Fu_1({\sf append}_{in}(L,M,N),X,L,M,N)}\\
\rrule{spl3}{\Fu_1({\sf append}_{out}(L,M,N),X,L,M,N) }{ {\sf append}_{out}(\point(X,L), M, \point(X,N))}\\
\lrule{rem5}{{\sf rotate}_{in}(N,O) }{  \Fu_2({\sf append}_{in}^{\{3\}}(L,M,N),N,O)}\\
\urule{\Fu_2({\sf append}_{out}^{\{3\}}(L,M,N),N,O) }{ \Fu_3(\Fappend_{in}(M,L,O),L,M,N,O)}\\
\rrule{spl6}{\Fu_3(\Fappend_{out}(M,L,O),L,M,N,O) }{  {\sf
rotate}_{out}(N,O)}\\
\rrule{lab4''}{{\sf rotate}_{in}^{\{1\}}(N,O) }{  \Fu_2^{\{1\}}({\sf append}_{in}^{\{3\}}(L,M,N),N,O)}\\
\rrule{lab5'''}{\Fu_2^{\{1\}}({\sf append}_{out}^{\{3\}}(L,M,N),N,O) }{\Fu_3^{\{1\}}(\Fappend_{in}^{\{1,2\}}(M,L,O),L,M,N,O)}\\
\rrule{lab6'}{\Fu_3^{\{1\}}(\Fappend_{out}^{\{1,2\}}(M,L,O),L,M,N,O) }{  {\sf rotate}_{out}^{\{1\}}(N,O)}\\
\rrule{lab7}{{\sf append}_{in}^{\{2,3\}}([\,], M,M) }{ \Fappend_{out}^{\{2,3\}}([\,],M,M)}\\
\urule{{\sf append}_{in}^{\{2,3\}}(\point(X,L), M, \point(X,N)) }{\Fu_1^{\{2,3\}}({\sf append}_{in}^{\{2,3\}}(L,M,N),X,L,M,N)}\\
\lrule{rem6}{\Fu_1^{\{2,3\}}({\sf append}_{out}^{\{2,3\}}(L,M,N),X,L,M,N) }{
{\sf append}_{out}^{\{2,3\}}(\point(X,L), M, \point(X,N))}\\
\rrule{lab10}{{\sf append}_{in}^{\{3\}}([\,], M,M) }{ \Fappend_{out}^{\{3\}}([\,],M,M)}\\
{\sf append}_{in}^{\{3\}}(\point(X,L), M, \point(X,N)) &\to  \Fu_1^{\{3\}}({\sf append}_{in}^{\{3\}}(L,M,N),X,L,M,N)\\
\urule{\Fu_1^{\{3\}}({\sf append}_{out}^{\{3\}}(L,M,N),X,L,M,N) }{ {\sf append}_{out}^{\{3\}}(\point(X,L), M, \point(X,N))}\\
\rrule{lab1}{{\sf append}_{in}^{\{1,2\}}([\,], M,M) }{ \Fappend_{out}^{\{1,2\}}([\,],M,M)}\\
\urule{{\sf append}_{in}^{\{1,2\}}(\point(X,L), M, \point(X,N)) }{ \Fu_1^{\{1,2\}}({\sf append}_{in}^{\{1,2\}}(L,M,N),X,L,M,N)}\\
\urule{\Fu_1^{\{1,2\}}({\sf append}_{out}^{\{1,2\}}(L,M,N),X,L,M,N) }{ {\sf append}_{out}^{\{1,2\}}(\point(X,L), M, \point(X,N))}
\end{align}
}

\vspace*{-.2cm}

\noindent
The refined argument filter $\pi'$ is given by

\vspace*{-.1cm}

{\footnotesize
\[ \begin{array}{lll@{\quad}lll@{\quad}lll}
\pi'({\sf append}_{in}) &\!\!=\!\!& \{1,2,3\} &
\pi'({\sf rotate}_{in}^{\{1\}})  &\!\!=\!\!& \{1\}&
\pi'({\sf append}_{in}^{\{2,3\}}) &\!\!=\!\!& \{2,3\}\\
\pi'({\sf append}_{out}) &\!\!=\!\!& \{1,2,3\} &
\pi'(\Fu_2^{\{1\}})&\!\!=\!\!& \{1,2\}&
\pi'({\sf append}_{out}^{\{2,3\}}) &\!\!=\!\!& \{1,2,3\}\\
\pi'(\point)&\!\!=\!\!& \{2\}&
\pi'(\Fu_3^{\{1\}})&\!\!=\!\!& \{1,2,3,4\}&
 \pi'(\Fu_1^{\{2,3\}})&\!\!=\!\!& \{1,4,5\}\\
 \pi'(\Fu_1) &\!\!=\!\!& \{1,3,4,5\}&
\pi'({\sf append}_{in}^{\{1,2\}})&\!\!=\!\!& \{1,2\}&
 \pi'(\Fu_1^{\{3\}})&\!\!=\!\!& \{1,5\}\\
\pi'({\sf rotate}_{in}) &\!\!=\!\!& \{1,2\}&
\pi'({\sf append}_{out}^{\{1,2\}}) &\!\!=\!\!& \{1,2,3\}&
 \pi'(\Fu_1^{\{1,2\}})&\!\!=\!\!& \{1,3,4\}\\
\pi'(\Fu_2)&\!\!=\!\!& \{1,2,3\}&
\pi'({\sf rotate}_{out}^{\{1\}}) &\!\!=\!\!& \{1,2\} & &&
\\
\pi'({\sf append}_{in}^{\{3\}}) &\!\!=\!\!& \{3\}& &&&&&\\
\pi'({\sf append}_{out}^{\{3\}})  &\!\!=\!\!& \{1,2,3\}& &&&&&\\
\pi'(\Fu_3)&\!\!=\!\!& \{1,2,3,4,5\}& &&&&&\\
\pi'({\sf rotate}_{out}) &\!\!=\!\!& \{1,2\}& &&&&&
\end{array}\]
}

Termination for $\R_\P'$ w.r.t.\ the terms specified by $\pi'$ is now easy
to show using our results from \rSec{sec:termination}.

If one is only interested in termination of queries
$\Frotate(t_1,t_2)$ for a specific predicate symbol like $\Frotate$, then one
can remove superfluous (copies of) rules from the TRS before starting the
termination proof. For example, if one only wants to prove termination of queries
$\Frotate(t_1,t_2)$ for finite lists $t_1$, then it now suffices to prove
termination of the above TRS for those ``start terms'' $\Frotate_{in}^{\{1\}}(\ldots)$
that are finite and ground under the filter $\pi'$ and where the arguments of 
$\Frotate_{in}^{\{1\}}$ do not contain any function symbols except $\point$
and $[\,]$. 
Since the rules for $\Frotate_{in}$, $\Fappend_{in}$, and
$\Fappend_{in}^{\{2,3\}}$ (i.e., the rules (\ref{spl1}) - (\ref{spl3}), 
(\ref{spl6}), (\ref{lab7}), and (\ref{rem5}) - (\ref{rem6}))
are not reachable from these ``start terms'',
they can immediately be removed. In other words, for the queries
$\Frotate(t_1,t_2)$ we indeed need rules for $\Frotate_{in}^{\{1\}}\!$, 
$\Fappend_{in}^{\{1,2\}}\!$, and
$\Fappend_{in}^{\{3\}}\!$, 
but the rules for $\Frotate_{in}$, $\Fappend_{in}$, and 
$\Fappend_{in}^{\{2,3\}}\!$ are superfluous.

Note however  that such superfluous copies
of rules are never problematic for the termination
analysis. If the rules for
$\Fappend_{in}^{\{3\}}$ are terminating for terms that are 
 finite and ground under the filter $\pi'$, then this also holds for the
$\Fappend_{in}^{\{2,3\}}$- and the $\Fappend_{in}$-rules, since here $\pi'$
filters away 
less arguments. A corresponding statement holds for the connection between the
$\Frotate_{in}^{\{1\}}\!$- and the  $\Frotate_{in}$-rules.

The following theorem proves the correctness of Algorithm
\ref{alg:ImprovedRefinement}.
More precisely, it shows that one can use $\pi'$ and $\R_\P'$ instead of $\pi$
and $\R_\P$ in \rThm{Soundness of the Transformation}. 
So it is sufficient to prove that all terms in the set 
$S' = \{ p_{in}^{\pi(p)}(\vec{t})\mid
p\in\Delta,
\; \vec{t} \in \vec{\T}^\infty(\Sigma,\V), \; \pi'(p_{in}^{\pi(p)}(\vec{t}))\in\T(\Sigma_{\P_{\pi'}})
\, \}$ are terminating 
w.r.t.\ the modified TRS $\R_\P'$.
In \rEx{ex:nonuniqueafs}, $S'$ would be the set of all terms ${\sf rotate}_{in}^{\{1\}}(t_1, t_2)$
that are ground after filtering with $\pi'$. Hence, this includes all terms where the
first argument is a finite list.

If all terms in $S'$ are terminating w.r.t.\
$\R_\P'$, we can conclude that 
all queries $Q\in\A^{rat}(\Sigma,\Delta,\V)$ with
$\pi(Q)\in\A(\Sigma_\pi,\Delta_\pi)$ are terminating for the original logic
program. Since $\pi'$ satisfies the variable condition for the TRS $\R_\P'$
(and also for $DP(\R_\P')$ if $1 \in \pi'(u_{c,i})$ for all symbols of the
form $u_{c,i}$), one can also use $\pi'$ and $\R_\P'$ for the termination
criterion of Corollary \ref{Termination of Logic Programs by Dependency Pairs}. 
In other words, then it is sufficient to prove that there is no infinite
$(DP(\R_\P'),\R_\P',\pi')$-chain.

\begin{theorem}[(Soundness of Algorithm
\ref{alg:ImprovedRefinement})]
\label{thm:ImprovedRefinement}
Let $\P$ be a logic program and let $\pi$ be an argument filter
over $(\Sigma,\Delta)$.
Let $\pi'$ and $\R_\P'$ result from  $\pi$  and $\R_\P$ by
\rAlg{alg:ImprovedRefinement}.
Let $S = \{ p_{in}(\vec{t})\mid
p\in\Delta,
\; \vec{t} \in \vec{\T}^\infty(\Sigma,\V), \; \pi(p_{in}(\vec{t}))\in\T(\Sigma_{\P_\pi})
\, \}$.
Furthermore, let
$S' = \{ p_{in}^{\pi(p)}(\vec{t})\mid
p\in\Delta,
\; \vec{t} \in \vec{\T}^\infty(\Sigma,\V), \;
\pi'(p_{in}^{\pi(p)}(\vec{t}))\in\T(\Sigma_{\P_{\pi'}}) 
\, \}$. 
All terms $s\in S$ are terminating for $\R_\P$ if all terms $s' \in S'$
are terminating for $\R_\P'$.
\end{theorem}
\begin{proof}
We
first show that every reduction of a term from $S$ with $\R_\P$ can be simulated by
the reduction of a term from $S'$ with $\R_\P'$. More precisely, we show the following proposition
where $\mathbb{S}^n = \{t \mid 
p_{in}(\vec{t})
\to^n_{\R_\P} t$ for some 
$p\in\Delta$,
$\vec{t} \in \vec{\T}^\infty(\Sigma,\V)$, and
$\pi(p_{in}(\vec{t})) \in {\cal T}(\Sigma_{\P_\pi})\, \}$
and $\mathbb{S'} = \{t \mid 
p_{in}^{\pi(p)}(\vec{t})
\to^*_{\R_\P'} t$  for some 
$p\in\Delta$,
$\vec{t} \in \vec{\T}^\infty(\Sigma,\V)$, and
$\pi'(p_{in}^{\pi(p)}(\vec{t})) \in {\cal T}(\Sigma_{\P_{\pi'}})\, \}$
\begin{equation}
\label{ilemma}
\parbox{10cm}{If $s \in \mathbb{S}^n$ and $s' \in \mathbb{S'}$ 
with $\Unlab(s') = s$, then $s \to_{\R_\P}  t$ implies that there is
a $t'$ 
with $\Unlab(t') = t$ and  $s' \to_{\R_\P'}  t'$.}
\end{equation}
Here, $\Unlab$ removes all labels introduced by \rAlg{alg:ImprovedRefinement}:
\begin{eqnarray*}
\Unlab(s) &=& \begin{cases}
f(\Unlab(s_1),\ldots,\Unlab(s_n)), & \mbox{ if }s = f^I(s_1,\ldots,s_n)\\
s, & \mbox{ otherwise}
\end{cases}
\end{eqnarray*}

We prove (\ref{ilemma}) by induction on $n$. There are three possible cases for
$s$ and for the rule that is applied 
in the step from $s$ to $t$. \\

\noindent
\underline{Case 1: $n= 0$ and thus, $s = p_{in}(\vec{s})$}

\noindent
So $s \in S$ and there is a rule $\ell \to r \in \R_\P$
with $\ell = p_{in}(\vec{\ell})$ such that $s = \ell \sigma$ and $t = r\sigma$
for some substitution $\sigma$ with
terms from $\T^\infty(\Sigma,\V)$. 

Let $s' \in \mathbb{S'}$ 
with $\Unlab(s') = s$. Thus, we also have $s' \in S'$ where $s' =
p_{in}^{\pi(p)}(\vec{s})$ (since
a term with a root symbol $p_{in}^I$ cannot be obtained from $S'$ if one has
performed at least one rewrite step with $\R_\P'$).
Due to the construction of $\R_\P'$, there exists
a rule $\ell^{\pi(p)} \to r' \in \R_\P'$ where $\Unlab(r') = r$. We define
$t'$ to be $r' \sigma$. Then we clearly
have $s' = \ell^{\pi(p)}\sigma \to_{\R_\P'} r'\sigma = t'$ and
$\Unlab(t') = t$.\\

\noindent
\underline{Case 2: $n \geq 1$ and  $s = u_{c,i}(\overline{s}, \vec{q})$, $\overline{s}
\to_{\R_\P} \overline{t}$, $t = u_{c,i}(\overline{t},\vec{q})$}

\noindent
Since $s \in \mathbb{S}^n$, there exists a $p_{in}(\vec{s})$ with 
$\vec{s} \in \vec{\T}^\infty(\Sigma,\V)$ such that $p_{in}(\vec{s}) \to^*_{\R_\P}
\overline{s}$, i.e., $\overline{s} \in \mathbb{S}^m$ for some $m \in \mathbb{N}$.
Since the reduction from  $p_{in}(\vec{s})$ to $\overline{s}$
is shorter than the overall reduction that led to $s$, we obtain that
$m < n$. 

Let $s'  \in \mathbb{S'}$ 
with $\Unlab(s') = s$. Hence, we have
$s' =  u_{c,i}^I(\overline{s'}, \vec{q})$ for some label $I$ 
and
$\Unlab(\overline{s'}) = \overline{s}$.
Since $s' \in \mathbb{S'}$, there exists a $p_{in}^J(\vec{s})$ with 
$\vec{s} \in \vec{\T}^\infty(\Sigma,\V)$ such that $p_{in}^J(\vec{s}) \to^*_{\R_\P'}
\overline{s'}$. Hence, $\overline{s'} \in \mathbb{S'}$ as well. Now the induction
hypothesis implies that there exists a $\overline{t'}$ such
that $\overline{s'} \to_{\R_\P'} \overline{t'}$ and $\Unlab(\overline{t'}) =
\overline{t}$.
We define $t' = u_{c,i}^I(\overline{t'}, \vec{q})$. Then we clearly have
$s' \to_{\R_\P'} t'$ and $\Unlab(t') = t$.\\

\noindent
\underline{Case 3: $n \geq 1$ and  $s = u_{c,i}(p_{out}(\vec{s}), \vec{q})$}

\noindent
Here, there exists a rule $\ell \to r \in \R_\P$ with $\ell =
u_{c,i}(p_{out}(\vec{\ell}),\vec{x})$ such that $s = \ell\sigma$ and $t =
r\sigma$.

Let $s'  \in \mathbb{S'}$ 
with $\Unlab(s') = s$. Hence, we have
$s' =  u_{c,i}^I(p_{out}^J(\vec{s}), \vec{q})$ for some labels $I$ and $J$.
Since $s' \in \mathbb{S'}$, $s'$ resulted from rewriting the term 
$u_{c,i}^I(p_{in}^J(\vec{s}), \vec{q})$ which must be an instantiated right-hand
side of a rule from $\R_\P'$. 
Due to the construction of $\R_\P'$, then there
also exists
a rule $\ell' \to r' \in \R_\P'$
where $\ell' = u_{c,i}^I(p_{out}^J(\vec{\ell}), \vec{x})$
and $\Unlab(r') = r$. We define $t' = r'\sigma$. Then we have $s' =
\ell'\sigma \to_{\R_\P'} r'\sigma = t'$ and clearly $\Unlab(t') = t$.

We now proceed to prove the theorem by contradiction. Assume there is a term
$s_0 \in S$ that is non-terminating w.r.t.\ $\R_\P$, i.e., there is an infinite sequence of terms
$s_0, s_1, s_2, \ldots$ with $s_i \to_{\R_\P} s_{i+1}$.
We must have $s_0 = p_{in}(\vec{t})$ with $\vec{t} \in
\vec{\T}^\infty(\Sigma,\V)$ and $\pi(p_{in}(\vec{t}))\in\T(\Sigma_{\P_\pi})$.
Let $s_0' = p_{in}^{\pi(p)}(\vec{t})$. Then $s_0' \in S'$, since 
$\pi'(p_{in}^{\pi(p)}(\vec{t}))
\in\T(\Sigma_{\P_{\pi'}})$. The reason is that $\pi'(p_{in}^{\pi(p)}) = \pi(p)
= \pi(p_{in})$ and for all $f \in \Sigma$ we have $\pi'(f) \subseteq \pi(f)$. 

So by (\ref{ilemma}), $s_0' \in \mathbb{S'}$ and $\Unlab(s_0') = s_0$ imply
that
there is an $s_1'$ with $\Unlab(s_1') = s_1$ and $s_0' \to_{\R_\P'}
s_1'$. Clearly, this also implies $s_1' \in \mathbb{S'}$. By applying
(\ref{ilemma}) repeatedly, we therefore obtain an infinite sequence of labelled terms
$s_0', s_1', s_2', \ldots$ with $s'_i \to_{\R_\P'} s'_{i+1}$.
\end{proof}

\section{Formal Comparison of the Transformational Approaches}
\label{sec:previous}

In this section we  formally compare the
power of the classical transformation from
\rSec{The Classical Transformation} with the power of our new approach.
In the classical approach, the class of queries is characterized by a moding
function whereas in our approach, it is characterized by an argument
filter. Therefore, the following definition
establishes a relationship between modings and argument
filters.

\begin{definition}[(Argument Filter Induced by Moding)]
\label{def:afsfrommode}
Let $(\Sigma,\Delta)$ be a signature and
let $m$ be a moding over the set of predicate symbols $\Delta$.
Then for every predicate symbol $p \in \Delta$ we define the \emph{induced argument filter $\pi_m$} over $\Sigma_\P$
as $\pi_m(p_{in}) = \pi_m(P_{in}) = \{i \mid m(p,i) = \mIn\}$ and
$\pi_m(p_{out}) = \{i \mid m(p,i) = \mOut\}$. All other function
symbols $f$ from $\Sigma_\P$ are not filtered, i.e.,
$\pi_m(f/n) = \{1,\ldots,n\}$.
\end{definition}

\begin{example}\label{compare1}
{\sl
Regard again the well-moded logic program from \rEx{ex_not_sm}.
\[\begin{array}{lll}
{\sf p}(X, X). & &\\
{\sf p}({\sf f}(X), {\sf g}(Y)) &\from& {\sf p}({\sf f}(X), {\sf f}(Z)),
{\sf p}(Z,{\sf g}(Y)).
\end{array}\]
We used the moding $m$ with  $m({\sf p},1) = \mIn$ and $m({\sf p},2) =
\mOut$.
Thus, for the induced argument filter $\pi_m$ we have
$\pi_m(\Fp_{in}) = \pi_m({\sf P}_{in}) = \{ 1 \}$ and $\pi_m(\Fp_{out}) = \{ 2
\}$.}
\end{example}

As the classical approach is only applicable to well-moded
logic programs, we restrict our comparison to this class.
For non-well-moded programs, our new approach is clearly
more powerful, since it can often prove termination (cf.\ \rSec{sec:experiments}), whereas the classical
transformation is never applicable.

Our goal is to show the connection between the TRSs resulting from the
two transformations. If one refines $\pi_m$ to a filter $\pi_m'$ by \rAlg{alg:refinement} using \emph{any}
arbitrary refinement heuristic, then the TRS of the classical transformation
corresponds to the TRS of our new
transformation after filtering it with
$\pi_m'$.

\begin{example}\label{compare2}{\sl
We continue with \rEx{compare1}.
The TRS $\R_\P$ resulting from our new transformation was given in
\rEx{ex_not_sm_we}:
\begin{align}
\rrule{32-1}{
{\sf p}_{in}(X,X) }{ {\sf p}_{out}(X,X)}\\
\rrule{32-2}{
{\sf p}_{in}({\sf f}(X), {\sf g}(Y)) }{ {\sf u}_1({\sf p}_{in}({\sf f}(X),
{\sf f}(Z)),X,Y)} \\
\rrule{32-3}{
{\sf u}_1({\sf p}_{out}({\sf f}(X), {\sf f}(Z)),X,Y) }{ {\sf u}_2({\sf p}_{in}(Z,{\sf g}(Y)),X,Y,Z)}\\
\rrule{32-4}{
{\sf u}_2({\sf p}_{out}(Z,{\sf g}(Y)),X,Y,Z) }{ {\sf p}_{out}({\sf f}(X),{\sf g}(Y))}
\end{align}
If
we apply the induced argument filter $\pi_m$, then we obtain the  TRS
$\pi_m(\R_\P)$:
\begin{align*}
\urule{{\sf p}_{in}(X) }{ {\sf p}_{out}(X)} \\
\urule{{\sf p}_{in}({\sf f}(X)) }{{\sf u}_1({\sf p}_{in}({\sf f}(X)),X,Y) } \\
\urule{{\sf u}_1({\sf p}_{out}({\sf f}(Z)),X,Y) }{ {\sf u}_2({\sf p}_{in}(Z),X,Y,Z)} \\
\urule{{\sf u}_2({\sf p}_{out}({\sf g}(Y)),X,Y,Z) }{ {\sf p}_{out}({\sf g}(Y))}
\end{align*}

The second rule has the ``extra'' variable $Y$ on the right-hand side, i.e.,
it does not satisfy the variable condition. Thus, we have to refine
the filter $\pi_m$ to a filter $\pi_m'$ with $\pi_m'(\Fu_1) = \pi_m'(\FU_1) =\{1,2\}$ and
$\pi_m'(\Fu_2) = \pi_m'(\FU_2) = \{1,2,4\}$. The resulting TRS $\pi_m'(\R_\P)$ is
identical to the TRS $\R_\P^{old}$ resulting from the classical
transformation, cf.\ \rEx{old transformation}:
\begin{align*}
\urule{{\sf p}_{in}(X) }{ {\sf p}_{out}(X)} \\
\urule{{\sf p}_{in}({\sf f}(X)) }{{\sf u}_1({\sf p}_{in}({\sf f}(X)),X) } \\
\urule{{\sf u}_1({\sf p}_{out}({\sf f}(Z)),X) }{ {\sf u}_2({\sf p}_{in}(Z),X,Z)} \\
\urule{{\sf u}_2({\sf p}_{out}({\sf g}(Y)),X,Z) }{ {\sf p}_{out}({\sf g}(Y))}
\end{align*}}
\end{example}

The following theorem shows that our approach (with Corollary \ref{Termination
of Logic Programs by Dependency Pairs}) succeeds
whenever the classical transformation of \rSec{The Classical
Transformation} yields a terminating TRS.

\begin{theorem}[(Subsumption of the Classical Transformation)]
\label{subsumption}
Let $\P$ be a well-moded logic program over a signature $(\Sigma,\Delta)$
w.r.t.\ the moding $m$. Let $\R_\P^{old}$ be the result
of applying the classical transformation of \rSec{The Classical
Transformation} and let
$\R_\P$ be the result of our new transformation from \rDef{unmodedTranslation}.
Then there is a refinement of $\pi_m'$ of $\pi_m$ such that (a)
$\pi_m'(\R_\P)$ and $\pi_m'(DP(\R_\P))$ satisfy the variable condition
and (b) if $\R_\P^{old}$ is terminating (with ordinary rewriting),
then there is no infinite $(DP(\R_\P),\R_\P,\pi_m')$-chain.
Thus, in particular, termination of $\R_\P^{old}$ implies that
$\R_\P$ is terminating (with infinitary constructor rewriting) for all terms  $p_{in}(\vec{t})$ with
$p\in\Delta$,
$\vec{t} \in \vec{\T}^\infty(\Sigma,\V)$, and $\pi(p_{in}(\vec{t}))\in \T(\Sigma_{\P_\pi})$.
\end{theorem}
\begin{proof}
Let $\pi_m'$ result from \rAlg{alg:refinement} using any refinement
heuristic $\rho$ which does not filter away the first argument
of any $u_{c,i}$.

We now analyze the structure of the TRS $\pi_m'(\R_\P)$.
For any predicate symbol $p \in \Delta$, let
``$p(\vec{s},\vec{t})$''  denote that $\vec{s}$
and $\vec{t}$ are
the sequences of terms on $p$'s in- and output positions w.r.t.\ the moding
$m$.

When \rAlg{alg:refinement} is applied to compute the refinement $\pi_m'$ of
$\pi_m$, one looks for a rule  $\el \to r$ from $\pi_m(\R_\P)$
such that $\V(r) \not\subseteq
\V(\el)$. Such a rule cannot result from the facts of the logic
program. The reason is that for each fact
$p(\vec{s},\vec{t})$,
$\pi_m(\R_\P)$
contains the rule
\[p_{in}(\vec{s}) \to
p_{out}(\vec{t})\]
and by well-modedness, we have
$\V(\vec{t}) \subseteq \V(\vec{s})$.

For each rule $c$ of the form
$p(\vec{s},\vec{t}) \from p_1(\vec{s}_1,\vec{t}_1), \ldots,
p_k(\vec{s}_k,\vec{t}_k)$ in $\P$, the TRS\linebreak
 $\pi_m(\R_\P)$ contains:
\[\begin{array}{l}
p_{in}(\vec{s}) \; \to \; u_{c,1}(p_{1_{in}}(\vec{s}_1),
\V(\vec{s}) \cup \V(\vec{t}))\\
u_{c,1}(p_{1_{out}}(\vec{t}_1), \V(\vec{s}) \cup \V(\vec{t})) \; \to \;
u_{c,2}(p_{2_{in}}(\vec{s}_2), \V(\vec{s})  \cup \V(\vec{t}) \cup \V(\vec{s}_1)\cup \V(\vec{t}_1))\\
\hspace*{5cm} \vdots\\
u_{c,k}(p_{k_{out}}(\vec{t}_k), \V(\vec{s})  \cup \V(\vec{t}) \cup \V(\vec{s}_1)\cup \V(\vec{t}_1) \cup \ldots
\cup \V(\vec{s}_{k-1})  \cup \V(\vec{t}_{k-1})) \; \to \; p_{out}(\vec{t})
\end{array}
\]

For the first rule, by well-modedness we have
$\V(\vec{s}_1) \subseteq \V(\vec{s})$ and thus, the only ``extra'' variables on the
right-hand side of the first rule must be from $\V(\vec{t})$. There is only
one possibility to refine the argument filter in order to remove them:
one has to filter away the respective argument positions of
$u_{c,1}$.
Hence, the filtered right-hand side of the first rule is $u_{c,1}(p_{1_{in}}(\vec{s}_1),
\V(\vec{s}))$ and the filtered left-hand side of the second rule is 
$u_{c,1}(p_{1_{out}}(\vec{t}_1), \V(\vec{s}))$.

Similarly, for the second rule, well-modedness implies
$\V(\vec{s}_2) \cup  \V(\vec{s})  \cup \V(\vec{s}_1)  \cup \V(\vec{t}_1) \subseteq \V(\vec{t}_1) \cup
\V(\vec{s})$. So the only ``extra'' variables on the right-hand side of the
second rule are again from $\V(\vec{t})$. As before, to remove them
one has to filter away the respective argument positions of $u_{c,2}$.
Moreover, since $\V(\vec{s}_1) \subseteq \V(\vec{s})$ we obtain the filtered 
right-hand side $u_{c,2}(p_{2_{in}}(\vec{s}_2), \V(\vec{s}) \cup
\V(\vec{t}_1))$ for the second rule and the filtered left-hand side $u_{c,2}(p_{2_{out}}(\vec{t}_2), \V(\vec{s}) \cup \V(\vec{t}_1))$
side in the third rule.

An analogous argument holds for the other rules. The last rule has no extra
variables, since
$\V(\vec{t}) \subseteq \V(\vec{s}) \cup \V(\vec{t}_1)\cup \ldots \cup \V(\vec{t}_k)$
by well-modedness.

So   for any rule $c$ of the logic program $\P$, $\pi_m'(\R_\P)$ has
the following rules:
\[\begin{array}{rcl}
p_{in}(\vec{s}) & \to & u_{c,1}(p_{1_{in}}(\vec{s}_1),
\V(\vec{s}))\\
u_{c,1}(p_{1_{out}}(\vec{t}_1), \V(\vec{s})) & \to &
u_{c,2}(p_{2_{in}}(\vec{s}_2), \V(\vec{s}) 
\cup \V(\vec{t}_1))\\
&  \vdots & \\
u_{c,k}(p_{k_{out}}(\vec{t}_k), \V(\vec{s}) 
\cup \V(\vec{t}_1) \cup \ldots
\cup \V(\vec{t}_{k-1})) & \to & p_{out}(\vec{t})
\end{array}
\]
Hence, $\pi_m'(\R_\P) = \R_\P^{old}$. Since the
refined argument filter $\pi_m'$ does not filter away the first argument of
any $u_{c,i}$, by defining $\pi_m'(U_{c,i}) := \pi_m'(u_{c,i})$, then the
variable condition is satisfied for both $\pi_m'(\R_\P)$ and $\pi_m'(DP(\R_\P))$
and, thus, (a) is fulfilled.

Now to prove (b), we assume that $\R_\P^{old}$ is terminating. We have
to show that then there is no infinite $(DP(\R_\P),\R_\P,\pi_m')$-chain. By
the soundness of the argument filter processor (\rThm{Argument Filter
Processor}), it suffices to show
that there is no infinite $(\pi_m'(DP(\R_\P)), \pi_m'(\R_\P), id)$-chain.

Note that $\pi_m'(DP(\R_\P)) = DP(\pi_m'(\R_\P))$. The reason is that all $u_{c,i}$ only occur on the
root level in $\R_\P$. Moreover, all $p_{in}$-symbols only occur in the first
argument of a $u_{c,i}$ and $1
\in \pi_m'(u_{c,i})$. In other words, occurrences of defined function symbols are not
removed by the filter $\pi_m'$.
So we have
\[\begin{array}{ll}
&u \to v \in \pi_m'(DP(\R_\P)) \vspace*{.1cm}\\
\mbox{iff}& \parbox[t]{10cm}{there is a rule $\ell \to r
\in \R_\P$ with $u = \pi_m'(\ell^\sharp), v = \pi_m'(t^\sharp)$\\for a subterm $t$ of $r$ with
defined root\vspace*{.2cm}}\\
\mbox{iff}& \parbox[t]{10cm}{there is a rule $\ell \to r
\in \R_\P$ with $u = (\pi_m'(\ell))^\sharp, v = (\pi_m'(t))^\sharp$\\for a subterm $\pi_m'(t)$ of $\pi_m'(r)$ with
defined root\vspace*{.2cm}}\\
\mbox{iff}& \parbox[t]{10cm}{there is a rule $\ell \to r
\in \pi_m'(\R_\P)$ with $u = \ell^\sharp, v = t^\sharp$\\ for a subterm $t$ of $r$ with
defined root\vspace*{.2cm}}\\
\mbox{iff}& u \to v \in DP(\pi_m'(\R_\P))
\end{array}\]
Hence,   $\pi_m'(\R_\P) = \R_\P^{old}$ and $\pi_m'(DP(\R_\P)) =
 DP(\pi_m'(\R_\P)) = DP(\R_\P^{old})$. Thus, it suffices to show absence of
infinite
$(DP(\R_\P^{old}), \R_\P^{old}, id)$-chains.
But this follows from termination of $\R_\P^{old}$, cf.\ \cite[Thm.\ 6]{AG00}, since
$(\P,\R, id)$-chains correspond to chains for ordinary (non-infinitary)
rewriting.

Hence by \rThm{Proving Infinitary Termination}, termination of $\R_\P^{old}$ also implies that
 all terms  $p_{in}(\vec{t})$ with
$p\in\Delta$,
$\vec{t} \in \vec{\T}^\infty(\Sigma,\V)$, and
$\pi(p_{in}(\vec{t}))\in\T(\Sigma_{\P_\pi})$ are terminating w.r.t.\ $\R_\P$ (using
infinitary constructor rewriting).
\qed
\end{proof}

The reverse direction of the above theorem does not hold, though.
As a counterexample, regard again the logic program from \rEx{ex_not_sm}, cf.\
\rEx{compare2}. As shown in \rEx{old transformation}, the TRS resulting from
the classical transformation is not terminating. Still,
for the filter $\pi_m'$ from \rEx{compare2}, there is no infinite
$(DP(\R_\P),\R_\P,\pi_m')$-chain and thus, our method of Corollary
\ref{Termination of Logic Programs by Dependency Pairs}
succeeds with the termination proof. In other words, our new approach is
\emph{strictly} more powerful than the classical transformation, even on
well-moded programs.

Thus, a termination analyzer based on our new transformation
should be strictly more successful in practice, too. That this
is in fact the case will be demonstrated in the next section.

\section{Experiments and Discussion}
\label{sec:experiments}

We integrated our approach (including all refinements presented) in the termination tool {\sf AProVE}
\cite{IJCAR06} which implements the DP framework.
 To evaluate our results, we
tested {\sf AProVE} against four other representative termination tools for logic
programming: {\sf TALP}~\cite{TALP} is the only other available tool based on
transformational methods (it uses the classical transformation of
\rSec{The Classical Transformation}),
whereas
{\sf Polytool}~\cite{Nguyen:DeSchreye06},
{\sf
TerminWeb}~\cite{Codish:Taboch}, and {\sf cTI}~\cite{Mesnard:Bagnara}
are based on
direct approaches.
\rSec{Experimental Evaluation} describes the results of our experimental
evaluation and in \rSec{sec:limitations} we discuss the limitations of our approach.

\subsection{Experimental Evaluation}\label{Experimental Evaluation}

We ran the tools
on a set of 296 examples in fully automatic mode.\footnote{We combined  \emph{termsize} and
\emph{list-length} norm for {\sf TerminWeb} and
 allowed 5 iterations
before widening for \textsf{cTI}. Apart from that, we used the default
settings of the tools. For both {\sf AProVE} and {\sf Polytool} we used the 
(fully automated)
original executables from the \emph{Termination Competition} 2007 \cite{Competition}.
To refine argument filters, this version of {\sf AProVE} uses the refinement heuristic
$\rho_{\mathit{tb}'}$ from \rDef{def:tb2}. For a list of the main termination
techniques used in {\sf AProVE}, we refer to \cite{LPAR04,JAR06}. Of these
techniques, only the ones in Section \ref{Automation by Adapting the DP
Framework} were adapted to infinitary 
constructor rewriting.
}
 This set
includes all logic programming examples from the \emph{Termination
Problem Data Base}  \cite{TPDB}  which is
used in the
annual international \emph{Termination Competition} \cite{Competition}. It
contains
collections
provided by the developers of several different tools including
all  examples from the experimental evaluation
of \cite{CodishTOPLAS}. However, to eliminate the  influence of the
translation  from
\textsf{Prolog}
 to logic programs,
we removed all examples that use non-trivial built-in
predicates or that are not definite logic
programs after ignoring the cut operator.
All tools were run locally on an AMD Athlon 64 at 2.2 GHz under GNU/Linux 2.6.
For each example we used a time limit of 60 seconds. This is similar to the
way that tools are evaluated in the annual competitions for termination
tools.

\begin{center}
\begin{tabular}{l||c|@{\quad\;}c@{\quad\;}|@{\qquad}c@{\qquad}|c|@{\quad\;}c@{\quad\;}}
          & {\sf AProVE}  & {\sf Polytool} &  {\sf TerminWeb} &{\sf cTI} &
{\sf TALP}\\
\hline
Successes &             232 &            204 & 177 &       167 &  163\\
Failures  &              57 &             82 & 118 &       129 &  112\\
Timeouts  &               7 &             10 &   1 &        0 &   21
\end{tabular}
\end{center}

As shown in the table above, {\sf AProVE} succeeds on more examples than
any other tool.
The comparison of {\sf AProVE} and {\sf TALP} shows that our approach
improves significantly upon
the previous transformational method that {\sf TALP} is based on, cf.\ Goals (A)
and (B).
In particular, {\sf TALP} fails for all
non-well-moded programs.

While we have shown our technique to be strictly more powerful than the
previous transformational method, due to the higher
arity of the function symbols produced by our transformation, proving termination could
take more time in some cases. However, 
in the above experiments this did not affect the practical
power of our implementation. In fact, {\sf AProVE} is able to prove termination
well within the time limit for all examples where {\sf TALP} succeeds. Further analysis shows
that while {\sf AProVE} never takes more than 15 seconds longer than {\sf TALP}, there are
indeed 6 examples where {\sf AProVE} is more than 15 seconds \emph{faster} than {\sf TALP}.

The comparison with {\sf Polytool}, {\sf TerminWeb}, and {\sf cTI} demonstrates that our new
transformational approach is not only comparable in power, but usually more
powerful
than direct approaches.
In fact,
there is only a single example where one of the other
tools (namely {\sf Polytool}) succeeds and {\sf AProVE} fails.
This is the rather contrived example from (2) in \rSec{sec:limitations} which
we developed to demonstrate the limitations of our method.
{\sf Polytool} is only able to handle this example via a pre-processing
step based on partial evaluation 
\cite{Nguyen:Bruynooghe:DeSchreye:Leuschel,DBLP:conf/lopstr/SerebrenikS03,Tamary:Codish}.
In this example, this 
pre-processing step results in a trivially
terminating logic program. Thus, if one combined this pre-processing with any
of the other tools, then they would also be able to prove termination of this
particular example.\footnote{Similarly, 
with such a pre-processing the existing ``direct'' tools would also be able to
prove termination of the program in \rEx{ex_not_sm}.}
Integrating some form of partial evaluation into
{\sf AProVE} might be an interesting possibility for further improvement.
For all other examples, {\sf AProVE} can show termination whenever at least one of the
other tools succeeds. Moreover, there are several examples where {\sf
AProVE} succeeds whereas no other tool shows the termination. 
These include examples where the termination proof requires more
complex orders. 
For instance, termination of the example {\tt SGST06/hbal\_tree.pl}
can be proved using recursive
path orders with status and termination of  {\tt
talp/apt/mergesort\_ap.pl} is shown using matrix orders.\footnote{For
recursive path orders with status 
and matrix orders see \cite{Lesc83} resp. \cite{MatrixIJCAR}.}

Note that 52 examples in this collection are known to be non-terminating,
i.e., there are at most 244 terminating examples. In other words, there are
only at
most 12 terminating examples where {\sf AProVE} did not manage to prove
termination. With this performance, {\sf AProVE} won
the \emph{Termination Competition} with {\sf Polytool} being the second most
powerful tool.
The best tool for non-termination analysis of logic programs was {\sf NTI} \cite{Payet:Mesnard}.

However, from the experiments above one should not draw the conclusion that
the transformational approach is always better than the direct approach to
termination analysis of logic programs. There are several extensions (e.g., termination
inference \cite{Codish:Taboch,Mesnard:Bagnara},
non-termination analysis \cite{Payet:Mesnard}, handling numerical data
structures \cite{Serebrenik:DeSchreye:numerical:TPLP,SD05}) that can currently only be handled by direct techniques
and tools.

Regarding the use of term rewriting techniques for termination analysis of
logic programs,
it is interesting to note that
the
currently most powerful tool for direct termination analysis of logic
programs ({\sf Polytool}) implements the framework of \cite{Nguyen:DeSchreye,Nguyen:DeSchreye06} for
applying techniques from term rewriting (most notably polynomial interpretations)
to logic programs directly. This framework forms the basis for further
extensions to other TRS-termination techniques.
For example, it can be extended further by
adapting also basic results of the dependency pair method to the logic programming setting
\cite{LOPSTR07}. Preliminary investigations with a prototypical implementation
indicate that in this way, one can prove termination of several examples where
the transformational approach with {\sf AProVE} currently fails.

So  transformational and direct approaches both have their advantages and
the most powerful solution might be to combine  direct
tools like {\sf Polytool}
with a transformational prover like {\sf AProVE}
which is based on the contributions of this paper. But it is clear that
it is indeed beneficial to use termination techniques from
TRSs for logic programs, both for direct and for transformational approaches.

In addition to the experiments described above (which compare different
termination provers), we also performed experiments with several versions of
{\sf AProVE} in order to evaluate the different heuristics and algorithms for
the computation of argument filters from 
\rSec{sec:refining}.
The following table shows that indeed our improved type-based refinement
heuristic (tb$'$) 
from
\rSec{Type-Based Refinement Heuristic}
significantly outperforms the simple improved outermost (om$'$) and innermost (im)
heuristics
from
\rSec{Simple Refinement Heuristics}. In fact, all examples that could be proved terminating by
any of the simple heuristics can also be proved terminating by the type-based
heuristic.

\begin{center}
\begin{tabular}{l||c|c|c}
          & {\sf AProVE} tb$'$ & {\sf AProVE} om$'$ & {\sf AProVE} im\\
\hline
Successes &             232 &             218 &             195\\
Failures  &              57 &              76 &              98\\
Timeouts  &               7 &               2 &               3
\end{tabular}
\end{center}

So far, for all experiments we used \rAlg{alg:ImprovedRefinement} (from
\rSec{sec:afsmoding})
in order to
compute a refined argument filter from the initial one.
To evaluate the advantage of this improved algorithm over
\rAlg{alg:refinement} (from \rSec{Refinement
Algorithm for Argument Filters}), we performed experiments with the two
algorithms (again using the type-based refinement heuristic
tb$'$ from \rSec{Type-Based Refinement Heuristic}).
The following table shows that  \rAlg{alg:ImprovedRefinement}
is indeed significantly more powerful than \rAlg{alg:refinement}.

\begin{center}
\begin{tabular}{l||c|c}
          & {\sf AProVE} \rAlg{alg:ImprovedRefinement} & {\sf AProVE} \rAlg{alg:refinement}\\
\hline
Successes &             232 &             212\\
Failures  &              57 &              74\\
Timeouts  &               7 &              10
\end{tabular}
\end{center}

As mentioned in \rSec{Structure of the Paper}, 
preliminary versions of parts of this paper appeared in
\cite{LOPSTR06}. However, the table below clearly shows that the results
of \rSec{sec:refining} (which are new compared to \cite{LOPSTR06})
improve the power of termination analysis substantially. To this end, we
compare our new implementation that uses the improved type-based refinement
heuristic (tb$'$) and the improved refinement algorithm (\rAlg{alg:ImprovedRefinement})
from \rSec{sec:refining} with 
the version of \textsf{AProVE} from the \emph{Termination Competition} 2006
that only contains the
results of \cite{LOPSTR06}.
To find argument filters, it uses a simple ad-hoc heuristic which turns out to
be clearly disadvantageous to the new sophisticated techniques from \rSec{sec:refining}. 

\begin{center}
\begin{tabular}{l||c|@{\quad\;}c@{\quad\;}}
          & {\sf AProVE} tb$'$  & {\sf AProVE} \cite{LOPSTR06}\\
\hline
Successes &             232 &            208\\
Failures  &              57 &             69\\
Timeouts  &               7 &             19
\end{tabular}
\end{center}

To run {\sf AProVE}, for details on
our experiments, and to access our
collection of examples,
we refer  to
\texttt{http://aprove.informatik.rwth-aachen.de/eval/TOCL/}.

\subsection{Limitations}
\label{sec:limitations}

Our experiments also contain examples which demonstrate the limitations
of our approach. Of course, our implementation in {\sf AProVE} usually fails
if there are features outside of pure logic programming (e.g., built-in
predicates, negation as failure, meta programming, etc.). We consider the
handling of meta-logical features such as cuts and meta programming as
future work. We think that techniques from term rewriting are especially
well-suited to handle meta programming as term rewriting does not rely
on a distinction between predicate and function symbols.

In the following, we discuss
the limitations of the approach when applying it for pure logic programming.
In principle, there could be three points of failure:
\begin{enumerate}
\item The transformation of \rThm{Soundness of the Transformation} could fail,
      i.e., there could be a logic program which is terminating for the set of
      queries, but not all corresponding terms are terminating in the
      transformed TRS. We do not know of any such example. It is currently
open whether this step is in fact complete.
\item \label{item:incomplete}The approach via dependency pairs (\rThm{Proving Infinitary Termination})
      can fail to prove termination of the transformed TRS, although the TRS is
      terminating. In particular, this can happen because of the variable
      condition required for \rThm{Proving Infinitary Termination}. This is
      demonstrated by the following logic program $\P$:
      \[
      \begin{array}{lll}
      \Fp(X)       & \from & \Fq(\Ff(Y)), \Fp(Y).\\
      \Fp(\Fg(X))  & \from & \Fp(X).\\
      \Fq(\Fg(Y)).
      \end{array}
      \]
      The resulting TRS $\R_\P$ is
      \begin{align*}
      \urule{\Fp_{in}(X)               }{ {\sf u}_1(\Fq_{in}(\Ff(Y)),X)}\\
      \urule{{\sf u}_1(\Fq_{out}(\Ff(Y)),X) }{ {\sf u}_2(\Fp_{in}(Y),X,Y)}\\
      \urule{{\sf u}_2(\Fp_{out}(Y),X,Y)    }{ \Fp_{out}(X)}\\
      \urule{\Fp_{in}(\Fg(X))          }{ {\sf u}_3(\Fp_{in}(X),X)}\\
      \urule{{\sf u}_3(\Fp_{out}(X),X)      }{ \Fp_{out}(\Fg(X))}\\
      \urule{\Fq_{in}(\Fg(Y))          }{ \Fq_{out}(\Fg(Y))}
      \end{align*}
      and there are the following dependency pairs.
      \begin{align}
       \lrule{DPL1}{{\sf P}_{in}(X)               }{ {\sf Q}_{in}(\Ff(Y))}\\
      \lrule{DPL2}{{\sf P}_{in}(X)               }{ {\sf U}_1(\Fq_{in}(\Ff(Y)),X)}\\
      \lrule{DPL3}{{\sf U}_1(\Fq_{out}(\Ff(Y)),X) }{ {\sf P}_{in}(Y)}\\
      \lrule{DPL4}{{\sf U}_1(\Fq_{out}(\Ff(Y)),X) }{ {\sf U}_2(\Fp_{in}(Y),X,Y)}\\
      \lrule{DPL5}{{\sf P}_{in}(\Fg(X))          }{ {\sf P}_{in}(X)}\\
      \lrule{DPL6}{{\sf P}_{in}(\Fg(X))          }{ {\sf U}_3(\Fp_{in}(X),X)}
     \end{align}
      We want to prove termination of all queries $\Fp(t)$ where $t$ is finite
and ground (i.e., $m(\Fp,1) = \mIn$). Looking at the
      logic program $\P$, it is obvious that they are all terminating.
      However, there is no argument filter $\pi$ such that $\pi(\R_\P)$ and
      $\pi(DP(\R_\P))$ satisfy the variable condition and such that there is
      no infinite $(DP(\R_\P),\R_\P,\pi)$-chain.

      To see this, note that if $\pi({\sf P}_{in}) = \varnothing$ or $\pi(\Fg) = \varnothing$
      then we can build an infinite chain with the last dependency pair where
      we instantiate $X$ by the infinite term $\Fg(\Fg(\ldots))$. So, let
      $\pi({\sf P}_{in}) = \pi(\Fg) = \{1\}$.
      Due to the variable condition
      of the dependency pair (\ref{DPL3}) we know $\pi(\Ff) = \pi(\Fq_{out}) = \{1\}$ and
      $1 \in \pi({\sf U}_1)$. Hence, to satisfy the variable condition in 
      dependency pair (\ref{DPL2}) we must set $\pi(\Fq_{in}) = \varnothing$. But then the
      last rule of $\pi(\R_\P)$ does not satisfy the variable condition.

\item Finally
 it can happen that the resulting DP problem of
      \rThm{Proving Infinitary Termination} is finite, but that our
      implementation fails to prove it. The reason can be that one should
      apply other DP processors or DP processors with other parameters.
      After all, finiteness of DP problems is undecidable. This is shown
      by the following example where we are interested in all queries
      $\Ff(t_1,t_2)$ where $t_1$ and $t_2$ are ground terms:
      \[
      \begin{array}{lll}
      \Ff(X,Y)                     & \from & \Fg(\Fs(\Fs(\Fs(\Fs(\Fs(X))))),Y).\\
      \Ff(\Fs(X),Y)                & \from & \Ff(X,Y).\\
      \Fg(\Fs(\Fs(\Fs(\Fs(\Fs(\Fs(X)))))),Y) & \from & \Ff(X,Y).
      \end{array}
      \]

      Termination can (for example) be proved if one uses a polynomial order
      with coefficients from $\{0,1,2,3,4,5\}$.
But the current automation does
      not use such polynomials and
      thus, it fails when trying to prove
      termination of this example.
\end{enumerate}

While the DP method can also be used for non-termination proofs if one
considers ordinary rewriting, this is less obvious for infinitary constructor
rewriting. The reason is that the main termination criterion is ``complete''
for ordinary rewriting, but incomplete for
infinitary constructor
rewriting (cf.\ the counterexample (\ref{item:incomplete}) to the completeness of
\rThm{Proving
Infinitary Termination} above). Therefore, in order to also prove
\emph{non-termination} of logic programs, a combination of our method with a
loop-checker for logic programs would be fruitful. As mentioned before, a very powerful
non-termination tool for logic programs is {\sf NTI} \cite{Payet:Mesnard}.
Our collection of 296 examples contains 233 terminating examples (232 of these
can be successfully shown by {\sf AProVE}), 52
non-terminating examples, and 11 examples whose termination behavior is
unknown.
{\sf NTI} can prove non-termination of 42 of the 52 non-terminating
examples. Hence, a combination of {\sf AProVE} and {\sf NTI} would
successfully analyze the termination behaviour of
274 of the 296 examples.

\section{Conclusion}
\label{sec:conclusion}
In this paper, we developed a new transformation from
logic programs $\P$ to
TRSs $\R_\P$. To prove the termination of a class of queries for $\P$, it is
now sufficient to
analyze the termination behavior of $\R_\P$ on a corresponding class of terms
w.r.t.\ infinitary constructor
rewriting.  This class of terms is characterized by a so-called argument
filter and we showed how to generate such argument filters from the given class
of queries for $\P$.
Our approach is even
sound for logic programming without occur check. To prove termination of
infinitary rewriting automatically,
we showed how to
adapt the DP framework of \cite{AG00,LPAR04,JAR06} from ordinary term rewriting to
infinitary constructor
rewriting. Then the DP framework can be used for termination
proofs of $\R_\P$ and thus, for automated termination analysis of
$\P$. Since \emph{any}
termination technique for TRSs can be formulated \pagebreak as a DP processor \cite{LPAR04},
now any such technique can also be used for logic programs.

In addition to the results presented in \cite{LOPSTR06},
we showed that our new approach subsumes the classical transformational
approach to termination analysis of logic programs. We also
provided new heuristics and algorithms for refining the initial argument filter
that improve the power of our method (and hence, also of its implementation)
substantially.

Moreover, we implemented all contributions in our termination prover {\sf
AProVE} and performed extensive experiments which demonstrate that our results
are indeed applicable in practice. More precisely, due to our contributions,
{\sf AProVE} has become the currently most powerful automated termination prover for logic
programs.

\begin{acknowledgments}
We thank  Mike Codish,
Danny De Schreye, and Fred Mesnard for helpful comments and discussions
on the results presented here and  Roberto
Bagnara, Samir Genaim, and Manh Thang Nguyen for help with the
experiments.
\end{acknowledgments}

\bibliographystyle{acmtrans}

\begin{thebibliography}{}

\bibitem[\protect\citeauthoryear{Aguzzi and Modigliani}{Aguzzi and
  Modigliani}{1993}]{Aguzzi:Modigliani}
{\sc Aguzzi, G.} {\sc and} {\sc Modigliani, U.} 1993.
\newblock {Proving Termination of Logic Programs by Transforming them into
  Equivalent Term Rewriting Systems}.
\newblock In {\em FSTTCS~'93}. LNCS, vol. 761. 114--124.

\bibitem[\protect\citeauthoryear{Apt and Etalle}{Apt and
  Etalle}{1993}]{Apt:Etalle}
{\sc Apt, K.~R.} {\sc and} {\sc Etalle, S.} 1993.
\newblock {On the Unification Free \textsf{Prolog} Programs}.
\newblock In {\em MFCS~'93}. LNCS, vol. 711. 1--19.

\bibitem[\protect\citeauthoryear{Apt\noopsort{1}}{Apt\noopsort{1}}{1997}]{Apt:%
Book}
{\sc Apt\noopsort{1}, K.~R.} 1997.
\newblock {\em {From Logic Programming to \textsf{Prolog}}}.
\newblock Prentice Hall, London.

\bibitem[\protect\citeauthoryear{Arts and Zantema}{Arts and
  Zantema}{1995}]{Arts:Zantema:95}
{\sc Arts, T.} {\sc and} {\sc Zantema, H.} 1995.
\newblock {Termination of Logic Programs Using Semantic Unification}.
\newblock In {\em LOPSTR~'95}. LNCS, vol. 1048. 219--233.

\bibitem[\protect\citeauthoryear{Arts\noopsort{1} and Giesl}{Arts\noopsort{1}
  and Giesl}{2000}]{AG00}
{\sc Arts\noopsort{1}, T.} {\sc and} {\sc Giesl, J.} 2000.
\newblock {Termination of Term Rewriting Using Dependency Pairs}.
\newblock {\em Theoretical Computer Science\/}~{\em 236}, 133--178.

\bibitem[\protect\citeauthoryear{Baader and Nipkow}{Baader and
  Nipkow}{1998}]{BN98}
{\sc Baader, F.} {\sc and} {\sc Nipkow, T.} 1998.
\newblock {\em {Term Rewriting and All That}}.
\newblock Cambridge University Press.

\bibitem[\protect\citeauthoryear{Bossi, Cocco, and Fabris}{Bossi
  et~al\mbox{.}}{1992}]{Bossi92}
{\sc Bossi, A.}, {\sc Cocco, N.}, {\sc and} {\sc Fabris, M.} 1992.
\newblock {Typed Norms}.
\newblock In {\em ESOP~'92}. LNCS, vol. 582. 73--92.

\bibitem[\protect\citeauthoryear{Bruynooghe, Demoen, Boulanger, Denecker, and
  Mulkers}{Bruynooghe et~al\mbox{.}}{1996}]{DBLP:conf/sas/BruynoogheDBDM96}
{\sc Bruynooghe, M.}, {\sc Demoen, B.}, {\sc Boulanger, D.}, {\sc Denecker,
  M.}, {\sc and} {\sc Mulkers, A.} 1996.
\newblock {A Freeness and Sharing Analysis of Logic Programs Based on a
  Pre-Interpretation}.
\newblock In {\em SAS~'96}. LNCS, vol. 1145. 128--142.

\bibitem[\protect\citeauthoryear{Bruynooghe\noopsort{1}, Gallagher, and {Van
  Humbeeck}}{Bruynooghe\noopsort{1}
  et~al\mbox{.}}{2005}]{DBLP:conf/sas/BruynoogheGH05}
{\sc Bruynooghe\noopsort{1}, M.}, {\sc Gallagher, J.~P.}, {\sc and} {\sc {Van
  Humbeeck}, W.} 2005.
\newblock {Inference of Well-Typings for Logic Programs with Application to
  Termination Analysis}.
\newblock In {\em SAS~'05}. LNCS, vol. 3672. 35--51.

\bibitem[\protect\citeauthoryear{Bruynooghe\noopsort{2}, Codish, Gallagher,
  Genaim, and Vanhoof}{Bruynooghe\noopsort{2}
  et~al\mbox{.}}{2007}]{CodishTOPLAS}
{\sc Bruynooghe\noopsort{2}, M.}, {\sc Codish, M.}, {\sc Gallagher, J.~P.},
  {\sc Genaim, S.}, {\sc and} {\sc Vanhoof, W.} 2007.
\newblock {Termination Analysis of Logic Programs through Combination of
  Type-Based Norms}.
\newblock {\em ACM Transactions on Programming Languages and Systems\/}~{\em
  29,\/}~2, Article 10.

\bibitem[\protect\citeauthoryear{Charatonik and Podelski}{Charatonik and
  Podelski}{1998}]{DBLP:conf/sas/CharatonikP98}
{\sc Charatonik, W.} {\sc and} {\sc Podelski, A.} 1998.
\newblock {Directional Type Inference for Logic Programs}.
\newblock In {\em SAS~'98}. LNCS, vol. 1503. 278--294.

\bibitem[\protect\citeauthoryear{Chtourou and Rusinowitch}{Chtourou and
  Rusinowitch}{1993}]{Chtourou:Rusinowitch}
{\sc Chtourou, M.} {\sc and} {\sc Rusinowitch, M.} 1993.
\newblock {M\'{e}thode Transformationelle pour la Preuve de Terminaison des
  Programmes Logiques}.
\newblock Unpublished manuscript.

\bibitem[\protect\citeauthoryear{Codish and Taboch}{Codish and
  Taboch}{1999}]{Codish:Taboch}
{\sc Codish, M.} {\sc and} {\sc Taboch, C.} 1999.
\newblock {A Semantic Basis for Termination Analysis of Logic Programs}.
\newblock {\em Journal of Logic Programming\/}~{\em 41,\/}~1, 103--123.

\bibitem[\protect\citeauthoryear{Codish\noopsort{1}, Lagoon, and
  Stuckey}{Codish\noopsort{1} et~al\mbox{.}}{2005}]{Codish:Lagoon:Stuckey}
{\sc Codish\noopsort{1}, M.}, {\sc Lagoon, V.}, {\sc and} {\sc Stuckey, P.}
  2005.
\newblock {Testing for Termination with Monotonicity Constraints}.
\newblock In {\em ICLP~'05}. LNCS, vol. 3668. 326--340.

\bibitem[\protect\citeauthoryear{Codish\noopsort{2}, Lagoon, Schachte, and
  Stuckey}{Codish\noopsort{2}
  et~al\mbox{.}}{2006}]{Codish:Lagoon:Stuckey:ESOP06}
{\sc Codish\noopsort{2}, M.}, {\sc Lagoon, V.}, {\sc Schachte, P.}, {\sc and}
  {\sc Stuckey, P.} 2006.
\newblock {Size-Change Termination Analysis in $k$-Bits}.
\newblock In {\em ESOP~'06}. LNCS, vol. 3924. 230--245.

\bibitem[\protect\citeauthoryear{Codish\noopsort{3}}{Codish\noopsort{3}}{2007}%
]{Codish:examples:URL}
{\sc Codish\noopsort{3}, M.} 2007.
\newblock {Collection of Benchmarks}.
\newblock Available at \url{http://lvs.cs.bgu.ac.il/~mcodish/}
  \url{suexec/terminweb/bin/terminweb.cgi?command=examples}.

\bibitem[\protect\citeauthoryear{Colmerauer}{Colmerauer}{1982}]{Colmerauer82}
{\sc Colmerauer, A.} 1982.
\newblock {\textsf{Prolog} and Infinite Trees}.
\newblock In {\em Logic Programming}, {K.~L. Clark} {and} {S.~T\"{a}rnlund},
  Eds. Academic Press, Oxford.

\bibitem[\protect\citeauthoryear{Cortesi and Fil{\'e}}{Cortesi and
  Fil{\'e}}{1999}]{DBLP:journals/jlp/CortesiF99}
{\sc Cortesi, A.} {\sc and} {\sc Fil{\'e}, G.} 1999.
\newblock {Sharing is Optimal}.
\newblock {\em Journal of Logic Programming\/}~{\em 38,\/}~3, 371--386.

\bibitem[\protect\citeauthoryear{{De Schreye} and Decorte}{{De Schreye} and
  Decorte}{1994}]{DeSchreye:Decorte:NeverEndingStory}
{\sc {De Schreye}, D.} {\sc and} {\sc Decorte, S.} 1994.
\newblock {Termination of Logic Programs: The Never-Ending Story}.
\newblock {\em Journal of Logic Programming\/}~{\em 19-20}, 199--260.

\bibitem[\protect\citeauthoryear{{De Schreye} and Serebrenik}{{De Schreye} and
  Serebrenik}{2002}]{DeSchreye:Serebrenik:Kowalski}
{\sc {De Schreye}, D.} {\sc and} {\sc Serebrenik, A.} 2002.
\newblock {Acceptability with General Orderings}.
\newblock In {\em Computational Logic: Logic Programming and Beyond. Essays in
  Honour of Robert A. Kowalski, Part I}, {A.~C. Kakas} {and} {F.~Sadri}, Eds.
  LNCS, vol. 2407. 187--210.

\bibitem[\protect\citeauthoryear{Decorte, {De Schreye}, and Fabris}{Decorte
  et~al\mbox{.}}{1993}]{Decorteetal93}
{\sc Decorte, S.}, {\sc {De Schreye}, D.}, {\sc and} {\sc Fabris, M.} 1993.
\newblock {Automatic Inference of Norms: A Missing Link in Automatic
  Termination Analysis}.
\newblock In {\em ILPS~'93}. MIT Press, Boston, 420--436.

\bibitem[\protect\citeauthoryear{Dershowitz}{Dershowitz}{1987}]{D87}
{\sc Dershowitz, N.} 1987.
\newblock {Termination of Rewriting}.
\newblock {\em Journal of Symbolic Computation\/}~{\em 3,\/}~1--2, 69--116.

\bibitem[\protect\citeauthoryear{Endrullis, Waldmann, and Zantema}{Endrullis
  et~al\mbox{.}}{2006}]{MatrixIJCAR}
{\sc Endrullis, J.}, {\sc Waldmann, J.}, {\sc and} {\sc Zantema, H.} 2006.
\newblock {Matrix Interpretations for Proving Termination of Term Rewriting}.
\newblock In {\em IJCAR~'06}. LNAI, vol. 4130. 574--588.

\bibitem[\protect\citeauthoryear{Gallagher and Puebla}{Gallagher and
  Puebla}{2002}]{GallagherPuebla02}
{\sc Gallagher, J.~P.} {\sc and} {\sc Puebla, G.} 2002.
\newblock {Abstract Interpretation over Non-deterministic Finite Tree Automata
  for Set-Based Analysis of Logic Programs}.
\newblock In {\em PADL~'02}. LNCS, vol. 2257. 243--261.

\bibitem[\protect\citeauthoryear{Ganzinger and Waldmann}{Ganzinger and
  Waldmann}{1993}]{Ganzinger:Waldmann}
{\sc Ganzinger, H.} {\sc and} {\sc Waldmann, U.} 1993.
\newblock {Termination Proofs of Well-moded Logic Programs via Conditional
  Rewrite Systems}.
\newblock In {\em CTRS~'92}. LNCS, vol. 656. 430--437.

\bibitem[\protect\citeauthoryear{Geser, Hofbauer, and Waldmann}{Geser
  et~al\mbox{.}}{2004}]{GHW}
{\sc Geser, A.}, {\sc Hofbauer, D.}, {\sc and} {\sc Waldmann, J.} 2004.
\newblock {Match-Bounded String Rewriting Systems}.
\newblock {\em Applicable Algebra in Engineering, Communication and
  Computing\/}~{\em 15,\/}~3--4, 149--171.

\bibitem[\protect\citeauthoryear{Giesl, Thiemann, and Schneider-Kamp}{Giesl
  et~al\mbox{.}}{2005}]{LPAR04}
{\sc Giesl, J.}, {\sc Thiemann, R.}, {\sc and} {\sc Schneider-Kamp, P.} 2005.
\newblock {The Dependency Pair Framework: Combining Techniques for Automated
  Termination Proofs}.
\newblock In {\em LPAR~'04}. LNAI, vol. 3452. 301--331.

\bibitem[\protect\citeauthoryear{Giesl\noopsort{1}, Schneider-Kamp, and
  Thiemann}{Giesl\noopsort{1} et~al\mbox{.}}{2006}]{IJCAR06}
{\sc Giesl\noopsort{1}, J.}, {\sc Schneider-Kamp, P.}, {\sc and} {\sc Thiemann,
  R.} 2006.
\newblock {\textsf{AProVE 1.2:} Automatic Termination Proofs in the DP
  Framework}.
\newblock In {\em IJCAR~'06}. LNAI, vol. 4130. 281--286.

\bibitem[\protect\citeauthoryear{Giesl\noopsort{1}, Swiderski, Schneider-Kamp,
  and Thiemann}{Giesl\noopsort{1} et~al\mbox{.}}{2006}]{RTA06}
{\sc Giesl\noopsort{1}, J.}, {\sc Swiderski, S.}, {\sc Schneider-Kamp, P.},
  {\sc and} {\sc Thiemann, R.} 2006.
\newblock {Automated Termination Analysis for \textsf{Haskell}: From Term
  Rewriting to Programming Languages}.
\newblock In {\em RTA~'06}. LNCS, vol. 4098. 297--312.

\bibitem[\protect\citeauthoryear{Giesl\noopsort{1}, Thiemann, Schneider-Kamp,
  and Falke}{Giesl\noopsort{1} et~al\mbox{.}}{2006}]{JAR06}
{\sc Giesl\noopsort{1}, J.}, {\sc Thiemann, R.}, {\sc Schneider-Kamp, P.}, {\sc
  and} {\sc Falke, S.} 2006.
\newblock {Mechanizing and Improving Dependency Pairs}.
\newblock {\em Journal of Automated Reasoning\/}~{\em 37,\/}~3, 155--203.

\bibitem[\protect\citeauthoryear{Hirokawa and Middeldorp}{Hirokawa and
  Middeldorp}{2005}]{HM05}
{\sc Hirokawa, N.} {\sc and} {\sc Middeldorp, A.} 2005.
\newblock {Automating the Dependency Pair Method}.
\newblock {\em Information and Computation\/}~{\em 199,\/}~1,2, 172--199.

\bibitem[\protect\citeauthoryear{Huet}{Huet}{1976}]{Huet76}
{\sc Huet, G.} 1976.
\newblock {R\'esolution d'Equations dans les Langages d'Ordre 1, 2, \ldots,
  $\omega$}.
\newblock PhD thesis.

\bibitem[\protect\citeauthoryear{Janssens and Bruynooghe}{Janssens and
  Bruynooghe}{1992}]{Janssens:Bruynooghe}
{\sc Janssens, G.} {\sc and} {\sc Bruynooghe, M.} 1992.
\newblock {Deriving Descriptions of Possible Values of Program Variables by
  Means of Abstract Interpretation}.
\newblock {\em Journal of Logic Programming\/}~{\em 13,\/}~2--3, 205--258.

\bibitem[\protect\citeauthoryear{{Krishna Rao}, Kapur, and
  Shyamasundar}{{Krishna Rao} et~al\mbox{.}}{1998}]{Rao}
{\sc {Krishna Rao}, M. R.~K.}, {\sc Kapur, D.}, {\sc and} {\sc Shyamasundar,
  R.~K.} 1998.
\newblock {Transformational Methodology for Proving Termination of Logic
  Programs}.
\newblock {\em Journal of Logic Programming\/}~{\em 34,\/}~1, 1--41.

\bibitem[\protect\citeauthoryear{Lagoon and Stuckey}{Lagoon and
  Stuckey}{2002}]{DBLP:conf/ppdp/LagoonS02}
{\sc Lagoon, V.} {\sc and} {\sc Stuckey, P.~J.} 2002.
\newblock {Precise Pair-Sharing Analysis of Logic Programs}.
\newblock In {\em PPDP~'02}. {ACM} {P}ress, New York, 99--108.

\bibitem[\protect\citeauthoryear{Lagoon\noopsort{1}, Mesnard, and
  Stuckey}{Lagoon\noopsort{1} et~al\mbox{.}}{2003}]{LMS03}
{\sc Lagoon\noopsort{1}, V.}, {\sc Mesnard, F.}, {\sc and} {\sc Stuckey, P.~J.}
  2003.
\newblock {Termination Analysis with Types Is More Accurate}.
\newblock In {\em ICLP~'03}. LNCS, vol. 2916. 254--268.

\bibitem[\protect\citeauthoryear{Lescanne}{Lescanne}{1983}]{Lesc83}
{\sc Lescanne, P.} 1983.
\newblock {Computer Experiments with the {REVE} Term Rewriting System
  Generator}.
\newblock In {\em POPL~'83}. {ACM} {P}ress, New York, 99--108.

\bibitem[\protect\citeauthoryear{Leuschel and S{\o}rensen}{Leuschel and
  S{\o}rensen}{1996}]{Leuschel:Sorensen}
{\sc Leuschel, M.} {\sc and} {\sc S{\o}rensen, M.~H.} 1996.
\newblock {Redundant Argument Filtering of Logic Programs}.
\newblock In {\em LOPSTR~'96}. LNCS, vol. 1207. 83--103.

\bibitem[\protect\citeauthoryear{Lindenstrauss, Sagiv, and
  Serebrenik}{Lindenstrauss
  et~al\mbox{.}}{1997}]{Lindenstrauss:Sagiv:Serebrenik}
{\sc Lindenstrauss, N.}, {\sc Sagiv, Y.}, {\sc and} {\sc Serebrenik, A.} 1997.
\newblock {{\sf TermiLog\/}: A System for Checking Termination of Queries to
  Logic Programs}.
\newblock In {\em CAV~'97}. LNCS, vol. 1254. 444--447.

\bibitem[\protect\citeauthoryear{Lu}{Lu}{2000}]{Lu}
{\sc Lu, L.} 2000.
\newblock {A Precise Type Analysis of Logic Programs}.
\newblock In {\em PPDP~'00}. {ACM} {P}ress, New York, 214--225.

\bibitem[\protect\citeauthoryear{March\'e and Zantema}{March\'e and
  Zantema}{2007}]{Competition}
{\sc March\'e, C.} {\sc and} {\sc Zantema, H.} 2007.
\newblock {The Termination Competition}.
\newblock In {\em RTA~'07}. LNCS, vol. 4533. 303--313.
\newblock See also \url{http://www.lri.fr/~marche/termination-competition/}.

\bibitem[\protect\citeauthoryear{Marchiori}{Marchiori}{1994}]{Marchiori:ALP94}
{\sc Marchiori, M.} 1994.
\newblock {Logic Programs as Term Rewriting Systems}.
\newblock In {\em ALP~'94}. LNCS, vol. 850. 223--241.

\bibitem[\protect\citeauthoryear{Marchiori}{Marchiori}{1996}]{Marchiori:AMAST9%
6}
{\sc Marchiori, M.} 1996.
\newblock {Proving Existential Termination of Normal Logic Programs}.
\newblock In {\em AMAST~'96}. LNCS, vol. 1101. 375--390.

\bibitem[\protect\citeauthoryear{Martin, King, and Soper}{Martin
  et~al\mbox{.}}{1996}]{Martin96}
{\sc Martin, J.}, {\sc King, A.}, {\sc and} {\sc Soper, P.} 1996.
\newblock {Typed Norms for Typed Logic Programs}.
\newblock In {\em LOPSTR~'96}. LNCS, vol. 1207. 224--238.

\bibitem[\protect\citeauthoryear{Mesnard and Ruggieri}{Mesnard and
  Ruggieri}{2003}]{Mesnard:Ruggieri}
{\sc Mesnard, F.} {\sc and} {\sc Ruggieri, S.} 2003.
\newblock {On Proving Left Termination of Constraint Logic Programs}.
\newblock {\em ACM Transactions on Computational Logic\/}~{\em 4,\/}~2,
  207--259.

\bibitem[\protect\citeauthoryear{Mesnard\noopsort{1} and
  Bagnara}{Mesnard\noopsort{1} and Bagnara}{2005}]{Mesnard:Bagnara}
{\sc Mesnard\noopsort{1}, F.} {\sc and} {\sc Bagnara, R.} 2005.
\newblock {{\sf cTI}: A Constraint-Based Termination Inference Tool for
  \textsf{ISO-Prolog}}.
\newblock {\em Theory and Practice of Logic Programming\/}~{\em 5,\/}~1{\&}2,
  243--257.

\bibitem[\protect\citeauthoryear{Mesnard\noopsort{1} and
  Serebrenik}{Mesnard\noopsort{1} and Serebrenik}{2007}]{Mesnard:Serebrenik}
{\sc Mesnard\noopsort{1}, F.} {\sc and} {\sc Serebrenik, A.} 2007.
\newblock {Recurrence with Affine Level Mappings is P-Time Decidable for
  CLP(R)}.
\newblock {\em Theory and Practice of Logic Programming\/}~{\em 8,\/}~1,
  111--119.

\bibitem[\protect\citeauthoryear{Nguyen and {De Schreye}}{Nguyen and {De
  Schreye}}{2005}]{Nguyen:DeSchreye}
{\sc Nguyen, M.~T.} {\sc and} {\sc {De Schreye}, D.} 2005.
\newblock {Polynomial Interpretations as a Basis for Termination Analysis of
  Logic Programs}.
\newblock In {\em ICLP~'05}. LNCS, vol. 3668. 311--325.

\bibitem[\protect\citeauthoryear{Nguyen\noopsort{1}, Bruynooghe, {De Schreye},
  and Leuschel}{Nguyen\noopsort{1}
  et~al\mbox{.}}{2006}]{Nguyen:Bruynooghe:DeSchreye:Leuschel}
{\sc Nguyen\noopsort{1}, M.~T.}, {\sc Bruynooghe, M.}, {\sc {De Schreye}, D.},
  {\sc and} {\sc Leuschel, M.} 2006.
\newblock {Program Specialisation as a Pre-Processing Step for Termination
  Analysis}.
\newblock In {\em WST~'06}. 7--11.
\newblock Available at
  \url{http://www.easychair.org/FLoC-06/WST-preproceedings.pdf}.

\bibitem[\protect\citeauthoryear{Nguyen\noopsort{1} and {De
  Schreye}}{Nguyen\noopsort{1} and {De Schreye}}{2007}]{Nguyen:DeSchreye06}
{\sc Nguyen\noopsort{1}, M.~T.} {\sc and} {\sc {De Schreye}, D.} 2007.
\newblock {\textsf{Polytool}: Proving Termination Automatically Based on
  Polynomial Interpretations}.
\newblock In {\em LOPSTR~'06}. LNCS, vol. 4407. 210--218.

\bibitem[\protect\citeauthoryear{Nguyen\noopsort{1}, Giesl, Schneider-Kamp, and
  {De Schreye}}{Nguyen\noopsort{1} et~al\mbox{.}}{2008}]{LOPSTR07}
{\sc Nguyen\noopsort{1}, M.~T.}, {\sc Giesl, J.}, {\sc Schneider-Kamp, P.},
  {\sc and} {\sc {De Schreye}, D.} 2008.
\newblock {Termination Analysis of Logic Programs based on Dependency Graphs}.
\newblock In {\em LOPSTR~'07}. LNCS, vol. 4915. 8--22.

\bibitem[\protect\citeauthoryear{Ohlebusch, Claves, and March\'e}{Ohlebusch
  et~al\mbox{.}}{2000}]{TALP}
{\sc Ohlebusch, E.}, {\sc Claves, C.}, {\sc and} {\sc March\'e, C.} 2000.
\newblock {{\sf TALP}: A Tool for the Termination Analysis of Logic Programs}.
\newblock In {\em RTA~'00}. LNCS, vol. 1833. 270--273.

\bibitem[\protect\citeauthoryear{Ohlebusch\noopsort{1}}{Ohlebusch\noopsort{1}}%
{2001}]{Ohlebusch01}
{\sc Ohlebusch\noopsort{1}, E.} 2001.
\newblock {Termination of Logic Programs: {T}ransformational Methods
  Revisited}.
\newblock {\em Applicable Algebra in Engineering, Communication and
  Computing\/}~{\em 12,\/}~1--2, 73--116.

\bibitem[\protect\citeauthoryear{Payet and Mesnard}{Payet and
  Mesnard}{2006}]{Payet:Mesnard}
{\sc Payet, E.} {\sc and} {\sc Mesnard, F.} 2006.
\newblock {Nontermination Inference of Logic Programs}.
\newblock {\em ACM Transactions on Programming Languages and Systems\/}~{\em
  28,\/}~2, 256--289.

\bibitem[\protect\citeauthoryear{{\noopsort{Raamsdonk}}{van
  Raamsdonk}}{{\noopsort{Raamsdonk}}{van Raamsdonk}}{1997}]{Raamsdonk}
{\sc {\noopsort{Raamsdonk}}{van Raamsdonk}, F.} 1997.
\newblock {Translating Logic Programs into Conditional Rewriting Systems}.
\newblock In {\em ICLP~'97}. MIT Press, Boston, 168--182.

\bibitem[\protect\citeauthoryear{Schneider-Kamp, Giesl, Serebrenik, and
  Thiemann}{Schneider-Kamp et~al\mbox{.}}{2007}]{LOPSTR06}
{\sc Schneider-Kamp, P.}, {\sc Giesl, J.}, {\sc Serebrenik, A.}, {\sc and} {\sc
  Thiemann, R.} 2007.
\newblock {Automated Termination Analysis for Logic Programs by Term
  Rewriting}.
\newblock In {\em LOPSTR~'06}. LNCS, vol. 4407. 177--193.

\bibitem[\protect\citeauthoryear{Serebrenik and {De Schreye}}{Serebrenik and
  {De Schreye}}{2003}]{DBLP:conf/lopstr/SerebrenikS03}
{\sc Serebrenik, A.} {\sc and} {\sc {De Schreye}, D.} 2003.
\newblock {Proving Termination with Adornments}.
\newblock In {\em \mbox{LOPSTR~'03}}. LNCS, vol. 3018. 108--109.

\bibitem[\protect\citeauthoryear{Serebrenik and {De Schreye}}{Serebrenik and
  {De Schreye}}{2004}]{Serebrenik:DeSchreye:numerical:TPLP}
{\sc Serebrenik, A.} {\sc and} {\sc {De Schreye}, D.} 2004.
\newblock {Inference of Termination Conditions for Numerical Loops in
  \textsf{{P}rolog}}.
\newblock {\em Theory and Practice of Logic Programming\/}~{\em 4,\/}~5{\&}6,
  719--751.

\bibitem[\protect\citeauthoryear{Serebrenik and {De Schreye}}{Serebrenik and
  {De Schreye}}{2005a}]{SerebrenikS05}
{\sc Serebrenik, A.} {\sc and} {\sc {De Schreye}, D.} 2005a.
\newblock {On Termination of Meta-Programs}.
\newblock {\em Theory and Practice of Logic Programming\/}~{\em 5,\/}~3,
  355--390.

\bibitem[\protect\citeauthoryear{Serebrenik and {De Schreye}}{Serebrenik and
  {De Schreye}}{2005b}]{SD05}
{\sc Serebrenik, A.} {\sc and} {\sc {De Schreye}, D.} 2005b.
\newblock {Termination of Floating Point Computations}.
\newblock {\em Journal of Automated Reasoning\/}~{\em 34,\/}~2, 141--177.

\bibitem[\protect\citeauthoryear{Smaus}{Smaus}{2004}]{Smaus}
{\sc Smaus, J.-G.} 2004.
\newblock {Termination of Logic Programs Using Various Dynamic Selection
  Rules}.
\newblock In {\em ICLP~'04}. LNCS, vol. 3132. 43--57.

\bibitem[\protect\citeauthoryear{Tamary and Codish}{Tamary and
  Codish}{2004}]{Tamary:Codish}
{\sc Tamary, L.} {\sc and} {\sc Codish, M.} 2004.
\newblock {Abstract Partial Evaluation for Termination Analysis}.
\newblock In {\em WST~'04}. 47--50.
\newblock Available at
  \url{http://aib.informatik.rwth-aachen.de/2004/2004-07.pdf}.

\bibitem[\protect\citeauthoryear{??}{TPDB}{2007}]{TPDB}
TPDB 2007.
\newblock {The {T}ermination {P}roblem {D}ata {B}ase 4.0 ({M}ay 24, 2007)}.
\newblock Available at \url{http://www.lri.fr/~marche/tpdb/}.

\bibitem[\protect\citeauthoryear{Vaucheret and Bueno}{Vaucheret and
  Bueno}{2002}]{DBLP:conf/sas/VaucheretB02}
{\sc Vaucheret, C.} {\sc and} {\sc Bueno, F.} 2002.
\newblock {More Precise Yet Efficient Type Inference for Logic Programs}.
\newblock In {\em SAS~'02}. LNCS, vol. 2477. 102--116.

\bibitem[\protect\citeauthoryear{Walther}{Walther}{1994}]{Walther94}
{\sc Walther, C.} 1994.
\newblock {On Proving the Termination of Algorithms by Machine}.
\newblock {\em Artificial Intelligence\/}~{\em 71,\/}~1, 101--157.

\bibitem[\protect\citeauthoryear{Zantema}{Zantema}{1995}]{Z95}
{\sc Zantema, H.} 1995.
\newblock {Termination of Term Rewriting by Semantic Labelling}.
\newblock {\em Fundamenta Informaticae\/}~{\em 24,\/}~1--2, 89--105.

\end{thebibliography}

\providecommand{\noopsort}[1]{}

\begin{received}
  Received March 2008;
  accepted July 2008
\end{received}
\end{document}